\documentclass[12pt]{iopart}

\usepackage[english]{babel}
\usepackage{mathrsfs}
\expandafter\let\csname equation*\endcsname\relax
\expandafter\let\csname endequation*\endcsname\relax
\usepackage{amsmath}
\usepackage{amsthm}
\usepackage{amsfonts}
\usepackage{BOONDOX-cal}
\usepackage{graphicx}
\usepackage{braket}
\usepackage{subcaption}
\usepackage{multirow}
\usepackage[table]{xcolor}
\usepackage{bbm}
\usepackage{perpage}
\usepackage{booktabs}
\usepackage{multirow}
\usepackage{amssymb}
\usepackage{xfrac}
\usepackage[switch]{lineno}
\usepackage{cite}

\definecolor{nicepurple}{HTML}{97a0cf}
\definecolor{lightgreen}{HTML}{88BFB7} 
\definecolor{lightergreen}{HTML}{88BFB7} 
\definecolor{cgreen}{HTML}{88BFB7}
\definecolor{melon}{HTML}{FCBCB8}
\definecolor{puce}{HTML}{C08497}

\renewcommand{\newline}{\\\\\noindent}

\newcommand{\expect}[1]{{\langle #1\rangle}}

\newcommand{\one}{\mathbbm{1}}
\newcommand{\scpr}[2]{\langle#1\, \vert \, #2 \rangle}
 
\newcommand{\norm}[1]{\left\lVert #1 \right\rVert}

\newcommand{\newparagraph}{\\\\\noindent}
\newcommand{\neuralqx}{\textsc{neuraLQX} }
\newcommand{\netket}{\textsc{NetKet} }
\newcommand{\jax}{\textsc{JAX} }
\newcommand{\flax}{\textsc{Flax} }
\newcommand{\mpijax}{\textsc{mpi4jax} }

\newtheorem{prop}{Proposition}

\MakePerPage{footnote}

\begin{document}

\title[]{Finding and characterising physical states of Euclidean Abelianized loop quantum gravity using neural quantum states}

\author{Hanno Sahlmann $^1$, Waleed Sherif $^2$\footnote{Author to whom any correspondence should be addressed}}

\address{Institute for Quantum Gravity, Department of Physics, Friedrich-Alexander-Universit\"{a}t Erlangen-N\"{u}rnberg (FAU), Staudtstraße 7, 91058 Erlangen, Germany}
\ead{$^1$ hanno.sahlmann@fau.de, $^2$ waleed.sherif@fau.de}
\vspace{10pt}
\begin{indented}
\item[] \today
\end{indented}

\begin{abstract}
We study physical (near-kernel of constraints) states of 4-d Euclidean loop quantum gravity in Smolin's weak coupling limit on the complete graph $K_5$ using variational Monte Carlo with neural network quantum states. We investigate the Hamilton constraint $\hat{H}$ in the ordering proposed by Thiemann, as well as $\hat{H}^\dagger$ and $\hat{H}+\hat{H}^\dagger$. We find that the variational optimisation selects distinct solution families for $\hat{H}$ and $\hat{H}^\dagger$ across several considered cutoffs on the kinematical degrees of freedom. The solution family of $\hat{H}$ is flat on all minimal loops and has non-vanishing volume expectation values. Its edge-charge marginals delocalise with increasing cutoff, which indicates they are approximations to solutions that are non-normalisable in the kinematical inner product. The solution family for $\hat{H}^\dagger$ is normalisable, shows non-trivial charge correlations, lies in the kernel of volume and is not flat. $\hat{H}+\hat{H}^\dagger$ turns out to be much harder to solve and yields quasi-solutions combining features of both previous families. We characterise all solutions using chromaticity 1- and 2-point functions, minimal loop holonomies, geometric area and volume observables and show that the two families can be interpreted as, on the one hand, a family of states  close to the Ashtekar-Lewandowski vacuum and the Dittrich-Geiller vacuum with some numerical noise on the other hand. We also present some results that link solutions of the truncated theory to solutions of the continuum theory.
\end{abstract}

\section{Introduction}
\label{sec:introduction}

Loop quantum gravity (LQG) \cite{Rovelli:1997yv,Thiemann:2001gmi,Thiemann:2007pyv,Ashtekar:2004eh} is a non-peturbative approach to quantise general relativity in a canonical framework. Instead of metric variables, one starts from the classical theory rewritten in terms of a connection and its conjugate momentum \cite{Ashtekar:1986yd,BarberoG:1994eia}, thus embedding its phase space in that of SU(2) Yang-Mills theory. On the kinematical level, the quantum theory is well developed. Namely, one has a kinematical Hilbert space spanned by graph-based states (spin-network functions (SNFs) \cite{Baez:1996aima} and one can define geometric operators such as areas and volumes with discrete spectra \cite{Ashtekar:1996eg,Rovelli:1995discreteness,Ashtekar:1997volume}.
\newparagraph
The difficulty in canonical LQG is not the kinematics, but rather the dynamics \cite{Thiemann:1996aw,Thiemann:1996av,Thiemann:1997rv,Thiemann:1997ru,Reisenberger:1996pu,Lewandowski:2014hza,Varadarajan:2022dgg,Varadarajan:2021zrk,Ashtekar:2021shortreview}. In the canonical formulation, Einstein's equations become constraints, namely the Gauß constraint generating internal gauge transformation, the diffeomorphism constraint generating spatial diffeomorphisms and, most importantly, the Hamiltonian constraint encoding the dynamics. Physical states are those which lie in the kernel of all three constraints. In the full theory, substantial effort that has been dedicated to formulating the quantum constraints, but \emph{obtaining} and \emph{interpreting} solutions to the Hamiltonian constraint remains technically and conceptually a challenging open problem. 
\newparagraph
Recently, among attempts to approach this problem \cite{Guedes:2024zbu,Guedes:2024duc,Assanioussi:2017tql,Kisielowski:2022wvk,Makinen:2026rof}, a complementary direction has started to look practicable whereby one treats a truncated, fixed-graph kinematical Hilbert space as a large, but finite, many-body system \cite{Sahlmann:2024pba}. This in turn allows one to utilise modern variational numerics to address the constraint problem. In earlier work on a $\mathrm{U}(1)$ BF-theory quantised with LQG methods, a two-fold truncation on the kinematical Hilbert space was introduced: First, by working on a fixed graph and its subgraphs, and second, by restricting the allowed representation labels (charges) to a finite set, effectively working with a $\mathcal{U}_q(1)$ theory at a root of unity. Even in this radically truncated setting, the Hilbert space grows exponentially with the graph size and cutoff, quickly leaving the regime of explicit exact numerical methods. Neural network quantum states (NQS) \cite{Carleo:2016svm} combined with variational Monte Carlo (VMC) provide an efficient compression whereby the network could approximate constraint-minimising states with high accuracy when benchmarked against exact diagonalisation on small instances, while using a number of parameters vanishingly small compared to the Hilbert space dimension.
\newparagraph
The latter point becomes even more pronounced on the follow-up work \cite{Sahlmann:2024kat} where the 3-dimensional Euclidean model of LQG in Smolin's weak coupling limit \cite{Smolin:1992wj} was considered. In such a case, the relevant gauge structure also becomes Abelian, and one obtains an effective $\mathrm{U}(1)^3$ BF-type description. Unlike previous work, this model demonstrated several important advancements in this program, the first of which is the ability to construct Thiemann-type constraints computationally, and obtain solutions for them. Second, the dimension of the Hilbert spaces in such work lied beyond any approach which does not efficiently compress the degrees of freedom of the problem at hand, demonstrating for the first time the capacity of NQS to nevertheless be robust in such LQG settings. Lastly, with the help of LQG specific geometric observables, it is also possible to not only obtain solutions, but also characterise them.
\newparagraph
The present work continues this numerical programme in a 4-dimensional Euclidean LQG setting, once again in Smolin's weak coupling limit which leads to, once again, an Abelian $\mathrm{U}(1)^3$ theory but now with genuine 4-d kinematics and dynamics on a non-trivial boundary graph. In this work, we consider the complete $K_5$ graph, the boundary graph of a 4-simplex, which is the canonical minimal graph for 4-d simplicial boundary data and a natural testbed for the 4-d dynamics on a fixed graph. Its symmetry helps to suppress purely combinatorial artefacts, while its non-planarity forces one to address the issue of embedding concretely. Lastly, it is a graph composed of four 4-valent vertices, allowing for non-trivial action of the Thiemann Hamiltonian on all vertices. As in earlier studies \cite{Sahlmann:2024pba,Sahlmann:2024kat}, we fix a graph and impose a charge cutoff on representation labels and work in the same $\mathcal{U}_q(1)^3$ limit to work in a truncated finite-dimensional Hilbert space. Unlike previous work, however, we show that using such numerical methods, one could in principle obtain solutions at arbitrarily high cutoff.
\newparagraph
Within this setting, we consider three orderings for the vertex Hamilton constraints: the ordering $H_v$ proposed by Thiemann \cite{Thiemann:1996aw,Thiemann:1996av}, the opposite ordering $H_v^\dagger$ and a symmetric ordering $(H_v+H_v^\dagger)/2$. It is an interesting open question, if different orderings of the constraint lead to similar spaces of solutions, and if so, using a variational ansatz, which solutions are easier to find for the various orderings. An important result of the present paper is that scalable numerics can make this an empirically tractable question. 
\newparagraph
Concretely, for each choice of ordering, we use a graph-adapted NQS architecture and VMC optimisation on the gauge invariant subspace to obtain variational near-kernel state across a range of cutoffs. We then characterise these states:  We analyse how their probability mass redistributes under cutoff refinement as a practical proxy for normalisability behaviour in the kinematical inner product. We evaluate 
correlation functions to better understand the structure of the solutions. We calculate expectation values of observables sensitive to spatial geometry and extrinsic curvature, namely area and volume operators and minimal loop holonomies. Altogether we obtain a concrete, multi-angle picture of the geometrical structure and physical content of the solutions that each ordering results in.
\newparagraph
The paper is organised as follows:
\begin{enumerate}
    \item Section \ref{sec:theoreticalframework} discusses the theoretical model considered. We outline the structure of the Thiemann Hamiltonian as well as the geometric observables to be used in this work in this setting. We also present results that link the kernel of the quadratic constraints considered in this work to solutions to the continuum theory in the sense of \cite{Thiemann:1996aw}. 

    \item In Section \ref{sec:computationalframework}, we present the computational framework where we discuss the chosen graph, the truncated gauge invariant Hilbert space as well as the NQS architecture used in this work.

    \item Section \ref{sec:results} presents the results obtained in this work where we demonstrate that solutions for both the quadratic constraints can be obtained successfully in Section \ref{subsec:solsfamily}. Throughout Section \ref{subsec:characterisation}, we provide a host of characterisations of the obtained solutions where we investigate their normalisability, their flatness, long-range correlations, relation to known kinematical vacua in LQG and their geometric structure. In Section \ref{subsec:symmordsols}, we demonstrate a systematic manner in which one can interpolate between different solution classes in this variational setting.

    \item We close this work with Section \ref{sec:discussionoutlook} whereby we discuss the limitations encountered in this work, what the results tell us about factor ordering from a practitioner's point of view, and how the same workflow can be pushed towards now the full $\mathrm{SU}(2)$ setting.
\end{enumerate}

\section{Theoretical framework}
\label{sec:theoreticalframework}
Canonical loop quantum gravity is based on a phase space with Ashtekar-Barbero variables $(A,E)$ \cite{Ashtekar:1986yd,BarberoG:1994eia,Ashtekar:2004eh} as coordinates. The dynamics of general relativity is encoded in three sets of constraints: Gau{\ss}, diffeomorphism and Hamilton constraint. Here, $A$ is a connection (i.e. gauge) field with structure group SU(2). In the present work, however, we will later pass to a simplified, Abelianised theory with structure group $\mathrm{U}(1)^3$ (see Section \ref{subsec:weakcouplinglimit} below). The kinematical construction that we review first extends mutatis mutandis to that Abelianised setting and therefore provides the natural starting point for the fixed-graph model studied here.
\newparagraph
The kinematical Hilbert space of LQG can be written as \cite{Ashtekar:1994mh} 
\begin{equation}
    \mathcal{H}_\text{AL} = L^2(\overline{\mathcal{A}},\text{d}\mu_\text{AL})
\end{equation}
where $\overline{\mathcal{A}}$ is a space of certain generalized connection fields $A$. For an oriented graph $\gamma$ embedded in $\Sigma$, Cyl$_\gamma$ denotes the space of functionals in $\mathcal{H}_\text{AL}$ that depends on the connection as a continuous function on the holonomies along the edges of $\gamma$, and Cyl their union for all $\gamma$. For details see for example \cite{Ashtekar:1994mh,Ashtekar:1994wa,Thiemann:2001gmi,Ashtekar:2004eh}. 
We define 
\begin{equation}
    \mathcal{H}_\gamma := \overline{\text{Cyl}_\gamma}^{|| \cdot ||}. 
\end{equation}
These spaces are in general not orthogonal to each other for different $\gamma$. However, there is an orthogonal decomposition \cite{Ashtekar:1996eg}
\begin{equation}
    \mathcal{H}_\text{AL} = \bigoplus_{\gamma} \mathcal{H}'_\gamma. 
\end{equation}
The relation between $\mathcal{H}_\gamma$ and this decomposition is 
\begin{equation}
    \mathcal{H}_\gamma = \bigoplus_{\gamma'\subseteq \gamma} \mathcal{H}'_{\gamma'} 
\end{equation}
where here and in the following, the notation $\gamma'\subseteq \gamma$ denotes that $\gamma'$ is a subgraphs of $\gamma$, including the possibility that $\gamma' = \gamma$.
\newparagraph
Gau{\ss}, diffeomorphism and Hamilton constraints lead to conditions on the states in $\mathcal{H}_\text{AL}$. Solutions to the Gau{\ss} constraint are the states invariant under the action of SU(2) gauge transformations on $\overline{\mathcal{A}}$. We will work directly work with solutions of the Gau{\ss} constraint for the rest of the article. The diffeomorphism constraints impose spatial diffeomorphism invariance. We denote with $\mathcal{H}_\text{diff}$ the space of its solutions, which is part of the algebraic dual of $\mathcal{H}_\text{AL}$ 
\newparagraph
We consider a family of Hamilton constraint operators 
\begin{equation}
    N \longmapsto \hat{H}(N)
\end{equation}
defined on a common dense domain in Cyl $\subset \mathcal{H}_\text{AL}$. We will specify them more concretely below, but we want to first state some general results. We assume that  the $\hat{H}(N)$ are graph-preserving in the following sense:\footnote{To be precise, in the following equation we should let $\hat{H}(N)$ act on $\mathcal{H}_\gamma \cap \text{dom} \hat{H}(N)$. Here and in the following, we will not explicitly take care of the domain of the operator, however. We do not want to clutter the equations, but we do not foresee any obstacle to take the domains into account.}  
\begin{equation}
\label{eq:graph_preserving}
    \hat{H}(N)\, \mathcal{H}_\gamma \subseteq \mathcal{H}_\gamma. 
\end{equation}
We will also assume that 
\begin{equation}
    \hat{H}_v := \hat{H}(\delta_v)
\end{equation}
is a well defined operator, where $\delta_v(p)$ is the Kronecker delta function
\begin{equation}
    \delta_v(p) =\begin{cases}
    1 & \text{ if } p=v\\ 0 & \text{ otherwise}
    \end{cases}. 
\end{equation}
Following \cite{Thiemann:1996aw}, solutions $|\Psi) \in \mathcal{H}_\text{diff}$ to a family of Hamilton constraints $H(N)$ by definition satisfy
\begin{equation}
\label{eq:solution}
    \left (\Psi \right |  \hat{H}(N) f \rangle = 0 \quad,\quad\text{ for all } N: \Sigma \mapsto \mathbb{R}, f\in \text{Cyl}.  
\end{equation}
From this set of equations, we will consider only a subset of conditions, referring to a chosen graph $\gamma_0$, and take a form that is amenable to a minimization approach. We consider 
\begin{equation}
\label{eq:necessary}
    \left (\Psi \right |  \hat{H}(N) f_{\gamma_0} \rangle = 0 \quad,\quad\text{ for all } N: \Sigma \mapsto \mathbb{R}, f_{\gamma_0}\in \mathcal{H}_{\gamma_0}.
\end{equation}
These equations are obviously implied by \eqref{eq:solution}, hence they are necessary, but not necessarily sufficient. Finally, we define for any family $\{\hat{O}_v| v\in V(\gamma)\}$ of operators the non-negative symmetric operator 
\begin{equation}
    \hat{\mathcal{Q}}_{\hat{O}} := \sum_{v \in \gamma_0} \hat{O}_v \hat{O}_v^\dagger. 
\end{equation}
In this situation, we have the following  
\begin{prop}
\label{pr:main_proposition1}
For a set of constraint operators $\hat{H}(N)$ that are graph-preserving in  the sense of \eqref{eq:graph_preserving}, the set of conditions \eqref{eq:necessary} is equivalent to
\begin{equation}
\label{eq:minimal_expectation}
\langle \hat{\mathcal{Q}}_{\hat{H}\rvert_{\mathcal{H}_{\gamma_0}}} \rangle_{\widetilde{\Psi}} = 0, 
\end{equation}
where 
\begin{equation}
    \hat{H}_v\rvert_{\mathcal{H}_{\gamma_0}}:= P_{\gamma_0} \hat{H}_v P_{\gamma_0}
\end{equation}
is the restriction of $\hat{H}_v$ to  $\mathcal{H}_{\gamma_0}$ and $\widetilde{\Psi}\in \mathcal{H}_{\gamma_0}$ is a certain state in the kinematical Hilbert space constructed from  $|\Psi)$. 
\end{prop}
\noindent The proof will be given in appendix \ref{app:proof_minimization_problem}. What we take from this is that for each solution of the constraints $\hat{H}(N)$ in the sense of \cite{Thiemann:1996aw}, there must exist $\widetilde{\Psi}\in \mathcal{H}_{\gamma_0}$ in which the non-negative operator $\hat{\mathcal{Q}}_{\hat{H}}$ has expectation value 0. By minimizing the expectation value of $\hat{\mathcal{Q}}_{\hat{H}}$ numerically, we can generate states approximating such a $\widetilde{\Psi}$. Note that $\hat{\mathcal{Q}}_{\hat{H}}$ can be thought of as playing the role of a master constraint for the constraints $\hat{H}(N)$ \cite{Thiemann:2006phoenix}. However, $\hat{\mathcal{Q}}_{\hat{H}}$ is merely a simple implementation of the general idea, suited for the present application. For information on the general master constraint program see \cite{Thiemann:2006phoenix,Dittrich:2006master}. 
\newparagraph
We note that it is not trivial to invert the construction of $\widetilde{\Psi}$ from $\lvert\Psi)$. Specifically, it is not clear a priori that a solution $\widetilde{\Psi}$ of \eqref{eq:minimal_expectation} always comes from a solution $\lvert\Psi)$ of \eqref{eq:solution}. In the following, we will identify situations in which this is the case, however. 
\newparagraph
We consider a set of constraints $\hat{H}(N)$ as above. We assume in particular that $\hat{H}$ is graph preserving. We call $\hat{H}$ \emph{covariant} under the graph symmetries GS$_{\gamma_0}$ of $\gamma_0$, if 
\begin{equation}
    U_g \hat{H}_v\rvert_{\mathcal{H}_{\gamma_0}} U_g^{-1} = \hat{H}_{g(v)}\rvert_{\mathcal{H}_{\gamma_0}},\quad\text{ for all } g\in \text{GS}_{\gamma_0}
\end{equation}
Note that is this a bit of abuse of notation: elements in $\text{GS}_{\gamma_0}$ are equivalence classes of diffeomorphisms \cite{Ashtekar:1995diffeo,Ashtekar:2004eh}, so a priori it is not clear how they would act on $\mathcal{H}_{\text{kin}}$. For the adjoint action on $\hat{H} P_{\gamma_0}$ with $\hat{H}$ graph preserving, however, one can just use any representative of the equivalent class. The result does not depend on which. 
\newparagraph
Note that for any subgraph $\alpha\subseteq\gamma_0$ of $\gamma_0$ there is a natural group homomorphism
\begin{equation}
\label{eq:gs_map}
    \text{GS}_{\gamma_0} \longrightarrow \text{GS}_{\alpha}.
\end{equation}
The situation is especially simple if this homomorphism is onto, i.e. that any symmetry in $\text{GS}_\alpha$ can be realized as a symmetry of the whole graph $\gamma_0$. If and when this is the case depends on the nature of the graph and the category of manifolds and diffeomorphisms that is considered. 
\newparagraph
Note further that a subgraph $\alpha\subset\gamma_0$ of $\gamma_0$ may be embeddable into $\gamma_0$ in more than one way, in the sense that there are $[\varphi], [\varphi'], \ldots \in \sfrac{\text{Diff}}{\text{Diff}_\alpha}$ that map $\alpha$ onto a diffeomorphic subgraphs $\alpha', \alpha'',\ldots$ of $\gamma_0$. The situation is especially simple if $\varphi, \varphi',\ldots$ are equivalent to graph symmetries of $\gamma_0$, 
\begin{equation}
    \label{eq:gs_inclusion}
    \varphi \in  g, \varphi' \in g', \ldots, \text{ with }  g,g'\ldots \in \text{GS}_{\gamma_0}.
\end{equation}
If and when this is the case again depends on the nature of the graph and the category of manifolds and diffeomorphisms that is considered. 
In this situation, we denote with $n(\alpha, \gamma_0)$ the number of elements $[\varphi], [\varphi'], \ldots \in  \sfrac{\text{Diff}}{\text{Diff}_\alpha}$ inducing different embeddings of $\alpha$ in $\gamma_0$. 
\newparagraph
With these provisions, we have the following
\begin{prop}
\label{pr:main_proposition2}
Let
\begin{itemize}
    \item $\gamma_0$ be a graph for which \eqref{eq:gs_map} is onto, and \eqref{eq:gs_inclusion} holds in all applicable cases, 
    \item $\hat{H}(N)$ be a family of Hamilton constraints that is covariant under $\text{GS}_{\gamma_0}$, 
    \item $\widetilde{\Psi}$ be a state in $\mathcal{H}_{\gamma_0}$ such that 
    \begin{equation}
      \langle \hat{\mathcal{Q}}_{\hat{H}} \rangle_{\widetilde{\Psi}} = 0, 
    \end{equation}
\end{itemize}
Then the state 
\begin{equation}
    \lvert\Psi) = \sum_{\alpha\subseteq\gamma_0}\, \frac{1}{n(\alpha, \gamma_0)} \, \vert\eta (P'_\alpha\widetilde{\Psi}))\qquad \in \mathcal{H}_{\text{diff}}
\end{equation}
fulfils the necessary condition \eqref{eq:necessary}. Here $\eta: \mathcal{H}_\text{AL} \rightarrow \mathcal{H}_{\text{diff}}$ denotes the rigging map, and $P'_\alpha$ is the projector onto $\mathcal{H}'_{\alpha}$. 
\end{prop}
\noindent The proof will again be given in appendix \ref{app:proof_minimization_problem}.

\subsection{Weak coupling limit and charge cutoff}
\label{subsec:weakcouplinglimit}
So far we have discussed the full $\mathrm{SU}(2)$ theory, albeit without making explicit reference to the specific gauge group at hand. The model studied in the present work is, however, a fixed-graph weak coupling reduction, in the sense of Smolin's weak coupling limit \cite{Smolin:1992wj}, in which the non-Abelian structure group is replaced by $\mathrm{U}(1)^3$ \cite{Smolin:1992wj,Bakhoda:2020fiy,Bakhoda:2022rut}. In such a case, one eliminates non-Abelian effects in the action and thus in the full holonomy-flux algebra, yielding a a technically simpler framework in which the action of the Hamiltonian constraint can be investigated. In the weak coupling limit, the regularised Hamiltonian constraint retains the local, graph-dependent action, while the non-Abelian recoupling action with a simpler Abelian structure. 
\newparagraph
In the weak coupling limit, the kinematical states are no longer spin network functions labelled by irreducible $\mathrm{SU}(2)$ representations and intertwiners, but \emph{charge network functions} whose edges carry triples of $\mathrm{U}(1)$ charges.
\newparagraph
More concretely, each oriented edge $e$ now carries a charge vector
\begin{equation}
    \vec{m}_e = (m_e^{(1)}, m_e^{(2)}, m_e^{(3)}) \in \mathbb{Z}^3.
\end{equation}
Holonomies act by shifting these integer labels, one copy at a time. The $E$-dependent geometric operators become diagonal in this charge network basis. 
A simple example is provided by the area operator associated with a 2-face $S_e$ dual to an edge $e$. Suppressing overall convention dependent prefactors, its action reads \cite{Long:2022wcLQGArea}
\begin{equation}
    \hat{A}({S_e}) \ket{\{\vec{m}\}} = \sqrt{\sum_{i=1}^3 \left(m_e^{(i)}\right)^2}\,\ket{\{\vec{m}\}},
\end{equation}
where $\ket{\{\vec{m}\}}$ denotes a basis element in the charge network basis (see Section \ref{sec:computationalframework} below), and $S_e$ is a surface transversal to $e$. Thus, in the $\mathrm{U}(1)^3$ model, the area eigenvalue of the face dual to $e$ depends only on the norm of the charge vector carried by that edge. 
\newparagraph
Similarly, at a vertex $v$, the volume entering the regularised constraint is determined by the oriented triples of incident edges and takes the form \cite{Bakhoda:2025U1cubedVolume}
\begin{equation}
    \hat{V}_v \ket{\{\vec{m}\}} = \sqrt{\left|\sum_{(e_i, e_j, e_k) \text{ at } v} \epsilon_v(e_i, e_j, e_k) \, \vec{m}_{e_i} \cdot (\vec{m}_{e_j} \times \vec{m}_{e_k})\right|} \ket{\{\vec{m}\}},
\end{equation}
We note that here, the sign factor $\epsilon_v(e_i, e_j, e_k)$ is influenced by a choice of embedding of the considered graph $\gamma$ into the spatial submanifold. More precisely, if one denotes by $\hat{t}_e(v)$ the unit tangent of the embedded edge $\iota(e)$ at the vertex $v$, oriented away from $v$, then
\begin{equation}
    \epsilon_v(e_i, e_j, e_k) = \mathrm{sgn} \,\mathrm{det}[\hat{t}_{e_i}(v), \hat{t}_{e_j}(v), \hat{t}_{e_k}(v)]
\end{equation}
with $\epsilon_v(e_i, e_j, e_k) = 0$ whenever the three tangents are linearly dependent. It follows in particular that $\epsilon_v$ is totally antisymmetric in its edge arguments and vanishes whenever the local embedding is degenerate.
\newparagraph
For a fixed abstract graph, the operator therefore depends on a choice of embedding. 
\newparagraph
For our work, it seems important that the chosen embedding will be relatively generic, to avoid effects from very special configurations. In particular, at each vertex under consideration, the tangent directions of incident edges should be such that no triple which contributes to the constraint is coplanar. In the present work, a generic embedding is obtained by assigning to each vertex $v \in V(\gamma)$ an independent random point with components that are Gaussian distributed,
\begin{equation}
    \iota(v) = X_v \in \mathbb{R}^3, \quad,\quad X_v \sim \mathcal{N}(\mu, \sigma^2),
\end{equation}
and taking edges to be the corresponding embedded segments connecting their endpoints. 
\newparagraph
There is a practical issue left for working with the model outlined so far on the computer: arbitrary large charge components imply an infinite dimensional Hilbert space. A cutoff 
\begin{equation}
    |m_e^{(i)}| \leq m_\text{max}
\end{equation}
must be introduced to render the Hilbert space finite dimensional. Since holonomies shift the charges, their action must be modified at the cutoff. A mathematically elegant and minimal way to do so is to go over to a discrete structure group $\mathbb{Z}_{2m_\text{max}+1}^3$ from the continuous U(1)$^3$ \cite{Sahlmann:2024pba}. This shift can be understood as going over to a quantum group. Charge labels then become
\begin{equation}
    m_e^{(i)} \in \{-m_\text{max},-m_\text{max}+1, \ldots, m_\text{max}\} 
\end{equation}
and holonomies act by cyclic addition \cite{Sahlmann:2024pba}. The remain unitary operators and commute among themselves. Formulas for operators depending on $E$ are completely unchanged. 

\subsection{Hamiltonian constraint on the fixed graph}
\label{subsec:hamconstraintpresentwork}

The dynamics considered in the present work is defined by a graph-preserving regularisation of the Hamilton constraint on a fixed graph $\gamma$. More precisely, the vertex Hamiltonian $\hat{H}_v$ is constructed from holonomies along edges adjacent to $v$, the volume operator $\hat{V}_v$ and holonomies around loops based at $v$. Its action is given by \cite{Thiemann:1996aw}
\begin{equation}
    \hat{H}_v = c_v \sum_{(e_i, e_j, e_k) \text{ at } v} \varepsilon^{ijk} \; \mathrm{tr}\left[ \left(\hat{h}_{\alpha_{ij}} - \hat{h}_{\alpha_{ij}}^\dagger\right) \hat{h}_{e_k}^\dagger [\hat{h}_{e_k}, \hat{V}_v] \right].
\end{equation}
In this expression, $(e_i, e_j, e_k)$ denotes an ordered triple of edges meeting at $v$, while $\alpha_{ij}$ denotes the loop based at $v$ which is outgoing along $e_i$ and incoming along $e_j$. The $\varepsilon^{ijk}$ denotes the standard Levi-Civita symbol and the prefactor $c_v$ is, in general, valence dependent. Since the graph considered here is symmetric, the same value is obtained for all vertices and thus, this factor may be dropped without loss of generality.
\newparagraph
The Hamiltonian is made graph preserving by employing \emph{minimal loop\footnote{Here, we mean a loop $\alpha$ that belongs to the set of minimal length cycles of the cycle basis of the cycle space of the considered graph.} holonomies} $\hat{h}_{\alpha_{ij}}$. To make this explicit, consider the standard Thiemann regularisation \cite{Thiemann:1996aw}, where the curvature at a vertex $v$ is approximated by the holonomy around a small auxiliary loop
\begin{equation}
    \alpha_{ij} (\Delta) = s_i(\Delta) \circ a_{ij} (\Delta) \circ s_j(\Delta)^{-1},
\end{equation}
where $\Delta$ is the regulating tetrahedron attached to $v$, $s_i(\Delta)$ and $s_j(\Delta)$ are short segments of edges at $v$ and $a_{ij}(\Delta)$ is an additional arc joining their endpoints. Since this arc is not, in general, already contained in the underlying graph $\gamma$, the corresponding holonomy operator maps a state over $\gamma$ to a state over the enlarged graph $\gamma \cup \alpha_{ij}(\Delta)$. The action is therefore genuinely graph changing. It introduces new segments, creates new vertices at the endpoints of the auxiliary arc, and reroutes the original edges through the newly attached loop. A second action of the Hamiltonian would then, a priori, also act at these newly created vertices. In Thiemann's construction, this potential obstruction is controlled by the choice of routing of the auxiliary arcs. Namely, the new vertices are arranged so that all incident tangents lie in a common plane, equivalently such that only two independent tangent directions are present there, causing the volume operator to therefore vanish at such vertices and consequently, the non-trivial action of the Hamiltonian remains effectively supported at the original vertices of the graph.
\newparagraph
In the graph preserving setting adopted here, one replaces the auxiliary arc $a_{ij}$ of Thiemann's construction by a path intrinsic to the fixed graph $\gamma$, chosen such that, together with the ordered pair of edges $(e_i, e_j)$ at $v$, it defines a minimal loop $\alpha_{ij}$ based at $v$. In the case of a complete $K_5$ graph, the simplest example is provided by a triangular cycle of length three.
\newparagraph
In this graph preserving, weak coupling $\mathrm{U}(1)^3$ theory, the action of such a constraint is relatively straightforward. The trace decomposes into three Abelian copies here labelled by $n = 1, 2, 3$. A minimal loop holonomy acts by shifting the charge labels along the edges belonging to the chosen loop, while the commutator with the segment holonomy produces a finite difference of the vertex volume. Let $\delta \alpha_{v \ell}^{(n)}$ denote the modular charge shift associated with a loop $\ell \in L_v(e_i, e_j)$ in the $n$\textsuperscript{th} gauge copy, where $L_v(e_i, e_j)$ is the set of minimal loops of the graph that pass through $v$, leaving $v$ along $e_1$ and returning through $e_2$. Further, let $\delta^{(n)}_{e_k}$ denote a modular unit shift of the $n$\textsuperscript{th} charge on the segment edge $e_k$. The Hamiltonian can be written in the charge network basis in the form
\begin{align}
\label{eq:hv_wcl_discrete}
    \hat{H}_v \ket{\{\vec{m}\}} = \frac{1}{T_v} \sum_{(e_i, e_j, e_k) \text{ at } v} & \frac{\varepsilon^{ijk}}{|L_v(e_i, e_j)|} \sum_{\ell \in L_v(e_i, e_j)}  \rho(\ell; \, e_i, e_j) \times \nonumber \\
    & \times \sum_{k=1}^3 \Delta \hat{V}_v^{(n)} (e_k; \vec{m}) \left(\ket{\{\vec{m} + \delta \alpha_{v\ell}^{(n)}\}} - \ket{\{\vec{m} - \delta \alpha_{v\ell}^{(n)}\}}\right),
\end{align}
where 
\begin{equation}
    \Delta \hat{V}_v^{(n)}(e_k; \vec{m}) := \hat{V}_v \ket{\{\vec{m}\}} - \hat{V}_v \ket{\{\vec{m} - \delta_{e_k}^{(n)}\}}.
\end{equation}
The factor $T_v$ is a combinatorial normalisation associated with the vertex sum, namely the total number of contributing triplets entering the chosen regularisation at $v$. The sign $\rho(\ell;\, e_i, e_j) \in \{\pm 1\}$ encodes an orientation convention used to identify the loop contribution associated with $(e_i, e_j)$. Namely, following the Thiemann construction, the loop must be outgoing from $e_1$ and at the vertex $v$ and incoming through $e_2$ when it closes back at the same vertex. A minimal loop of the oriented chosen graph need not be presented in that canonical order from the outset. Whenever the loop has to be reordered in order to satisfy this condition, the corresponding term acquires a multiplicative minus sign.
\newparagraph
The constraint in \eqref{eq:hv_wcl_discrete} is the graph preserving constraint realised in this work. The operator is local in the sense that, for each vertex $v$, it only involves triples of edges incident at $v$, the minimal loops through that same vertex, and the corresponding volume differences computed at $v$. Finally, the object that is minimised numerically is the quadratic operator built from these vertex contributions,
\begin{equation}
    \hat{\mathcal{Q}}_{\hat{H}} = \sum_{v \in V(\gamma_0)} \hat{H}_v \hat{H}_v^\dagger,
\end{equation}
and, depending on the operator ordering consideration, one may analogously study the variant built from $\hat{H}_v^\dagger \hat{H}_v$.

\section{Computational framework}
\label{sec:computationalframework}
Unlike previous work, this works focuses on a slightly larger graph, namely the complete $K_5$ graph. This choice, as previously mentioned, is motivated by several reasons. First, and perhaps most central, is that the $K_5$ graph is the boundary graph of a 4-simplex, and thus the canonical minimal graph underlying 4-dimensional simplicial boundary data. Secondary to that are two consistency reasons: (i) it contains a total of 5 vertices, all of which are 4-valent and thus admit a non-trivial action by the Thiemann constraint, and (ii) it is highly symmetric which strongly reduces and graph-dependent artefacts. Both of the points mentioned allow for a controlled study of genuine operator ordering and dynamics driven effects and the ability to cleanly eliminate purely combinatorial counterparts.
\begin{figure}[h]
    \centering
    \includegraphics[width=0.45\linewidth]{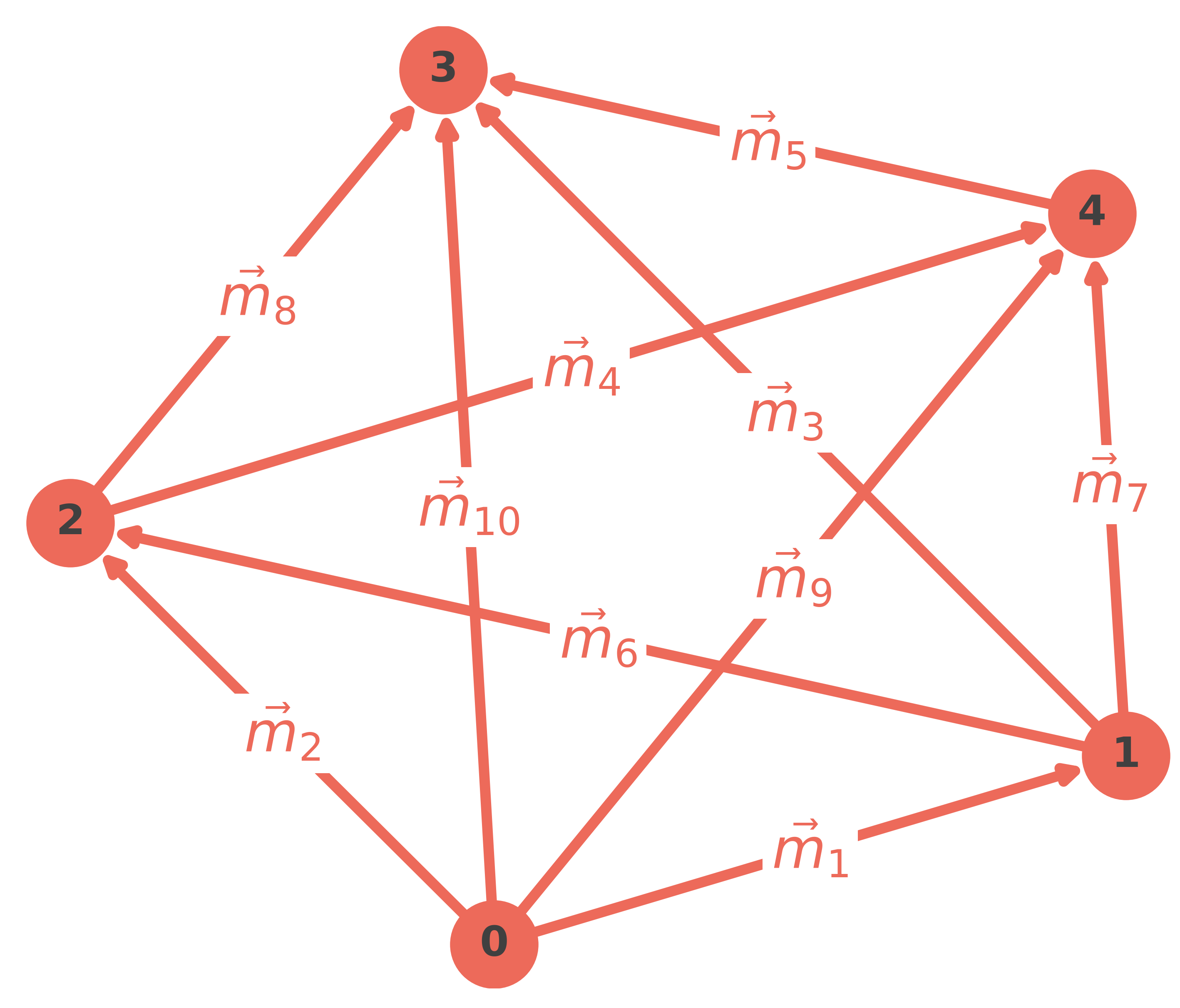}
    \caption{A planar representation of the oriented $K_5$ graph used in this work. Each edge carries a charge vector ($\mathrm{U}(1)^3$ representation label) $\vec{m}_i$.}
    \label{fig:planark5graph}
\end{figure}
\newparagraph
Figure \ref{fig:planark5graph} presents a \emph{planar} illustration of the oriented $K_5$ graph, which we will denote by $\gamma$, used throughout this work. As in \cite{Sahlmann:2024pba,Sahlmann:2024kat}, we perform computations in the dual representation \footnote{the dual graph $\Tilde{\gamma}$ composed of assigned to each edge in $\gamma$ a vertex in $\Tilde{\gamma}$ and then connecting them}. This choice is purely for facilitating the computational aspect and does not affect any results.
\newparagraph
A graph-theoretic aspect which is physically relevant in the present work is the spatial embedding of the non-planar $K_5$ graph. Unlike earlier studies restricted to planar graphs, this work uses non-planar graphs and thus the embedding data becomes essential as it fixes the orientation sign
\begin{equation}
    \epsilon_v(e_i, e_j, e_k) \in \{\pm 1, 0\},
\end{equation}
associated with any ordered triple of incident edges $(e_i, e_j, e_k)$ at a vertex $v \in V(\gamma)$, where $V(\gamma)$ is the set of all vertices of $\gamma$. As discussed in Section \ref{subsec:weakcouplinglimit}, $\epsilon_v(e_i, e_j, e_k)$ is determined by the handedness of the corresponding tangent vectors at $v$ (equivalently, by the sign of the scalar triple product), and it enters the definition of the vertex volume operator $\hat{V}_v$ and thus, this embedding-dependent sign feeds into the Thiemann regularised constraint operator used in our work.
\newparagraph
To avoid introducing ad hoc preferences while still fixing the signs unambiguously, we select a deterministic generic embedding that yields a correspondingly generic pattern of $\mathrm{sgn}$-factors. Starting from initial Cartesian vertex coordinates $(x^1, x^2, x^3)$, we apply a coordinate-wise map
\begin{equation}
    x^i \longmapsto \omega(x^i),
\end{equation}
where $\omega$ deterministically returns a pseudorandom draw from a Gaussian distribution with large variance (and a mean chosen to keep vertices well separated). This construction produces, with probability one, a non-degenerate embedding in $\mathbb{R}^3$. Namely, no three edge tangents at a vertex are coplanar and the scalar triple products are thus nonzero and consequently $\epsilon_v(e_i, e_j, e_k) = \pm 1)$ is well defined for all triplets. In this sense, the embedding is generic (free of special symmetries or degeneracies) without committing to a hand-crafted geometric realisation of the graph.
\newparagraph
As in \cite{Sahlmann:2024pba,Sahlmann:2024kat}, each edge $e \in E(\gamma)$, where $E(\gamma)$ is the set of all edges of $\gamma$, carries a local Hilbert space 
\begin{equation}
    \mathscr{H}_e = \bigotimes_{i = 1}^{3} \mathscr{H}_e^{(1)},
\end{equation}
where here, $\mathscr{H}_e^{(1)}$ denotes a local Hilbert space with basis elements labelled by $\mathrm{U}(1)$ representation labels (charges) and thus, $\mathscr{H}_e$ contains basis elements which carry $\mathrm{U}(1)^3$ labels (charge vectors). The kinematical Hilbert space is thus simply
\begin{equation}
    \mathscr{H}_\mathrm{kin} = \bigotimes_{e \in E(\gamma)} \mathscr{H}_e.
\end{equation}
Similar to previous work \cite{Sahlmann:2024pba,Sahlmann:2024kat}, we truncate the allowed representation labels by specifying a cutoff hereby denoted by $m_\mathrm{max}$ such that any charge $m^i$ in any given charge vector $\vec{m}_k$ is allowed to take values in a symmetric set of charges which we denote by $M := [-m_\mathrm{max}, \cdots, 0, \cdots, m_\mathrm{max}]$. Consequently, the kinematical Hilbert space has dimensions $\dim\mathscr{H}_\mathrm{kin} = (2m_\mathrm{max} + 1)^{3|E(\gamma)|}$.
\newparagraph
Computationally, basis states in $\mathscr{H}_\mathrm{kin}$ are stored using a gauge-strided layout. Concretely, for a generic basis element
\begin{equation}
    \ket{\{\vec{m}\}} := \ket{\vec{m}_1 \vec{m}_2 \cdots \vec{m}_{10}} = \ket{\vec{m}_1}\otimes \ket{\vec{m}_2}\otimes \cdots \otimes \ket{\vec{m}_{10}},
\end{equation}
with $\vec{m}_k=(m_k^{(1)},m_k^{(2)},m_k^{(3)})\in \mathrm{U}(1)^3$, we represent the associated charge data as three contiguous blocks obtained by grouping the components by gauge factor rather than by edge
\begin{equation}
    \big(m_1^{(1)},\ldots,m_{10}^{(1)} \ \big|\  m_1^{(2)},\ldots,m_{10}^{(2)} \ \big|\  m_1^{(3)},\ldots,m_{10}^{(3)}\big).
\end{equation}
Equivalently, the charges are stored as a length-$3|E(\gamma)|$ integer array with stride $|E(\gamma)|$ between consecutive $\mathrm{U}(1)$ components belonging to the same edge. This layout is convenient as many of the elementary operators act component-wise in the $\mathrm{U}(1)^3$ factors. Importantly, this choice is purely an implementation detail, and therefore has no bearing on any physical statements or numerical results.

\subsection{Network architecture}
\label{subsec:networkarchitecture}
Traditional, off-the-shelf neural network Ansätze (e.g. dense Restricted Boltzmann Machine (RBM)-like architectures, generic fully-connected multi-layer perceptrons (MLPs), or attention based models applied naively to flattened edges data) proved inadequate in our setting. This motivated the design of a physics-informed neural quantum state (PINN, in the broad sense of physics-informed inductive bias) whose architecture hard-wires the graph locality and operator support of the model we consider in this work.
\newparagraph
For what follows, we let $\gamma = (E(\gamma), V(\gamma))$ be an arbitrary but fixed oriented graph and denote by $n_e := |E(\gamma)|$ and $n_v := |V(\gamma)|$. For the sake of generality, each edge $e$ carries a $\mathrm{U}(1)^d$ representation label (in our case, $d = 3$) 
\begin{equation}
    \vec{m}_e = (m_e^{(1)}, \ldots, m_e^{(d)}) \in M^d \quad,\quad M = [-m_\mathrm{max}, \ldots, 0 \ldots m_\mathrm{max}] \subset \mathbb{Z}.
\end{equation}
We will denote a full configuration (e.g. basis element over $\gamma$) by $\sigma = \{\vec{m}_e\}_{e \in E(\gamma)} \in \mathscr{H}_\mathrm{kin}$. Computationally, we regard $\vec{m}_e$ as a real vector which we denote by $\vec{x}_e$. We encode the graph locality via the (undirected) incidence mask
\begin{equation}
    A = \in \{0, 1\}^{n_e \times n_v} \quad,\quad A_{ve} = \begin{cases}
        1, \quad e \text{ is incident at } v, \\
        0, \quad \text{otherwise}.
    \end{cases}
\end{equation}
Orientation is accounted for in the operators themselves, the ansatz only needs to incidence structure to respect vertex locality. The output of the network is understood as the complex log-amplitude
\begin{equation}
    \log \psi_\theta(\sigma) = \log|\psi_\theta(\sigma)| + \mathrm{i} \varphi(\sigma) \in \mathbb{C}.
\end{equation}
Throughout this work, we use the sigmoid linear unit (SiLU) activation
\begin{equation}
    \rho(t) = t \cdot \mathrm{sigmoid}(t) = \frac{t}{1 + e^{-t}},
\end{equation}
applied component-wise to vectors. The network is embodies a graph neural network (GNN)-like construction, and thus is constructed as a hierarchy of shared local maps combined with permutation-invariant aggregations, constructed using edge-wise, vertex-wise and global hierarchies which themselves are stacks of constant width MLPs with linear readout vectors. We denote by $F_e$ and $L_e$ the width and the depth of the edge-wise MLP tower. Likewise, $F_v, L_v$ and $F_g, L_g$ denote the vertex-wise and global MLP towers' width and depth respectively. 
\newparagraph
Inputs to the network are first passed through the edge tower, which we denote by $f_\mathrm{edge}$ and is shared over every edge. Thus, let $\vec{h}^{(0)} = \vec{x}_e \in \mathbb{R}^d$ denote an initial input, then for $l = 1, \ldots, L_e$
\begin{equation}
    \vec{h}_e^{(l)} = \rho\left(W_e^{(l)} \vec{h}_e^{(l-1)} + \vec{b}_e^{(l)}\right),
\end{equation}
with $W_e^{(1)} \in \mathbb{R}^{F_e \times d}$, $W_e^{(l)} \in \mathbb{R}^{F_e \times F_e}$ for $l \geq 2$, and the bias vector is $\vec{b}_e^{(l)} \in \mathbb{R}^{F_e}$. Thus, each charge vector $\vec{m}_e$ is embedded into a learned feature vector which we denote by $\vec{z}_e := \vec{h}_e^{(l)} \in \mathbb{R}^{F_e}$.
\newparagraph
The operators considered in this model (i.e. the volume or Thiemann constraint operators) act on small neighbourhoods around vertices and only couple incident edges. Therefore, for each $v \in V(\gamma)$, we aggregate incident edge embeddings by
\begin{equation}
    \vec{s}_v = \sum_{e \in E(v)} A_{ve} \vec{z}_e \in \mathbb{R}^{F_e}.
\end{equation}
This is aggregation makes the representation invariant under any reordering of edges within a vertex neighbourhood. For the vertex tower of MLPs, which we denote by $f_\mathrm{vertex}$, we initialise $\vec{u}_v^{(0)} := \vec{s}_v$. Then, for $l = 1, \ldots, L_v$, apply
\begin{equation}
    \vec{u}_v^{(l)} = \rho\left(W_v^{(l)}\vec{u}_v^{(l-1)} + \vec{b}_v^{(l)}\right),
\end{equation}
with once again the weights and biases being $W_v^{(1)} \in \mathbb{R}^{F_v \times F_e}$, $W_v^{(l)} \in \mathbb{R}^{F_v \times F_v}$ for $l \geq 2$ and $b_v^{(l)} \in \mathbb{R}^{F_v}$. The final vertex embedding is denoted by $\vec{q}_v := \vec{u}_v^{(L_v)} \in \mathbb{R}^{F_v}$.
\newparagraph
Finally, the vertex features are pooled over the full graph using another sum
\begin{equation}
    \vec{g} = \sum_{v \in V(\gamma)} \vec{q}_v \in \mathbb{R}^{F_v},
\end{equation}
which enforces invariance under relabelling of vertices and allows the network to only depend on the abstract graph structure and the configuration $\sigma$. The last stack of MLPs denote the global tower $f_\mathrm{global}$. Initialise $\vec{r}^{(0)} := \vec{g}$, then for $l = 1, \ldots, L_g$ apply
\begin{equation}
    \vec{r}_v^{(l)} = \rho \left(W_g^{(l)} \vec{r}^{(l-1)} + \vec{b}_g^{(l)}\right).
\end{equation}
Similarly, $W_g^{(1)} \in \mathbb{R}^{F_g \times F_v}$, $W_g^{(l)} \in \mathbb{R}^{F_g \times F_g}$ for $l \geq 2$ and $\vec{b}_g^{(l)} \in \mathbb{R}^{F_g}$. The readout of such global tower is denoted by $\vec{y} := r^{(l)} \in \mathbb{R}^{F_g}$. The amplitude head is the bias free
\begin{equation}
    a(\sigma) = \vec{w}_A^\top \vec{y} \quad,\quad \vec{w}_A \in \mathbb{R}^{F_g},
\end{equation}
and likewise, the phase branch head is
\begin{equation}
    \phi(\sigma) = \vec{w}_\varphi^\top \vec{y} \quad,\quad \vec{w}_\varphi \in \mathbb{R}^{F_g}.
\end{equation}
We note that, for reasons which will become apparent in subsequent sections, we add an optional smooth bias discouraging the network from a mode collapse to any all-zero configuration $\sigma$ using
\begin{equation}
    a(\sigma) \mapsto a(\sigma) + \lambda \log(1 + \|\sigma\|_2^2) \quad, \quad \|\sigma\|_2^2 := \sum_{e \in E(\gamma)} \sum_{k = 1}^d (m_e^{(k)})^2,
\end{equation}
with $\lambda \geq 0$, which acts as a gentle, nonsingular regulariser. Thus, for any input configuration
\begin{align}
    \log \psi_\theta(\sigma) = \vec{w}_A^\top f_\mathrm{global} \left(\sum_{v \in V(\gamma)} f_\mathrm{vertex}\left(\sum_{e \in E(v)} A_{ve} f_\mathrm{edge}(\vec{x}_e)\right) \right) + \lambda \log(1 + \|\sigma\|_2^2) \nonumber\\
    + \mathrm{i} \vec{w}_\varphi^\top f_\mathrm{global} \left(\sum_{v \in V(\gamma)}f_\mathrm{vertex} \left(\sum_{e \in E(v)} A_{ve} f_\mathrm{edge}(\vec{x}_e)\right)\right),
\end{align}
where $\vec{x}_e = (m_e^{(1)}, \ldots m_e^{(d)})$. 
\newparagraph
This, rather complicated, architecture exploits the graph based nature of our model as well as the relevant constraint action, leading to better performance compared to generic architectures. For example, the dominant cost and structure in the variational problem come from evaluating graph-local operators that couple only incident edges at every vertex. By constructing features through sums of the form $\sum_{e \in E(v)}$  and $\sum_{v \in V(\gamma)}$, the network's internal representation is aligned with exactly the neighbourhoods on which the physics acts, dramatically reducing the burden on the optimiser compared to a standard dense ansatz that must \emph{learn} locality from data.
\newparagraph
Additionally, those sums over incident edges and vertices make the representation invariant to arbitrary reordering of edges within $E(v)$ and the vertex indexing. Since such orderings are pure bookkeeping, enforcing these symmetries by construction reduces variance and improves sample efficiency. The model cannot waste capacity learning an ordering convention. Note that both $f_\mathrm{edge}$ and $f_\mathrm{vertex}$ are shared across all edges and vertices. This in turn means that the number of trainable parameters is essentially independent of $|E(\gamma)|$ and $|V(\gamma)|$ (up to the final pooling dimension). This is crucial for stability. In this work, as will be shown in what follows, we explore very large Hilbert spaces, with dimensions growing as $(2m_\mathrm{max} + 1)^{d|E(\gamma)|}$. The ansatz, however, remains lightweight but expressive.
\newparagraph
Using a smooth nonlinearity and a \emph{shallow} hierarchy tends to give well behaved gradients in VMC settings \cite{Sherif:2025hfl}, where stochastic estimates and expensive operator evaluations may already contribute significant noise. Deeper, or more generic, architectures often amplify this noise without commensurate representational benefit. Thus, this minimal graph-like neural network is tailored to fit the specific dynamics of this model, expressive exactly where the constraints demand it to be so (i.e. vertex neighbourhoods), invariant under irrelevant relabelling, and small enough that training is dominated by the operator evaluation rather than by an overhead caused by the forward passes of the network itself.

\section{Results}
\label{sec:results}
The 4-dimensional model, including all relevant operators, described in Section \ref{sec:theoreticalframework} was computationally implemented using \neuralqx \cite{Sherif:2026neuralqx}. In what follows, we discuss several results obtained from analysing this model. Our main interest remains in investigating operator ordering in the quadratic constraint and thus, in Section \ref{subsec:solsfamily}, we present the two solution classes obtained and characterise them in terms of their normalisability, flatness and geometric observables, long range correlations and give a qualitative analysis of the geometry they represent in Section \ref{subsec:characterisation}. We conclude with investigating the third class of quasi-solutions for the symmetrically ordered constraint in Section \ref{subsec:symmordsols}. 

\subsection{Solution families of standard and alternative constraint ordering}
\label{subsec:solsfamily}
The main question of this work is to investigate the simple factor ordering in the quadratic constraint by comparing the generic variational solutions one obtains. Namely, for both the standard and alternative ordered quadratic constraints
\begin{equation}
    \hat{\mathcal{Q}}_{\hat{H}} = \sum_{v \in V(\gamma)} \hat{H}_v \hat{H}_v^\dagger, \quad,\quad \hat{\mathcal{Q}}_{\hat{H}^\dagger} = \sum_{v \in V(\gamma)} \hat{H}_v^\dagger \hat{H}_v,
\end{equation}
respectively, we obtain and characterise solutions which lie close to their respective kernels. In all what follows, we work in the $\mathcal{U}_q(1)^3$ gauge invariant subspace of the kinematical Hilbert space, and thus the Gauß constraint is satisfied identically. Additionally, in this setting the diffeomorphism constraint action can be understood as the group averaging over certain graph symmetries. Consequently, once the variational solutions are obtained, a projected variational wavefunction
\begin{equation}
    \psi_\mathrm{diff}(\vec{m}) = \left(\hat{P}_\mathrm{diff}\, \psi\right)(\vec{m}) = \frac{1}{|G|} \sum_{g \in G} \psi(g^{-1} \cdot \vec{m}),
\end{equation}
where $\vec{m}$ is an input basis element, can be constructed either pre- or post-simulation. Unless stated otherwise, all results shown are with respect to the \emph{non-projected} state.
\begin{figure}[htbp]
    \centering

    \begin{subfigure}[t]{0.49\textwidth}
        \centering
        \includegraphics[width=\textwidth]{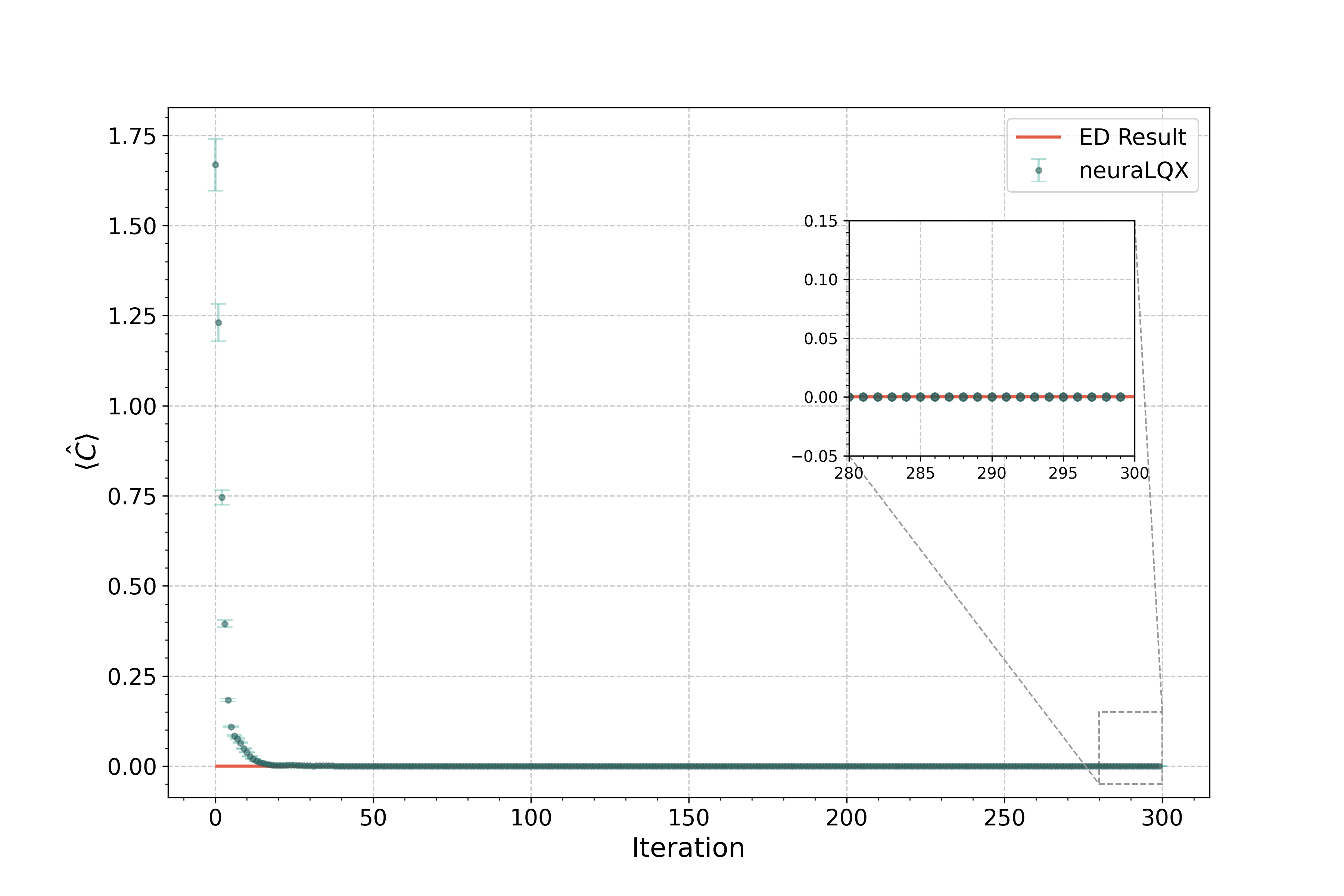}
        \caption{$\hat{\mathcal{Q}}_{\hat{H}} \,\,, \,\, m_\mathrm{max} = 1$}
        \label{fig:mincurvesstdalt:a}
    \end{subfigure}
    \hfill
    \begin{subfigure}[t]{0.49\textwidth}
        \centering
        \includegraphics[width=\textwidth]{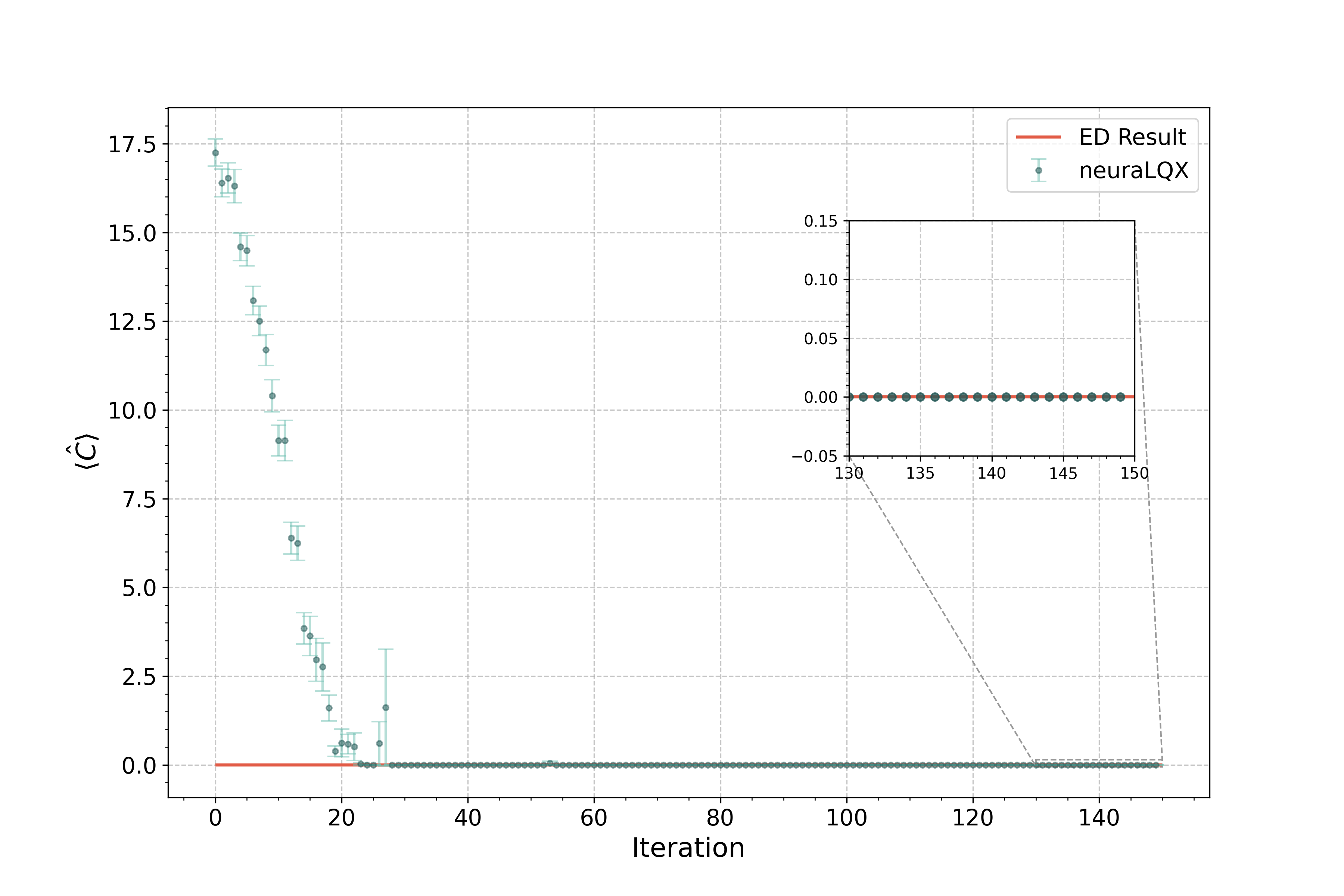}
        \caption{$\hat{\mathcal{Q}}_{\hat{H}^\dagger} \,\,, \,\, m_\mathrm{max} = 1$}
        \label{fig:mincurvesstdalt:b}
    \end{subfigure}

    \vspace{0.5cm}

    \begin{subfigure}[t]{0.49\textwidth}
        \centering
        \includegraphics[width=\textwidth]{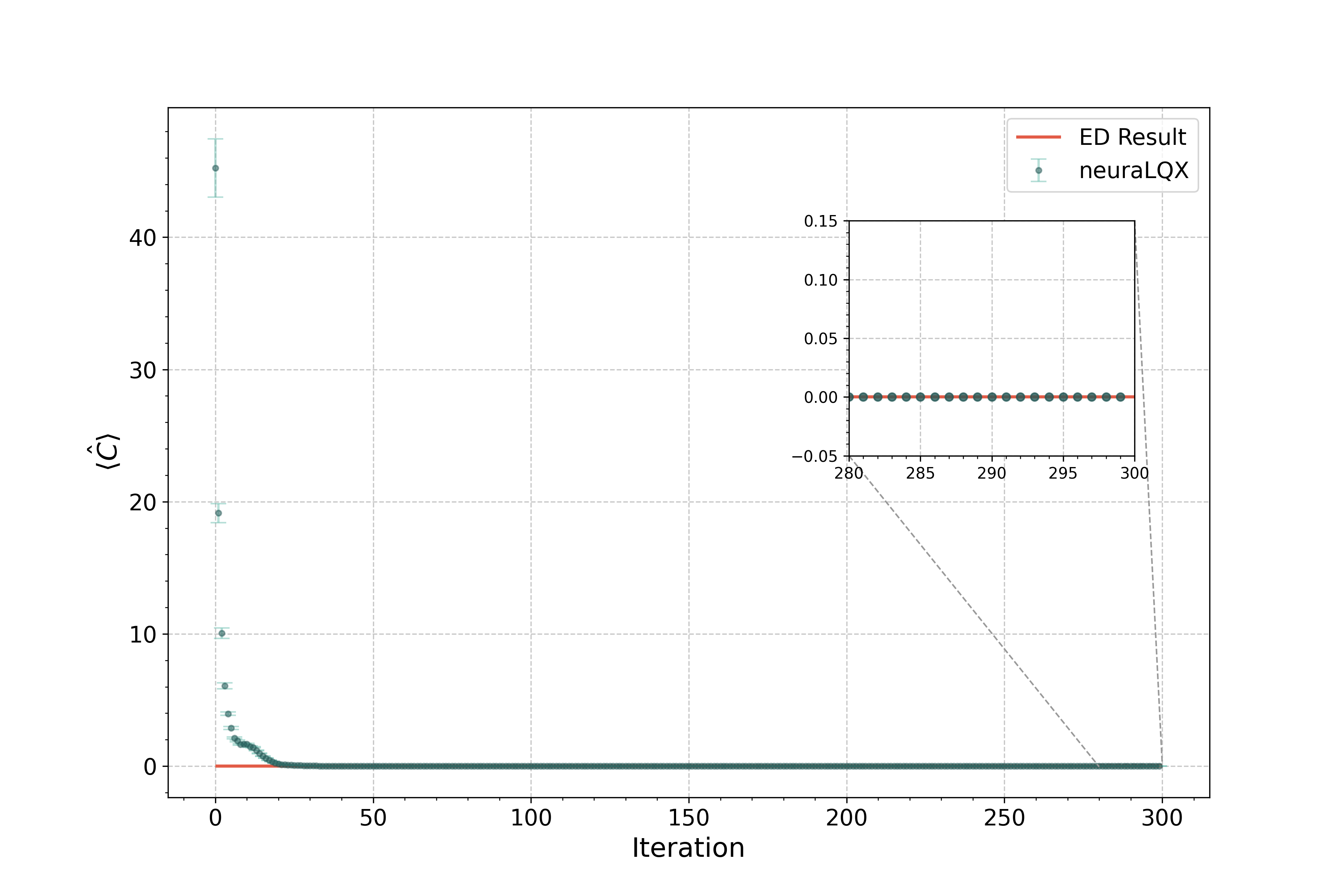}
        \caption{$\hat{\mathcal{Q}}_{\hat{H}} \,\,, \,\, m_\mathrm{max} = 5$}
        \label{fig:mincurvesstdalt:c}
    \end{subfigure}
    \hfill
    \begin{subfigure}[t]{0.49\textwidth}
        \centering
        \includegraphics[width=\textwidth]{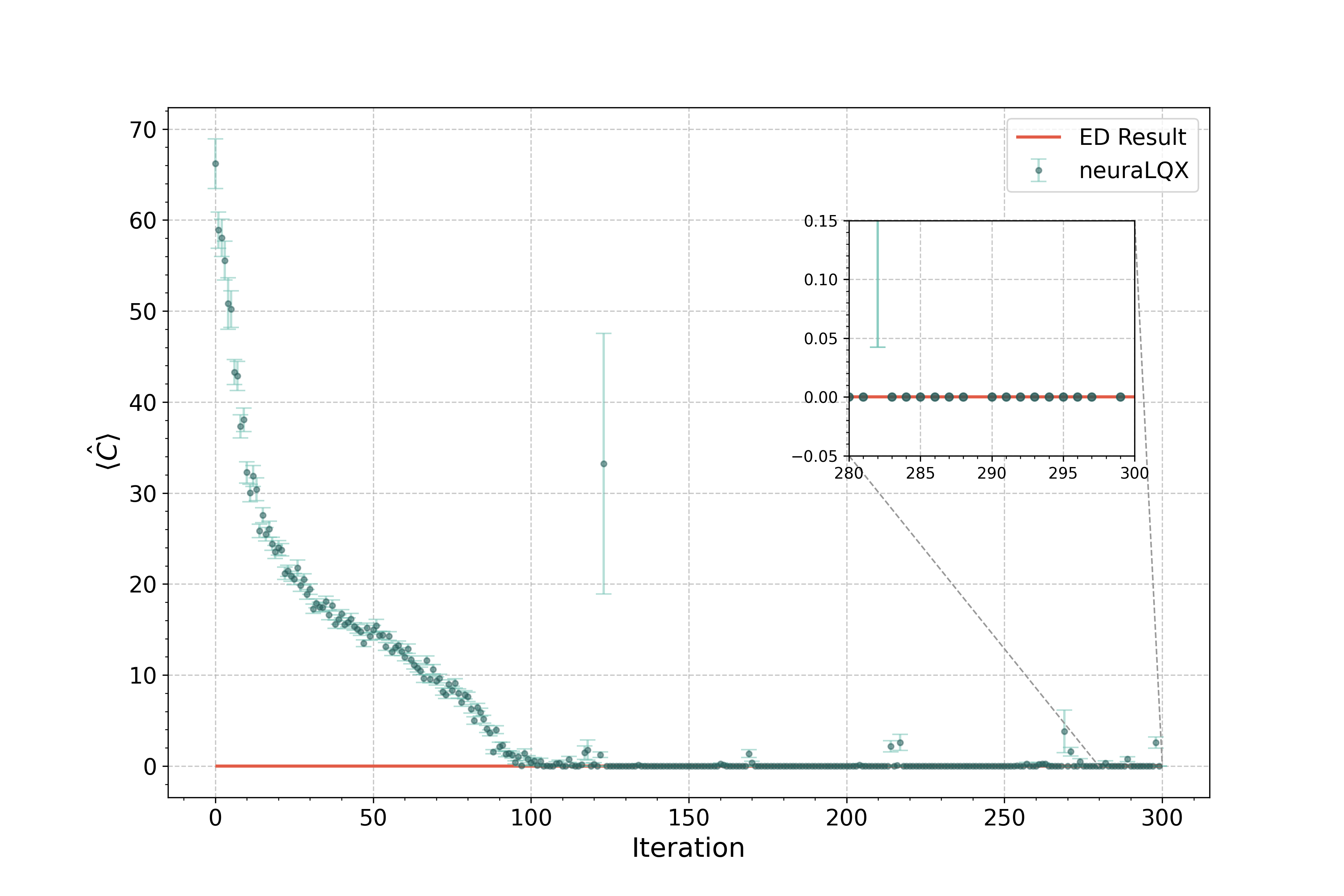}
        \caption{$\hat{\mathcal{Q}}_{\hat{H}^\dagger} \,\,, \,\, m_\mathrm{max} = 5$}
        \label{fig:mincurvesstdalt:d}
    \end{subfigure}

    \caption{
        Variational optimisation curves for solving the quadratic constraint $\hat{\mathcal{Q}_1}$ at $m_\mathrm{max} = 1, 5$ in (a) and (c) and similarly for $\hat{\mathcal{Q}_2}$ at $m_\mathrm{max} = 1, 5$ in (b) and (d). The results are representative of simulations conducted across 75 different seeds each.
    }
    \label{fig:mincurvesstdalt}
\end{figure}
\newparagraph
Figure \ref{fig:mincurvesstdalt} shows the variational optimisation curves for solving the standard ordered constraint $\hat{\mathcal{Q}}_{\hat{H}}$ at cutoffs of $m_\mathrm{max}=1, 5$ in Figures \ref{fig:mincurvesstdalt:a} and \ref{fig:mincurvesstdalt:c} respectively and similarly for the alternative ordered constraint $\hat{\mathcal{Q}}_{\hat{H}^\dagger}$ in Figures \ref{fig:mincurvesstdalt:b} and \ref{fig:mincurvesstdalt:d} respectively. The shown results are representative results of 75 different simulations for each case. 
\newparagraph
For the simulations shown in the Figure \ref{fig:mincurvesstdalt}, the neural network used in the variational ansatz and the optimisation parameters used in the work were as follows
\begin{table}[h]
\small
    \centering
    \caption{\small Simulation parameters for simulations solving the standard and alternative ordered constraints at $m_\mathrm{max} = 1$ and 5.}
    \begin{tabular}{c|c|cccccccc}
        \toprule[1.2pt]
        Constraint &
        $m_\mathrm{max}$ &
        $L_g$ &
        $L_v$ &
        $L_e$ &
        $F_g$ &
        $F_v$ &
        $F_e$ &
        $N_\mathrm{MC}$ &
        $n_\mathrm{chains}$ 
        \\
        \midrule

        \multirow{2}{*}{$\hat{\mathcal{Q}}_{\hat{H}}$}
           & 1 & 1 & 2 & 2 & 32 & 32 & 16 & 1350 & 10  \\
           & 5 & 3 & 2 & 2 & 32 & 32 & 16 & 2700 & 10  \\

        \midrule

        \multirow{2}{*}{$\hat{\mathcal{Q}}_{\hat{H}^\dagger}$}
           & 1 & 1 & 2 & 2 & 32 & 32 & 16 & 1350 & 10  \\
           & 5 & 3 & 2 & 2 & 32 & 32 & 16 & 2700 & 10  \\

        \bottomrule[1.2pt]
    \end{tabular}
    \label{tab:stdaltsimsetup}
\end{table}
\\\noindent
Additionally, all simulations were conducted using the adaptive momentum (Adam) optimiser with an annealed learning rate, starting at a value of $0.005$ and annealed linearly to a value of $0.0001$ over the course of the entire simulation. As also shown in Table \ref{tab:stdaltsimsetup}, the sizes of the networks used in such cutoffs is vanishingly small compared to the dimensions of the respective Hilbert spaces. This can be seen more explicitly in Table \ref{tab:stdaltresults}. 
\begin{table}[h]
\small
    \centering
    \caption{\small Representative simulation results for solving both standard and alternative ordered constraints at $m_\mathrm{max} = 1$ and 5. The compression rate is computed as $|\mathrm{Params}(\Psi_\mathrm{NQS})| / \dim\mathscr{H}_\mathrm{kin}^\mathrm{GI}$ in percentage.}
    \begin{tabular}{c|c|cc}
        \toprule[1.2pt]
        Constraint &
        $m_\mathrm{max}$ &
        $\langle \hat{\mathcal{Q}} \rangle_\mathrm{NQS}$ &
        Compression rate (\%)
        \\
        \midrule

        \multirow{2}{*}{$\hat{\mathcal{Q}}_{\hat{H}}$}
           & 1 & $(8.6 \pm 4.6)\times 10^{-6}$ & $7.8881 \times 10^{-4}$  \\
           & 5 & $(8.25 \pm 0.34)\times 10^{-5}$ & $3.1176 \times 10^{-26}$  \\

        \midrule

        \multirow{2}{*}{$\hat{\mathcal{Q}}_{\hat{H}^\dagger}$}
           & 1 & $(2.066 \pm 0.048)\times 10^{-5}$ & $7.8881 \times 10^{-4}$  \\
           & 5 & $0.0 \pm 0.0$ & $3.1176 \times 10^{-26}$  \\

        \bottomrule[1.2pt]
    \end{tabular}
    \label{tab:stdaltresults}
\end{table}
\\\noindent
As shown in Table \ref{tab:stdaltresults}, throughout the cutoffs of $m_\mathrm{max} = 1, \ldots, 5$, despite the vanishingly small number of network parameters in all cases as shown in the compression rate\footnote{Here, the compression rate denotes the rate at which the variational optimisation problem (i.e. the number of free parameters in the model which have to be fine tuned to approximate a solution, which has a worst case scenario of $\dim\mathscr{H}_\mathrm{kin}^\mathrm{GI}$) is compressed, and is computed using $|\mathrm{Params}(\Psi_\mathrm{NQS})| / \dim\mathscr{H}_\mathrm{kin}^\mathrm{GI}$.}, approximate variational near-kernel states to both the standard and alternative constraints were obtained whereby the final constraints expectation values are relatively small. The shown solutions are representative of several simulations conducted with different seeds. Note that for $m_\mathrm{max} = 5$ for the case of $\hat{\mathcal{Q}}_{\hat{H}^\dagger}$, unlike the shown shown, identically zero solution, the \emph{typical} solution expectation value is in the order of magnitude of $10^{-1}$. 
\newparagraph
We also note that solutions at large cutoffs of $m_\mathrm{max} = 50$ as well as $m_\mathrm{max} = 100$ for the case of, especially, $\hat{\mathcal{Q}}_{\hat{H}}$ were successfully obtained with comparable final constraint expectation values as shown in the table above. However, to keep the presentation consistent, and as several of our upcoming characterisation analyses become infeasible at such large Hilbert spaces, we will only concern ourselves with cutoffs of at most 5 when comparing two solutions of the two orderings. The optimisation feasibility and the coarse graining qualitative signatures accessible at high cutoff (e.g. constraint suppression and selected local diagnostics) persists at such high $m_\mathrm{max}$ cutoffs but the full characterisation suite is only reported for $m_\mathrm{max} \leq 5$ where it is statistically controlled. We note that, nevertheless, these large $m_\mathrm{max}$ simulations, however, serve to prove that one can indeed conduct LQG simulations using the developed numerical framework up to very large-scale, previously unattainable regimes.
\newparagraph
While the Figure \ref{fig:mincurvesstdalt} and Table \ref{tab:stdaltresults} demonstrate that the optimiser drives $\langle \hat{\mathcal{Q}}_i \rangle, i = \hat{H}, \hat{H}^\dagger$ to small, near-zero values across cutoffs, the \emph{absolute} magnitude of the final numbers is not, by itself, necessarily an adequate measure for how close one is to the kernel, especially when comparing different $m_\mathrm{max}$ or different orderings. This is particularly relevant for the $\hat{\mathcal{Q}}_{\hat{H}^\dagger}$ case at $m_\mathrm{max} = 5$ where, as mentioned, unlike the identically vanishing run shown in Table \ref{tab:stdaltresults} a typical converged value is of the order $10^{-1}$. Taken in isolation, this can appear less compelling than the $10^{-5}$-$10^{-6}$ values seen for other cutoffs. 
\newparagraph
However, $\langle \hat{\mathcal{Q}}_i \rangle, i = \hat{H}, \hat{H}^\dagger$ are positive semidefinite sums of vertex contributions and thus, for the alternative ordered constraint
\begin{equation}
    \langle \hat{\mathcal{Q}}_{\hat{H}^\dagger} \rangle_\mathrm{NQS} = \sum_{v \in V(\gamma)} \langle \Psi_\mathrm{NQS}, \hat{H}_v^\dagger \hat{H}_v \Psi_\mathrm{NQS} \rangle \geq 0,
\end{equation}
and thus what matters for interpretability is the scale on which these squared norms \emph{typically} live in the \emph{same} truncated theory. Crucially, as the cutoff is increased, the Hilbert space at lower cutoff becomes  a subspace. Thus the norm of the constrained operator cannot decrease with larger cutoff. In particular, any variational proxy for $||\hat{\mathcal{Q}}_{\hat{H}^\dagger}|| = \sup_\Psi \langle \hat{\mathcal{Q}}_{\hat{H}^\dagger} \rangle_\Psi$ is monotone non-decreasing under an increase in $m_\mathrm{max}$. Therefore, the observation that already at $m_\mathrm{max} = 3$ one finds a generic variational magnitude $\langle \hat{\mathcal{Q}}_{\hat{H}^\dagger} \rangle_\mathrm{gen} \approx 3.8 \times 10^2$ (cf. Section \ref{subsubsec:genericvmagconst12}) implies that the corresponding typical scale at $m_\mathrm{max} = 5$ must be \emph{at least} as large (and generically larger since increasing $m_\mathrm{max}$ permits larger charge excitations and hence larger matrix elements in the vertex operators), which can also be seen in Figure \ref{fig:Qsym_scale} (see Section \ref{subsec:symmordsols}) where such a generically achievable variational maximum of this constraint is an eigenspace well above 1200. From this perspective, an attained value of $\langle \hat{\mathcal{Q}}_{\hat{H}^\dagger} \rangle_\mathrm{NQS} \sim 10^{-1}$ at $m_\mathrm{max} = 5$ is already consistent with a suppression by \emph{multiple} orders of magnitude to the intrinsic scale of $\hat{\mathcal{Q}}_{\hat{H}^\dagger}$ at that cutoff.

\subsubsection{Generic variational magnitude of the constraints}
\label{subsubsec:genericvmagconst12}
A conceptually clean way to quantify how close a variational state is to the target physical subspace (in this case, the zero eigenspace) would be to compare it to the spectral information of the respective quadratic constraints by, for example, locating the lowest eigenvalues, estimating a spectral gap or projecting onto the exact kernel of the operators. In our setting, this route is simply unavailable. The kinematical Hilbert space grows as $\dim\mathscr{H}_\mathrm{kin} = (2m_\mathrm{max} + 1)^{30}$ for the present case of the $K_5$ graph and thus even for the smallest cutoff of $m_\mathrm{max} = 1$, exact diagonalisation is infeasible. This is true even if one considers the gauge invariant subspace, as in our case, in which even at $m_\mathrm{max} = 1$, then $\dim\mathscr{H}_\mathrm{kin}^\mathrm{GI} \sim 10^8$. Consequently, even in such a case exact diagonalisation methods remain unavailable and sparse extremal-eigenvalue methods are infeasible in a controlled manner, especially with higher $m_\mathrm{max}$ values. Consequently, throughout we adopt a genuinely variational viewpoint. Namely, we measure proximity to the kernel through the Monte Carlo estimate of the $\langle \hat{\mathcal{Q}}_i \rangle_\mathrm{NQS}, i = \hat{H}, \hat{H}^\dagger$ for the two quadratic constraints and interpret $\langle \hat{\mathcal{Q}}_i \rangle_\mathrm{NQS} \approx 0.0$ (within statistical uncertainty) as evidence that the variational state lies very close to the corresponding kernel.
\newparagraph
A natural follow-up question would be whether the \emph{absolute} values of $\langle \hat{\mathcal{Q}}_i \rangle_\mathrm{NQS}$ reported in Table \ref{tab:stdaltresults} are small only in an absolute sense, or also small relative to a \emph{typical} magnitude of the constraint operator at the same cutoff within the same variational family. Since we do not have access to the spectrum, we cannot normalise by, say, an operator norm or spectral gap. Nevertheless, one may define an \emph{intrinsic variational scale} by solving an auxiliary optimisation problem designed to make the constraint expectation value as large as possible within the chosen NQS manifold. A convenient way to implement this is to minimise a variant of the regularised inverse-square objective
\begin{equation}
    \mathcal{L}^{(i)}_\mathrm{inv}[\Psi_\mathrm{NQS}] := \frac{1}{(\langle \hat{\mathcal{Q}}_i \rangle_{\Psi_\mathrm{NQS}}+ a)^2} \quad, \quad a > 0,
\end{equation}
where the small shift $a$ prevents numerical pathologies when the expectation value $\langle \hat{\mathcal{Q}}_i \rangle_{\Psi_\mathrm{NQS}}$ is close to zero. Because $\hat{\mathcal{Q}}_i$ are positive semidefinite, minimising $\mathcal{L}_\mathrm{inv}^i$ is monotone-equivalent to maximising $\langle \hat{\mathcal{Q}}_i \rangle$ over the same variational class and with the same VMC estimators. Thus, the resulting value $\langle \hat{\mathcal{Q}}_i \rangle_\mathrm{gen} $ can be taken as a proxy for a \emph{generic maximal variational magnitude} of the constraint at that cutoff and ordering.
\newparagraph
To illustrate, at $m_\mathrm{max} = 3$, for the alternative ordering we find for the following variant of the inverse-squared objective
\begin{equation}
    \tilde{\mathcal{L}}_\mathrm{inv}[\hat{\mathcal{Q}}_{\hat{H}^\dagger}] := \sum_{v \in V(\gamma)} \frac{1}{(\langle \hat{H}_v^\dagger \hat{H}_v\rangle_{\Psi_\mathrm{NQS}} + a)^2},
\end{equation}
then
\begin{equation}
    \langle \hat{\mathcal{Q}}_{\hat{H}^\dagger} \rangle_\mathrm{gen} \approx 5 \,\cdot\, \frac{1}{\sqrt{\tilde{\mathcal{L}}_\mathrm{inv}[\hat{\mathcal{Q}}_{\hat{H}^\dagger}] / 5}} - a \approx 377.9747,
\end{equation}
for $a = 1 \times 10^{-4}$. This shows that the \emph{typical} variationally achieved near-kernel values of $\langle \hat{\mathcal{Q}}_{\hat{H}^\dagger} \rangle$ in this work (in the order of $10^{-1}$) can be viewed as being suppressed by many orders of magnitude relative to this intrinsic scale. Further, for $m_\mathrm{max} = 5$, the generic magnitude is well above 1200 (see Figure \ref{fig:Qsym_scale} in Section \ref{subsec:symmordsols}). In this specific sense, the inverse objective calibration provides a meaningful, ansatz-adapted normalisation without requiring any spectral input.
\newparagraph
It is important, however, to stress what this procedure \emph{does not} provide, and why we do not pursue it systematically across all cutoffs and orderings. First, $\langle \hat{\mathcal{Q}}_i \rangle_\mathrm{gen}$ is not an operator-norm estimate and it does not approximate the true spectral maximum. Rather, it is merely a \emph{variational reach} determined by the expressivity of the ansatz and the optimisation landscape. Second, it is a \emph{separate} (and, in practice, comparably expensive) optimisation problem. Driving $\langle \hat{\mathcal{Q}}_{\hat{H}^\dagger} \rangle$ large can be as non-trivial as driving it to zero, with additional and different sensitivity to the choice of regulariser $a$, optimisation hyperparameters, and so on. Third, even if obtained robustly, this scale does not by itself characterise nor help in characterising the physical solution space, it rather yields a useful dimensionless suppression factor. For these reasons, we rely on the variational observation that $\langle \hat{\mathcal{Q}}_i \rangle$ can be driven relatively close to zero, and in what follows we therefore regard the obtained solutions as close to the kernel and proceed to characterise them.

\subsection{Characterisation of the solution classes}
\label{subsec:characterisation}
The variational study of the two orderings $\hat{\mathcal{Q}}_{\hat{H}}$ and $\hat{\mathcal{Q}}_{\hat{H}^\dagger}$ reveals a striking and reproducible pattern. Namely, although both quadratic constraints can be driven to very small expectation values, the resulting near-kernel states are not the same. Instead, the optimisation consistently settles into two distinct solution classes that are naturally or generically associated with the two orderings. In practice, simulations targeting $\hat{\mathcal{Q}}_{\hat{H}}$ converge to a family of wavefunctions (which we will throughout denote by \emph{type-A solutions}) that share one set of qualitative physical features whereas those targeting $\hat{\mathcal{Q}}_{\hat{H}^\dagger}$ converge to a different family (which we will throughout denote by \emph{type-B solutions}) with clearly distinguishable characteristics. This dichotomy persists across all cutoffs $m_\mathrm{max}$  considered and remains stable under repeated trainings with different random initialisations and optimiser settings. This, in turn, suggests that it reflects genuine structure of the variational landscape induced by the different operator orderings rather than a contingent artefact of a particular simulation. 
\newparagraph
The goal of this section is to characterise and contrast the two ordering-induced solution classes. To do so, we organise the discussion around three complementary diagnostics probing distinct aspects of the obtained solution classes. First, Section \ref{subsubsec:normalisability} examines the normalisability and concentration properties of the learned wavefunctions in their respective Hilbert spaces. Section \ref{subsubsec:correlations} then investigates long-range correlations across the graph by studying 2-point functions (and their fluctuations) of the charge distribution on the edges of the graph, providing a sensitive handle on whether the two ordering-induced solution classes differ primarily by local re-weighting of different basis elements or by genuinely non-local collective structure. Lastly, Section \ref{subsubsec:geometricstructure} probes the geometric content of the two classes using expectation values of geometric observables such as both area and volume, thereby connecting the ordering-dependent near-kernel solutions to emergent (or absent) geometric features.

\subsubsection{Normalisability}
\label{subsubsec:normalisability}
Solutions of the quantum constraints are not expected to be normalisable, generally, with respect to the kinematical inner product. Rather, they often live naturally as distributional objects with the \emph{physical} inner product to be defined only after solving the constraints. From the perspective of Thiemann regularised quantum Hamilton constraints and the related quadratic constraints, this is not a pathology per se but rather reflects the fact that the Hamiltonian constraint generates a gauge flow rather than a true evolutions. Therefore, its solutions need not belong to $\mathscr{H}_\mathrm{kin}$ even when the constraint operator is densely defined there. 
\newparagraph
In the present work, however, our variational method is implemented directly in the gauge invariant subspace of the truncated kinematical Hilbert space, and the NQS defines an honest square-summable vector at each finite cutoff. It is therefore meaningful to ask whether the sequence of variational near-kernel states remains \emph{compatible} with normalisability as the cutoff is increased, or whether it exhibits signatures of a runaway toward a distributional object in the $m_\mathrm{max} \rightarrow \infty$ limit.
\newparagraph
Concretely, recall that in the charge-network basis, a kinematical state is specified by coefficients $\psi_{\vec{m}_1 \cdots \vec{m}_{10}}$ and normalisability in the (non-truncated) kinematical Hilbert space amounts to the $\ell^2$-condition
\begin{equation}
    \|\Psi\|^2 = \sum_{\vec{m}_1 \cdots \vec{m}_{10}} |\psi_{\vec{m}_1 \cdots \vec{m}_{10}}|^2 < \infty.
\end{equation}
At finite $m_\mathrm{max}$, this sum is finite by construction. The potential issue arises only when viewing the trained solutions $\Psi^{(m_\mathrm{max})}$ at a certain cutoff $m_\mathrm{max}$ as a \emph{sequence} under refinement of the truncation. A natural quantity to inspect under cutoff refinement is the probability carried by the newly admitted outer shell
\begin{equation}
    \mathcal{S}_{m_\mathrm{max}} := \{\vec{m} \in M^{3|E(\gamma)|} \,:\, \max_{e,k} |m_e^{(k)}| = m_\mathrm{max}\},
\end{equation}
together with the induced probability distribution $p_\mathrm{max}(\vec{m}) = |\Psi^{(m_\mathrm{max})}(\vec{m})|^2 / \|\Psi^{(m_\mathrm{max})}\|^2$. One may then monitor the shell weight
\begin{equation}
    q_{m_\mathrm{max}} := \sum_{\vec{m} \in \mathcal{S}_{m_\mathrm{max}}} p_{m_\mathrm{max}}(\vec{m})
\end{equation}
which measures how much probability mass is carried by configurations that only become available at the current cutoff. Persistent order-one values of $q_{m_\mathrm{max}}$ indicate that the variational state does not concentrate away from the cutoff boundary as the truncation is enlarged.
\newparagraph
This, however, must then be interpreted with care. In general, substantial weight near the cutoff does \emph{not} by itself distinguish between a sequence approximating a genuinely non-normalisable continuum solution and a cutoff-induced branch whose support is tied to truncation artefacts. Moreover, any comparison of $\Psi^{m_\mathrm{max}}$ across different cutoffs already assumes that the optimisation continues to track approximations to the same underlying continuum solution, which need not hold a priori and rarely does in practice. 
\newparagraph
For this reason, we do not elevate the shell weight to a general criterion for normalisability. Rather, we regard it as merely a heuristic diagnostic of whether the learned states remain tight under refinement. The conducted normalisability analysis rather relies on two operational and visual diagnostics that can be evaluated in the present case of rather large $\dim\mathscr{H}_\mathrm{kin}^\mathrm{GI}$, in which case methods relying on direct enumeration of the full state vector are not possible. The first is a state-vector scan based on the sequential labelling $\iota$ of Appendix \ref{app:sequential_labelling} and the stratified-envelope construction of Appendix \ref{subsubsec:fast_envelope_sampling}. Concretely, for a given cutoff we regard the variational map as a coefficient map $\psi_i, i \in \{0, \ldots, D-1\}$, where $D = \dim\mathscr{H}_\mathrm{kin}^{\mathrm{GI}}$ and partition the index set into into $B \ll D$ contiguous bins $\mathcal{I}_b$. We then draw $K$ indices per bin, $i_{b,t} \sim \text{Uniform}(\mathcal{I}_b)$, decode the corresponding configurations and evaluate the generally unnormalised amplitudes which we denote by $\tilde{\psi}_{b,t}$ as defined in Appendix \ref{subsubsec:fast_envelope_sampling}. From these samples, we compute the stratified, unbiased estimator of the squared norm
\begin{equation}
    \widehat{\|\tilde{\Psi}\|^2} = \sum_{b=0}^{B-1} \frac{W_b}{K} \sum_{t=1}^K |\tilde{\psi}_{b,t}|^2,
\end{equation}
and rescale the sampled coefficients for visualisation via
$\hat{\psi}_{b,t} := \tilde{\psi}_{b, t} / \widehat{\|\tilde{\Psi}\|}$. The resulting envelope bands, as discussed in Appendix \ref{subsubsec:fast_envelope_sampling}, provide a compressed but global view of both $\Re(\psi_i)$ and $\Im(\psi_i)$ across the full basis ordering. For fixed binning, this estimator converges to the exact squared norm $\|\tilde{\Psi}\|^2$ as the number of samples per bin $K$ is increased, while increasing the number of bins $B$ refines the spatial resolution of the envelope by decreasing the bin widths.
\newparagraph
This diagnostic is particularly informative in the refinement analysis. If a family of variational states $\Psi^{(m_\mathrm{max})}$ does approximate one and the same $\ell^2$-state under cutoff increase, then the probability mass in the newly admitted near-cutoff sectors is expected to decay with increasing cutoff. From this perspective, persistence of a noise-like, non-decaying envelope structure in the index regions dominated by near-cutoff configurations is a strong indication of a failure of concentration. By itself, this does not constitute a general proof of non-normalisability. In the type-A case as to be shown below, however, it fits consistently with the further evidence that the states approximate a non-normalisable continuum solution.
\newparagraph
To see this, we employ the algorithm in Appendix \ref{app:sequential_indexing_and_plotting} to obtain the stratified state-vectors for all solutions (across all seeds) per cutoff. It was observed that \emph{all} solutions obtained in a given cutoff, as well as across all cutoffs, exhibited the same behaviour depending on their solution class. To begin, stratified-state vectors of type-A solutions (corresponding to solutions of $\hat{\mathcal{Q}}_{\hat{H}}$) are shown below.
\begin{figure}[htbp]
    \centering

    \begin{subfigure}[t]{0.49\textwidth}
        \centering
        \includegraphics[width=\textwidth]{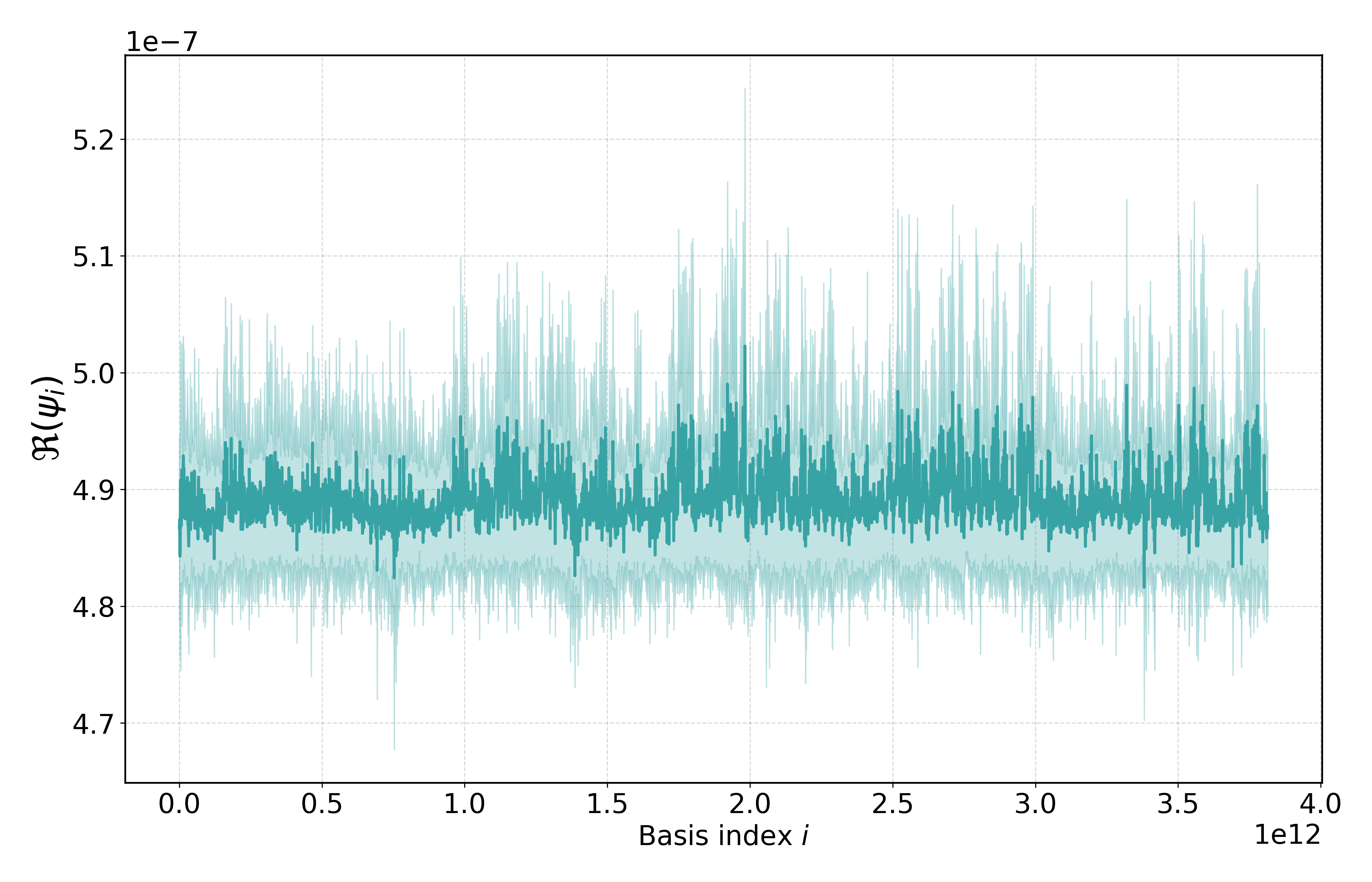}
        \caption{$m_\mathrm{max} = 2$}
        \label{fig:statevecstd:a}
    \end{subfigure}
    \hfill
    \begin{subfigure}[t]{0.49\textwidth}
        \centering
        \includegraphics[width=\textwidth]{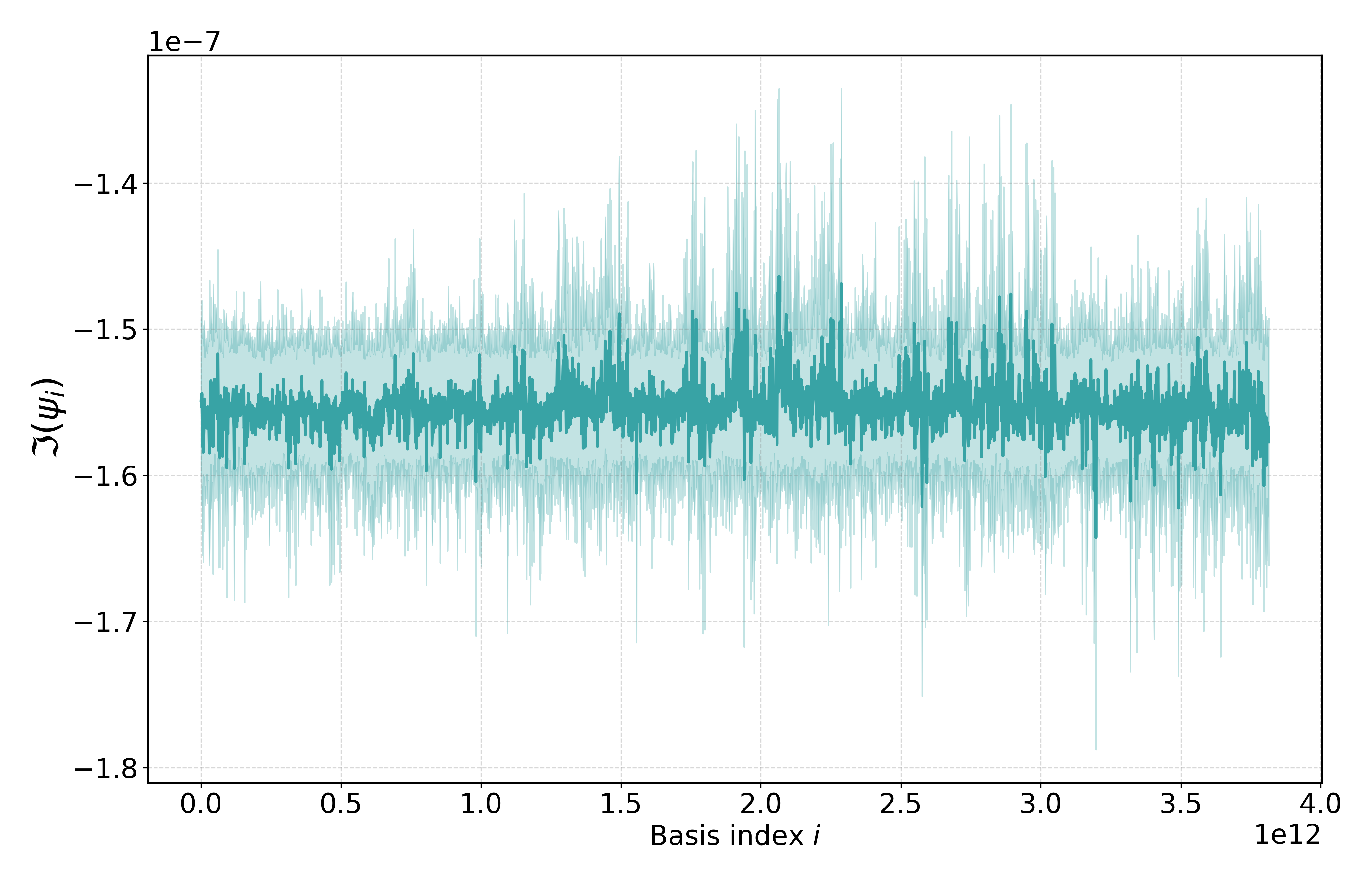}
        \caption{$ m_\mathrm{max} = 2$}
        \label{fig:statevecstd:b}
    \end{subfigure}

    \vspace{0.5cm}

    \begin{subfigure}[t]{0.49\textwidth}
        \centering
        \includegraphics[width=\textwidth]{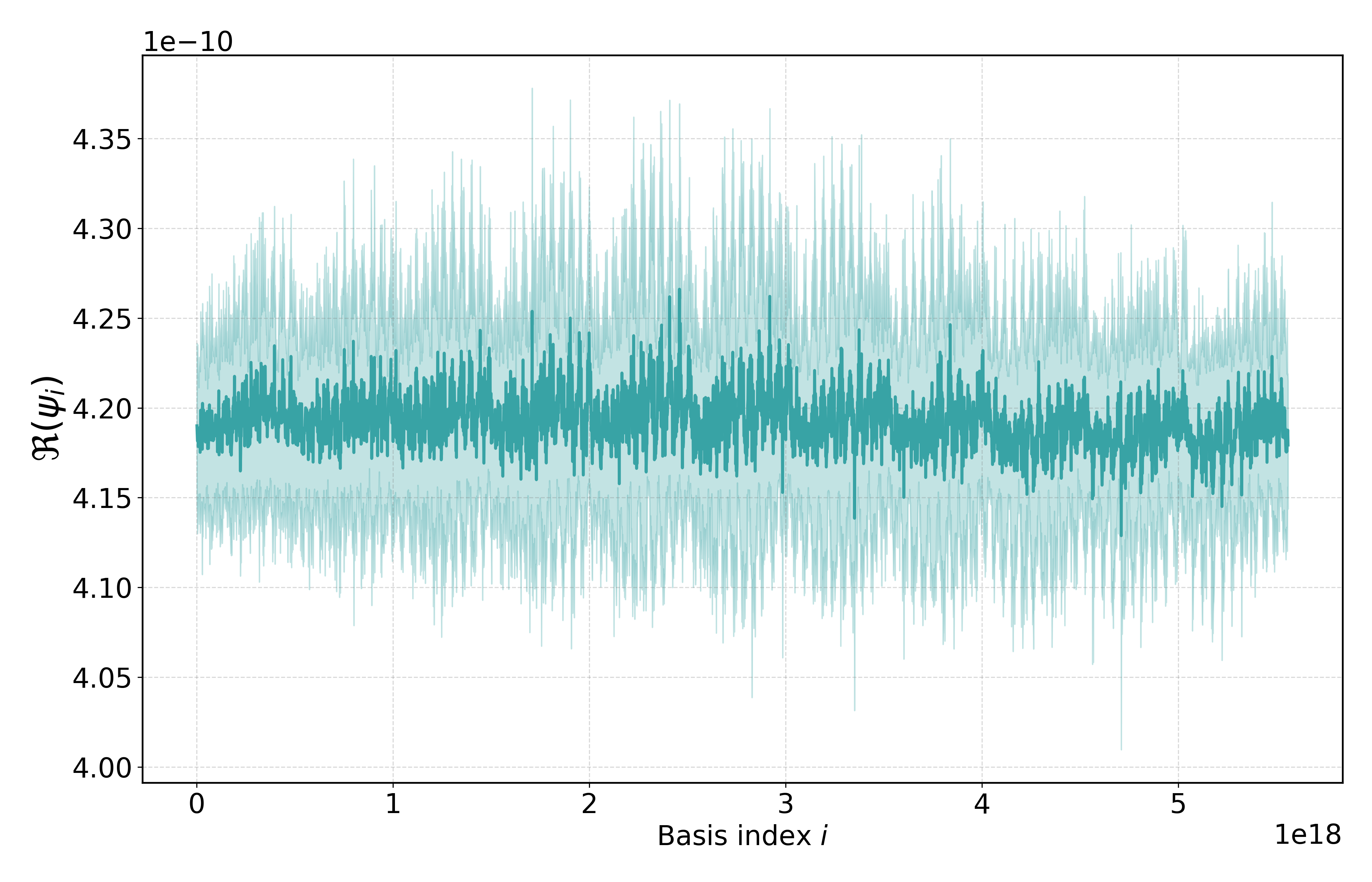}
        \caption{$m_\mathrm{max} = 5$}
        \label{fig:statevecstd:c}
    \end{subfigure}
    \hfill
    \begin{subfigure}[t]{0.49\textwidth}
        \centering
        \includegraphics[width=\textwidth]{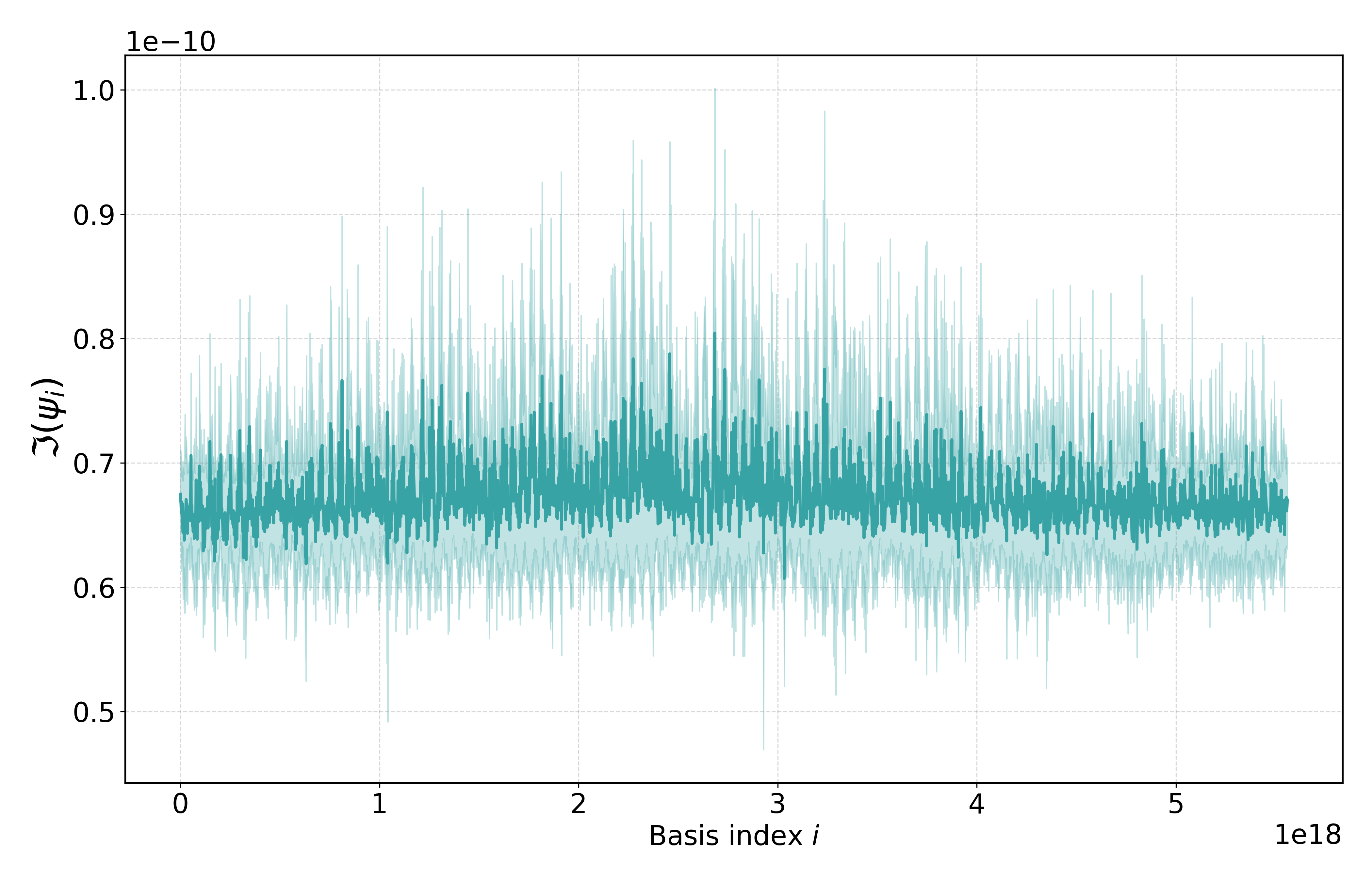}
        \caption{$m_\mathrm{max} = 5$}
        \label{fig:statevecstd:d}
    \end{subfigure}

    \caption{
        Stratified state-vector plots of representative variational near-kernel type-A solutions for the cutoffs $m_\mathrm{max} = 2, 5$
    }
    \label{fig:statevecstd}
\end{figure}
\newparagraph
Figure \ref{fig:statevecstd} shows the stratified state-vector plots of representative variational near-kernel type-A solutions for the cutoffs $m_\mathrm{max} = 2$ (Figures \ref{fig:statevecstd:a} and \ref{fig:statevecstd:b} for real and imaginary parts respectively) and similarly both real and imaginary parts for for $m_\mathrm{max} = 5$ in Figures \ref{fig:statevecstd:c} and \ref{fig:statevecstd:d} respectively. As shown in Figure \ref{fig:statevecstd}, the characteristic noise-like band structure exists irrespective of the cutoff. With the chosen sequential indexing, the majority of near cutoff basis elements are on both the left and right tails of each of the plots above. Since one does not observe any suppression of newly added outer shell amplitudes, this is taken as a strong indicator that the solutions obtained are not normalisable. 
\newparagraph
We emphasise that in this section the stratified estimator is used primarily to provide a consistent approximate normalisation and to expose global structure of the coefficients under refinement. For the type-A class, the absence of visible suppression in the near-cutoff index regions persists across cutoffs, matching the interpretation that weight continues to populate newly admitted outer-shell configurations rather than concentrating.
\newparagraph
To compare, stratified state-vector plots for representative type-B solutions are shown in Figure \ref{fig:statevecalt}.
\begin{figure}[htbp]
    \centering

    \begin{subfigure}[t]{0.49\textwidth}
        \centering
        \includegraphics[width=\textwidth]{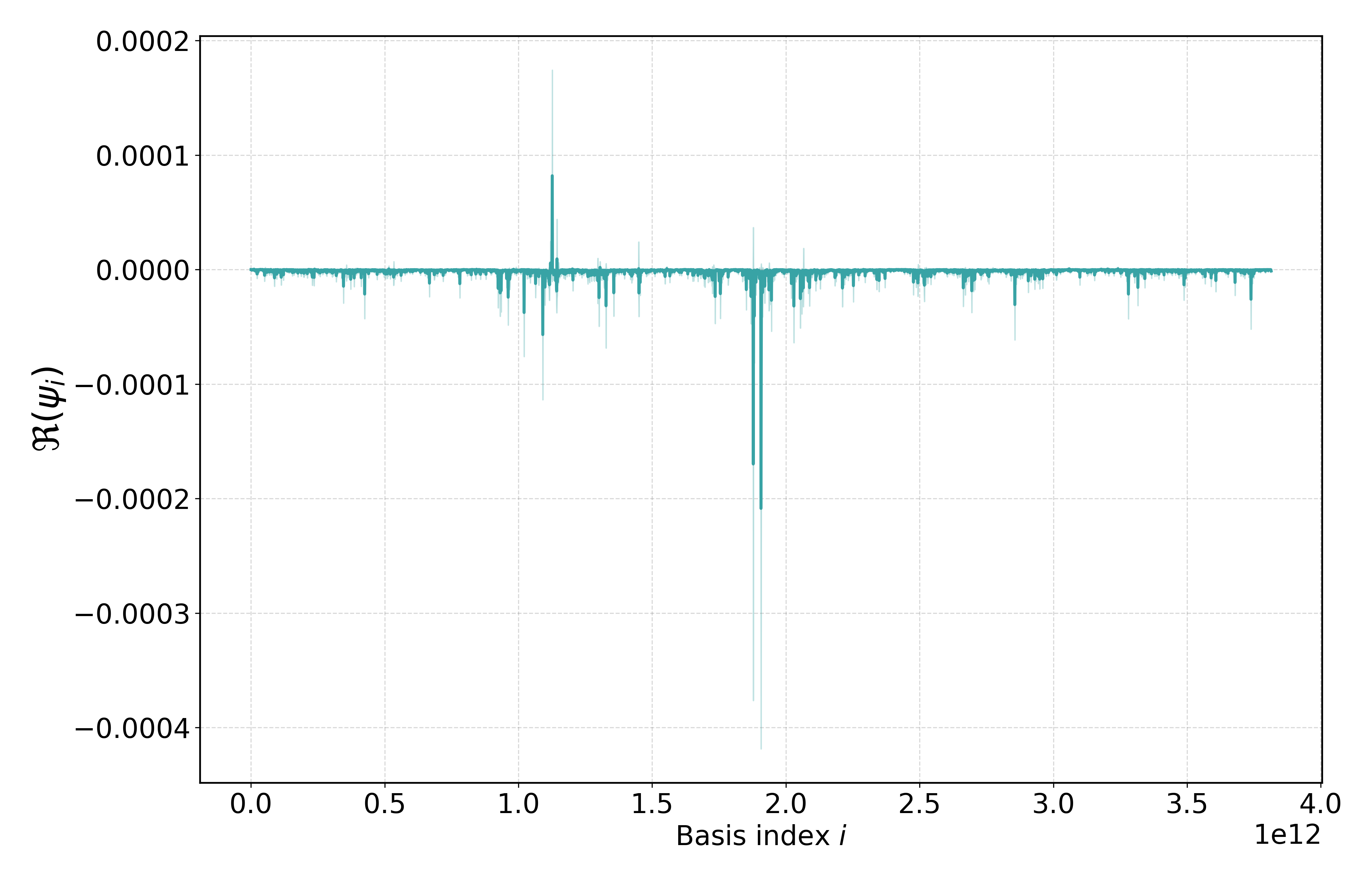}
        \caption{$m_\mathrm{max} = 2$}
        \label{fig:statevecalt:a}
    \end{subfigure}
    \hfill
    \begin{subfigure}[t]{0.49\textwidth}
        \centering
        \includegraphics[width=\textwidth]{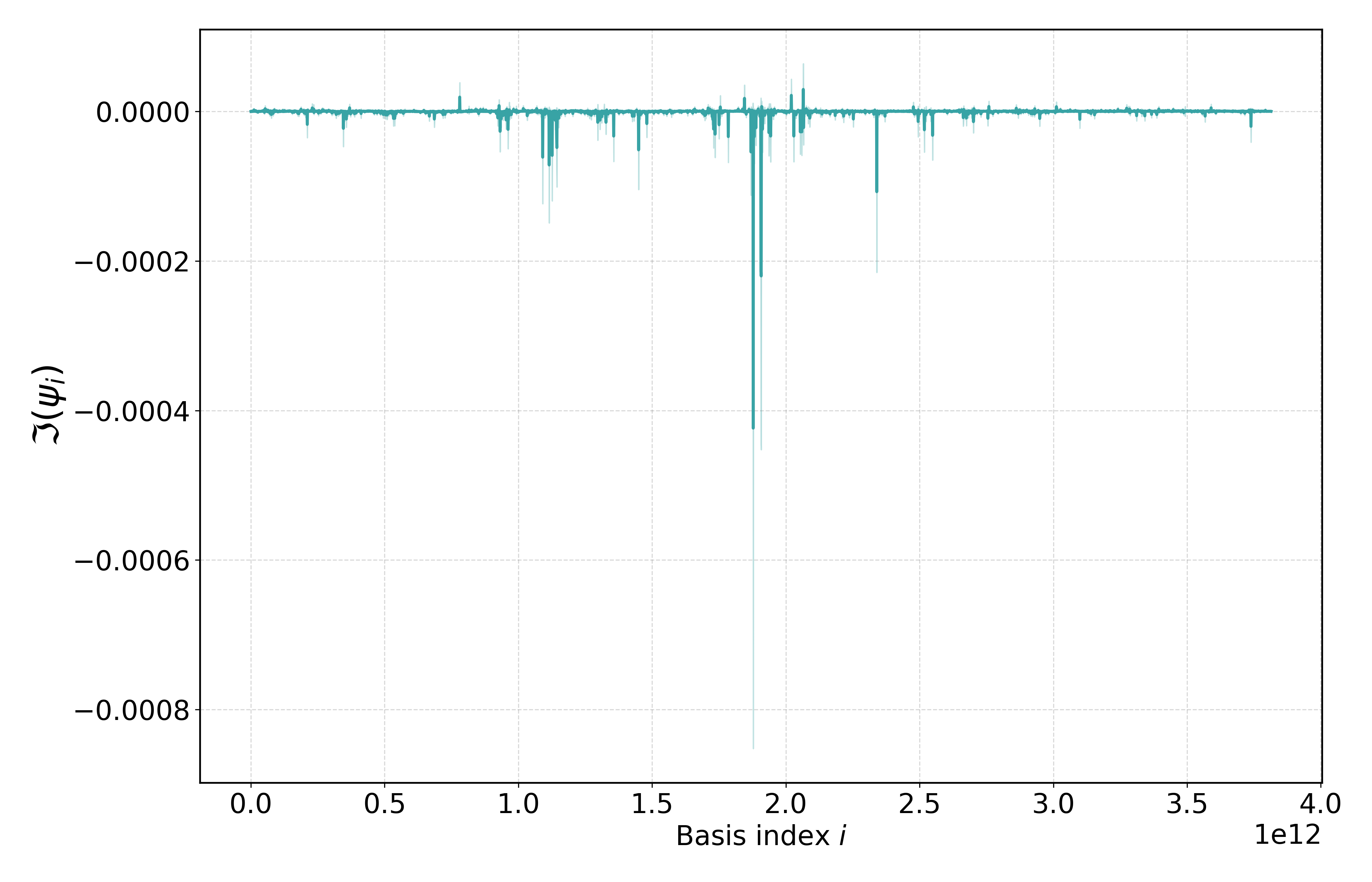}
        \caption{$ m_\mathrm{max} = 2$}
        \label{fig:statevecalt:b}
    \end{subfigure}

    \vspace{0.5cm}

    \begin{subfigure}[t]{0.49\textwidth}
        \centering
        \includegraphics[width=\textwidth]{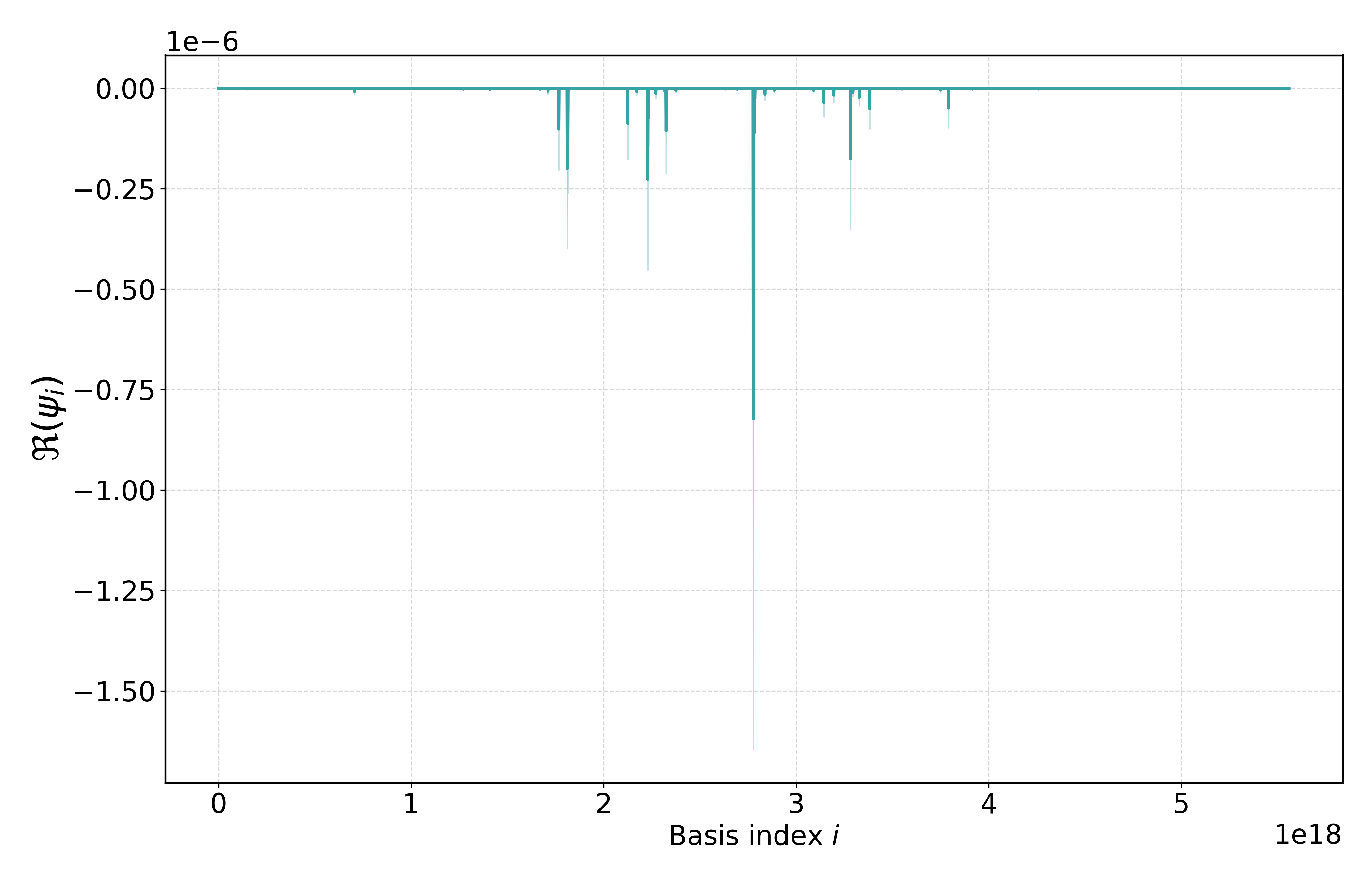}
        \caption{$m_\mathrm{max} = 5$}
        \label{fig:statevecalt:c}
    \end{subfigure}
    \hfill
    \begin{subfigure}[t]{0.49\textwidth}
        \centering
        \includegraphics[width=\textwidth]{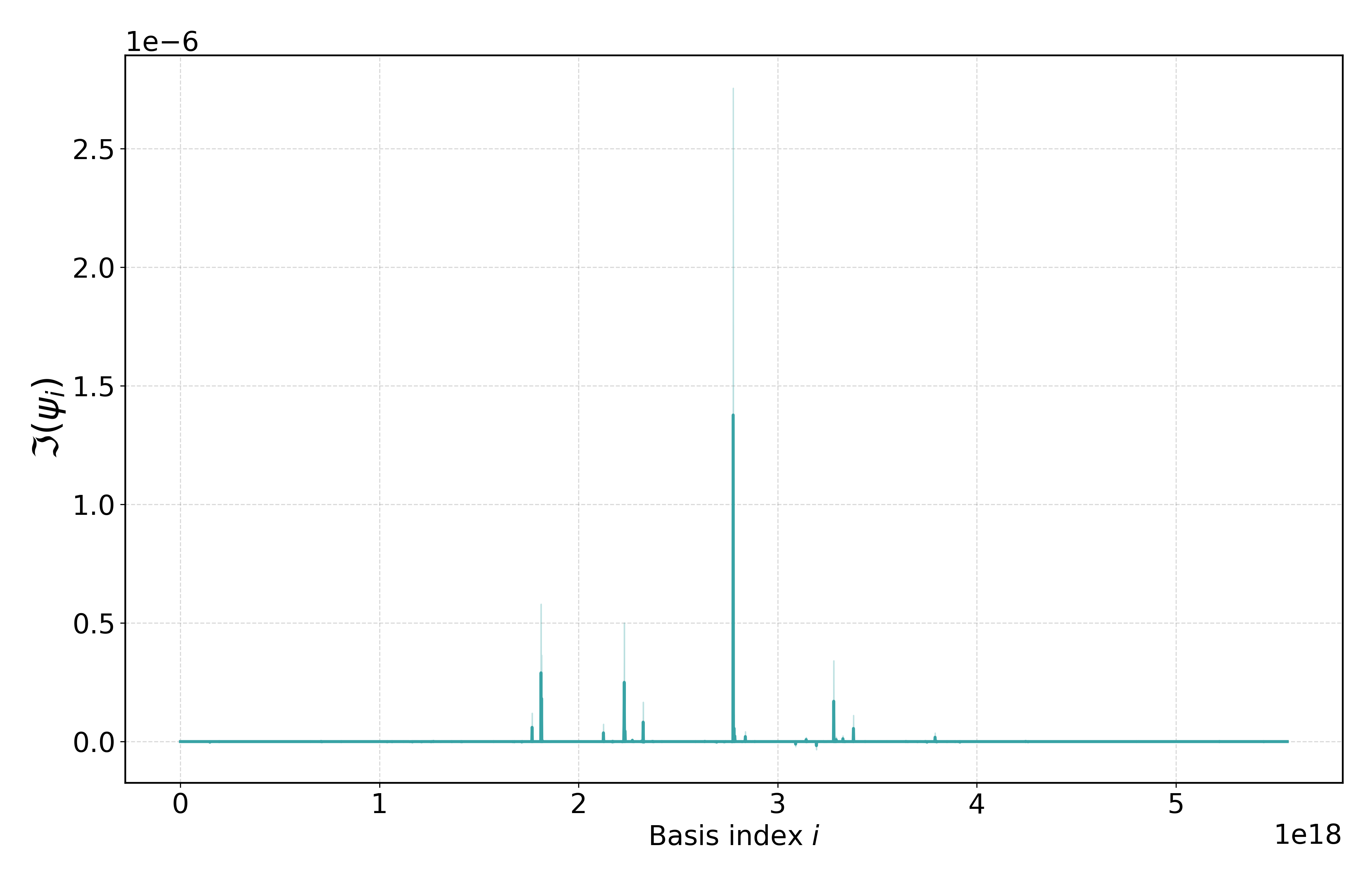}
        \caption{$m_\mathrm{max} = 5$}
        \label{fig:statevecalt:d}
    \end{subfigure}

    \caption{
        Stratified state-vector plots of representative variational near-kernel type-B solutions for the cutoffs $m_\mathrm{max} = 2, 5$
    }
    \label{fig:statevecalt}
\end{figure}
\newparagraph
As shown in Figure \ref{fig:statevecalt}, the type-B solutions exhibit a drastically different stratified state-vector profile when compared to the type-A counterpart. In particular, the envelope bands display clear concentration and a pronounced suppression of near-cutoff sectors under refinement. Taken together, this behaviour is consistent with compatibility with normalisability in the refinement limit. We note that the stratified sampling parameters used to produce these plots are summarised in Table \ref{tab:stratconfigs}.
\begin{table}[h]
\small
    \centering
    \caption{\small Stratified sampling parameters for the stratified state-vector plots for type-A and type-B solutions at $m_\mathrm{max} = 2, 5$. Here, $n_\mathrm{bins}$ denotes the number of bins, $n_\mathrm{spb}$ denotes the number of samples per bin, $n_s$ denotes the total number of samples used, $n_b$ denotes the total number of batches and $W_b$ denotes the width of each batch.}
    \begin{tabular}{c|c|ccccc}
        \toprule[1.2pt]
        Solution class &
        $m_\mathrm{max}$ &
        $n_\mathrm{bins}$ &
        $n_\mathrm{spb}$ &
        $n_s$ &
        $n_b$ &
        $W_b$
        \\
        \midrule

        \multirow{2}{*}{type-A}
           & 2 &  4096 & 4096 & 16,777,216 & 65,536 & 931,322,575 \\
           & 5 & 4096 & 16,384 & 67,108,864 & 65,536 & 1,357,401,687,864,315  \\

        \midrule

        \multirow{2}{*}{type-B}
           & 2 &  4096 & 4096 & 16,777,216 & 65,536 & 931,322,575  \\
           & 5 & 4096 & 16,384 & 67,108,864 & 65,536 & 1,357,401,687,864,315  \\

        \bottomrule[1.2pt]
    \end{tabular}
    \label{tab:stratconfigs}
\end{table}
\\\noindent
As shown in Table \ref{tab:stratconfigs}, the stratified state-vector plots are obtained from a comparatively small random subset of coefficients. In principle, increasing the number of samples per bin $K$ reduces sampling noise and drives the stratified estimator $\widehat{\|\tilde{\Psi}\|^2}$ closer to the exact squared norm, while increasing the number of bins $B$ refines the envelope resolution. In practice, however, we find that the qualitative envelope morphology relevant to our normalisability discussion (noise-like persistence versus concentration and near-cutoff suppression) stabilises already at the parameters reported in Table \ref{tab:stratconfigs}. Further increases in $K$ and/or $B$ predominantly improve the accuracy of the overall normalisation and reduce residual sampling fluctuations, without altering the qualitative features we focus on here. Accordingly, we do not pursue higher sampling budgets in the present analysis.
\newparagraph
The second diagnostic probes the same issue in a more physically transparent manner. Specifically, we analyse the edge-wise charge marginals (edge chromaticity). For each oriented edge $e \in E(\gamma)$ and charge vector $\vec{m}$, introduce the projector
\begin{equation}
    \hat{\mathcal{C}}_{\vec{m}}^{(e)} := \ket{\vec{m}_e} \bra{\vec{m}_e},
\end{equation}
and define the corresponding marginal distribution
\begin{equation}
    P_e(\vec{m}) := \frac{\langle \Psi, \hat{\mathcal{C}}_{\vec{m}}^{(e)} \Psi \rangle}{\|\Psi\|^2} = \sum_{\vec{m}_{\neq e}} \frac{|\psi(\vec{m}_e, \vec{m}_{\neq e})|^2}{\|\Psi\|^2},
\end{equation}
such that $\sum_{\vec{m}_e} P_e(\vec{m}_e) = 1$. The equation above simply computes the sum of the amplitudes of states which, on the given edge $e$, the charge vectors match the prescribed $\vec{m}_e$. If the variational state fails to concentrate as $m_\mathrm{max}$ increases and instead spreads nearly uniformly over newly available charge assignments, these marginals tend to flatten and approach the maximally delocalised benchmark
\begin{equation}
    P_e(\vec{m}_e) \approx \frac{1}{(2m_\mathrm{max} + 1)^3},
    \label{eq:flatchromexp}
\end{equation}
for all charge vectors $\vec{m}_e$ on all edges. This behaviour is directly related to the shell criterion. Namely, under such flattening, newly admitted near cutoff charge vectors do not become suppressed, and the probability for an edge to sit at outer charge values remains of order $1/m_\mathrm{max}$ rather than decaying rapidly and thus mirroring the lack of decay in $p_r$.
\newparagraph
To begin, we examine once again type-A solutions now at the cutoff $m_\mathrm{max} = 1$ where we compute the expectation value of $\hat{\mathcal{C}}_{\vec{m}}^{(e)}$ for each of the ten edges in the $K_5$ graph and for each of the 27 charge vectors at this cutoff. 
\begin{figure}[htbp]
    \centering

    \begin{subfigure}[t]{0.49\textwidth}
        \centering
        \includegraphics[width=\textwidth]{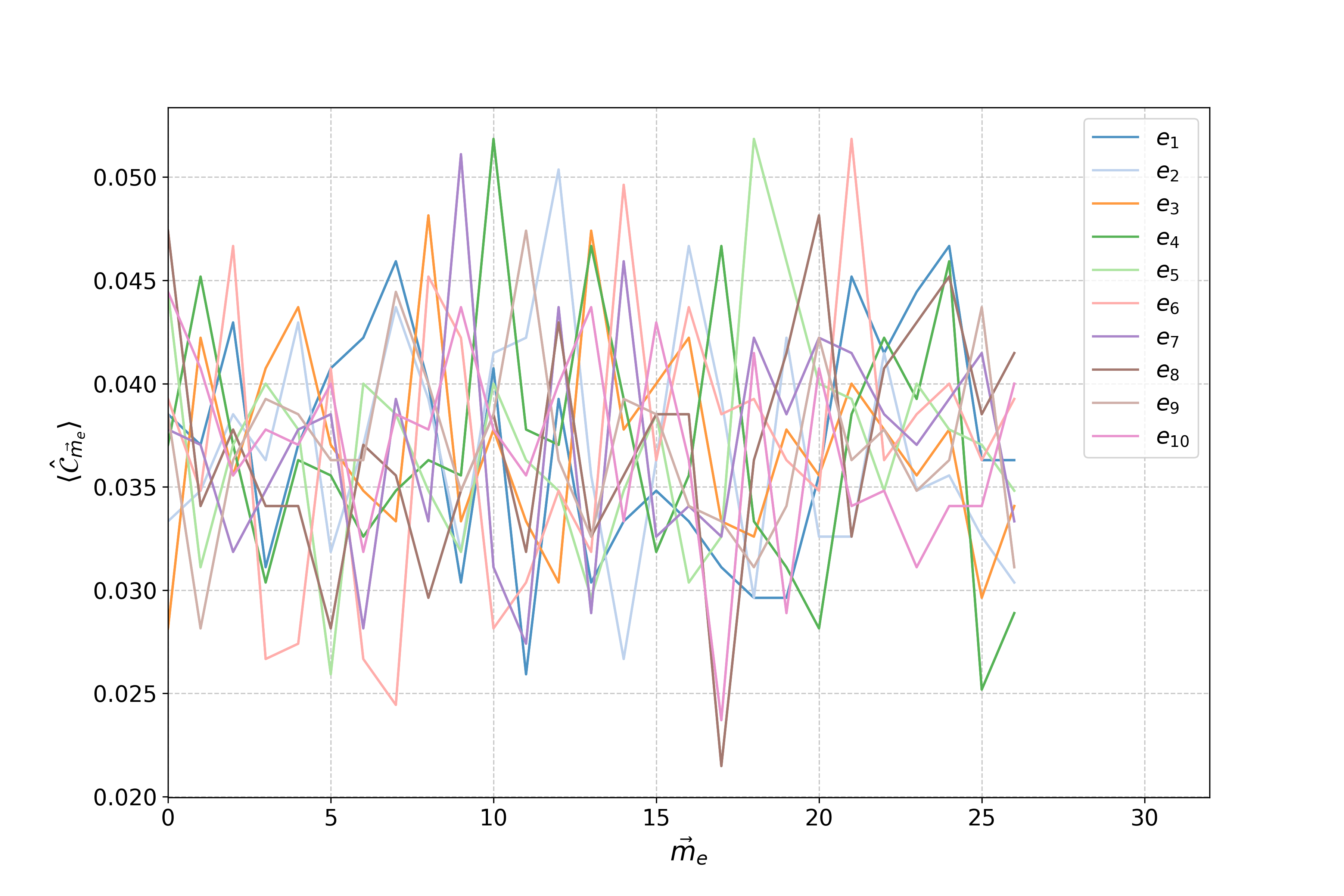}
        \caption{Monte Carlo}
        \label{fig:stdchrom:a}
    \end{subfigure}
    \hfill
    \begin{subfigure}[t]{0.49\textwidth}
        \centering
        \includegraphics[width=\textwidth]{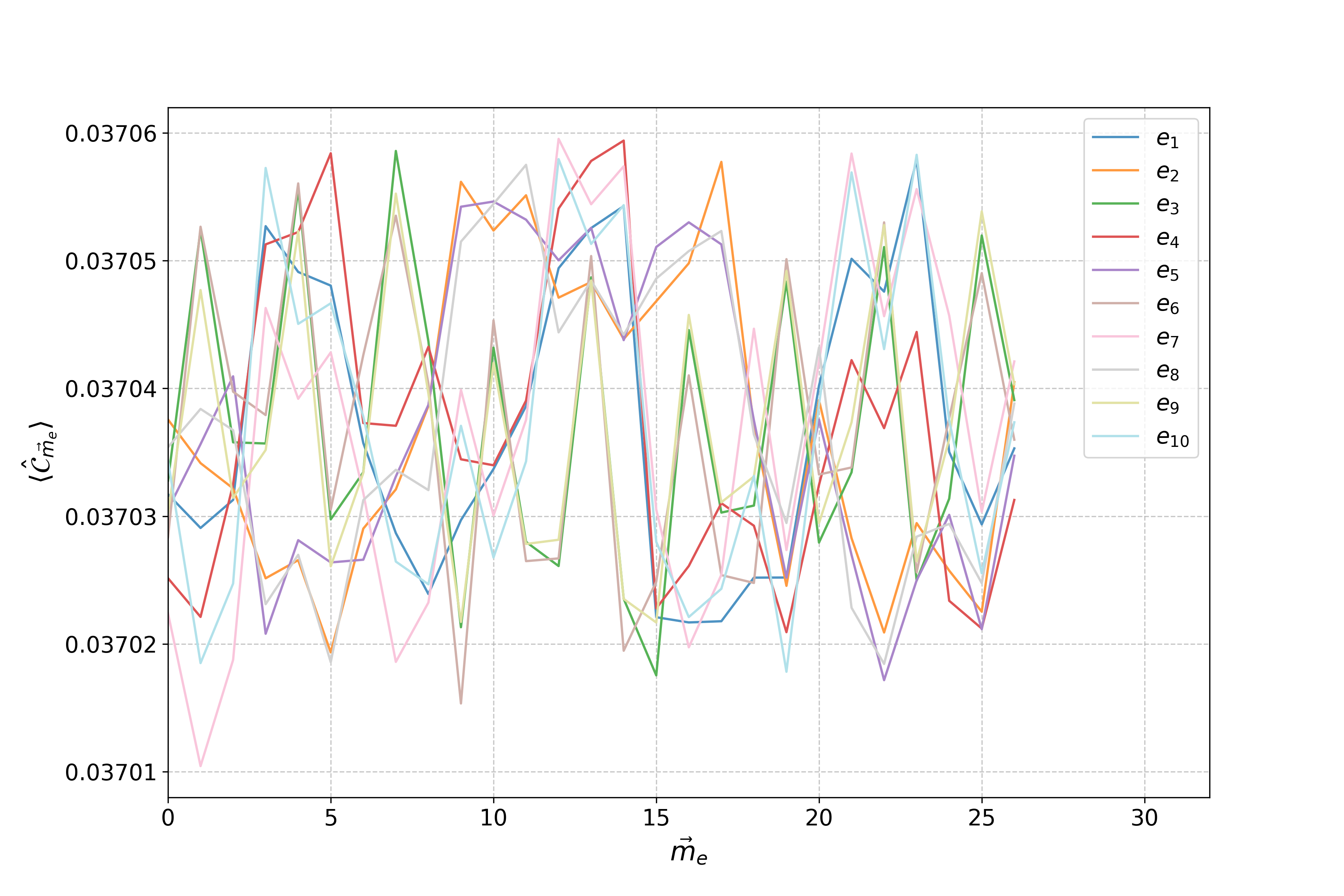}
        \caption{Full-state sum}
        \label{fig:stdchrom:b}
    \end{subfigure}

    \caption{
        Chromaticity plots for a representative type-A solution at $m_\mathrm{max} = 1$ computed using Monte Carlo methods on the left and using a full-state sum on the right.
    }
    \label{fig:stdchrom}
\end{figure}
\newparagraph
Figure \ref{fig:stdchrom} shows the chromaticity plots for a representative type-A solution at $m_\mathrm{max} = 1$. On the left (Figure \ref{fig:stdchrom:a}), the expectation values are plotted using Monte Carlo methods. As evident, one can see that, roughly, there is no characteristic profile for the chromaticity thus indicating that one obtains a state which is, in this context, not normalisable. To further eliminate any source of noise, since at this cutoff the Hilbert space dimension is manageable, we compute a full state sum of the expectation value which is shown in Figure \ref{fig:stdchrom:b}. In this case, one sees clearly that the Monte Carlo based estimate displays contributions inherent simply from not having enough samples. Contrastingly, the full-state sum shows that the chromaticity averages at around 0.0307035, which is almost the value expected for equation \eqref{eq:flatchromexp} for this cutoff (approx. $0,037\overline{037}$). This characteristic behaviour is observed for all simulations across all seeds and across all cutoffs. We note that while the tabulated results in Table \ref{tab:stdaltresults} show the cutoffs of $1, \ldots, 5$, \emph{type-A solutions obtained at cutoffs as high as $m_\mathrm{max} = 50$ and even $m_\mathrm{max} = 100$ exhibited identical behaviour.}
\newparagraph
The case of type-B solutions is drastically different. Namely, one observes a very characteristic chromaticity profile which is reoccurring across cutoffs. Specifically, a concentration on the all-zero charge vector is predominant in any $m_\mathrm{max}$ simulation, with distinct, but small, contributions from small-removed charge vectors\footnote{By small-removed charge vectors we mean charge vectors which differ from the all-zero charge vector by one charge being non-zero}.
\begin{figure}[htbp]
    \centering

    \begin{subfigure}[t]{0.49\textwidth}
        \centering
        \includegraphics[width=\textwidth]{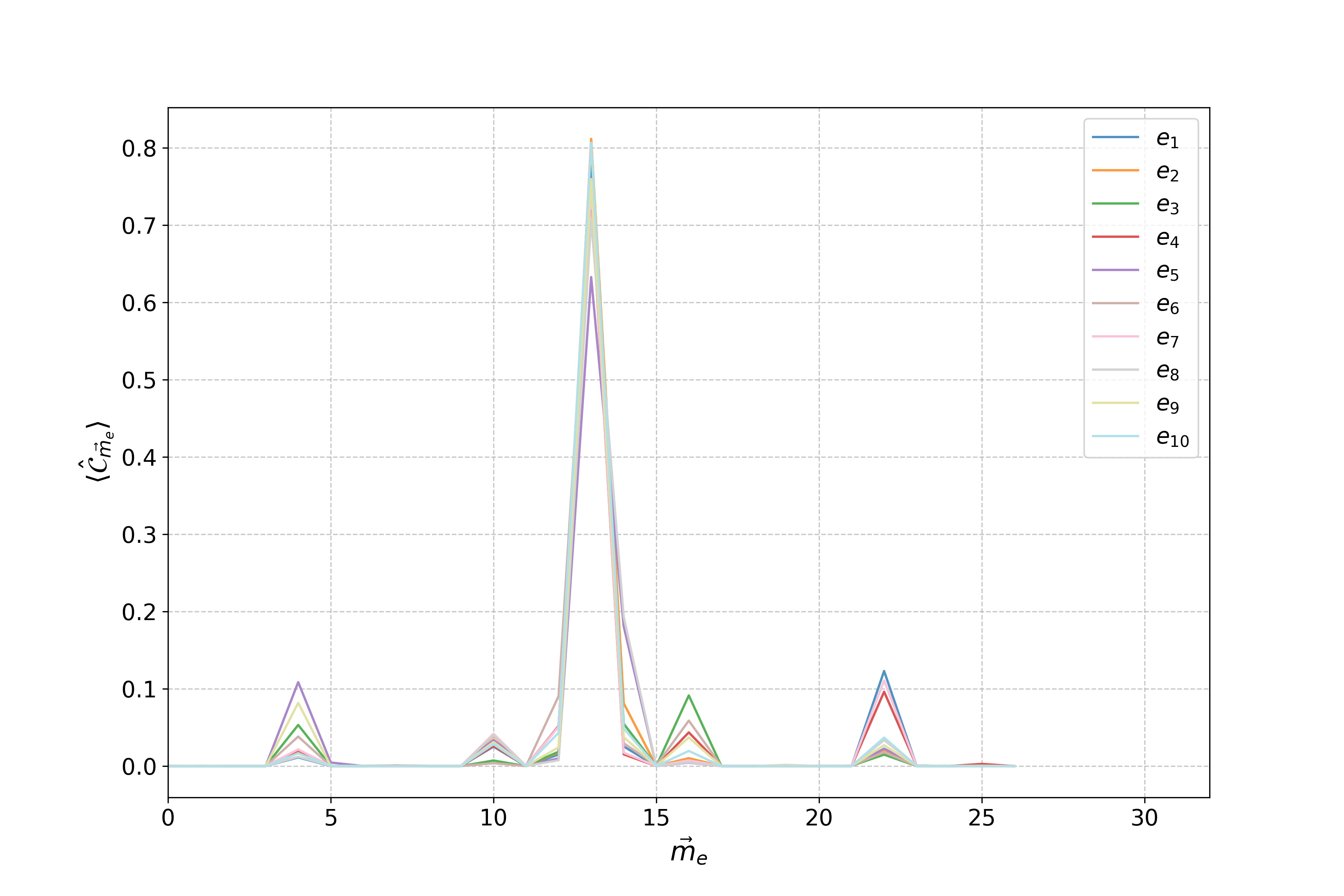}
        \caption{$m_\mathrm{max} = 1$}
        \label{fig:altchrom:a}
    \end{subfigure}
    \hfill
    \begin{subfigure}[t]{0.49\textwidth}
        \centering
        \includegraphics[width=\textwidth]{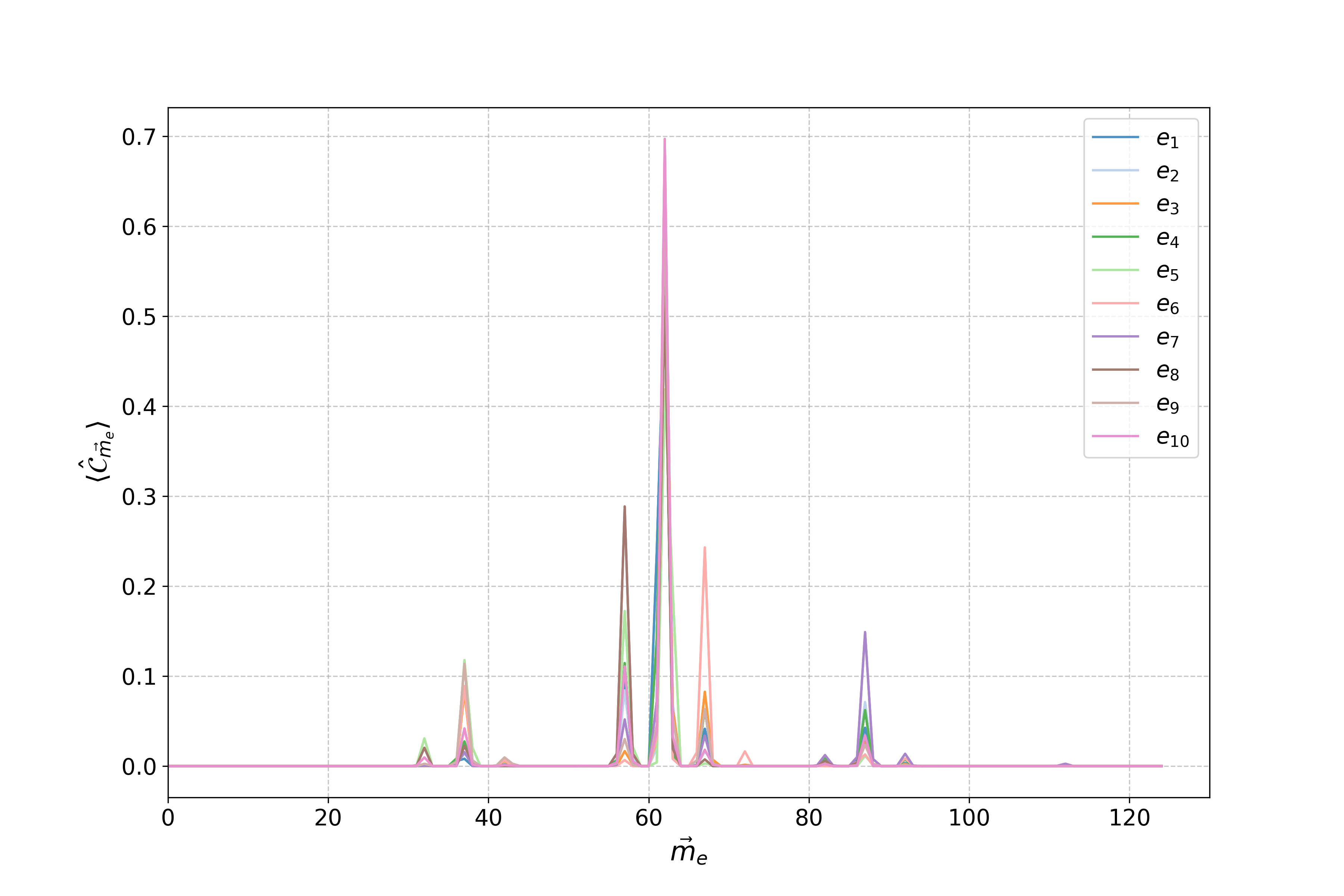}
        \caption{$ m_\mathrm{max} = 2$}
        \label{fig:altchrom:b}
    \end{subfigure}

    \vspace{0.5cm}

    \begin{subfigure}[t]{0.49\textwidth}
        \centering
        \includegraphics[width=\textwidth]{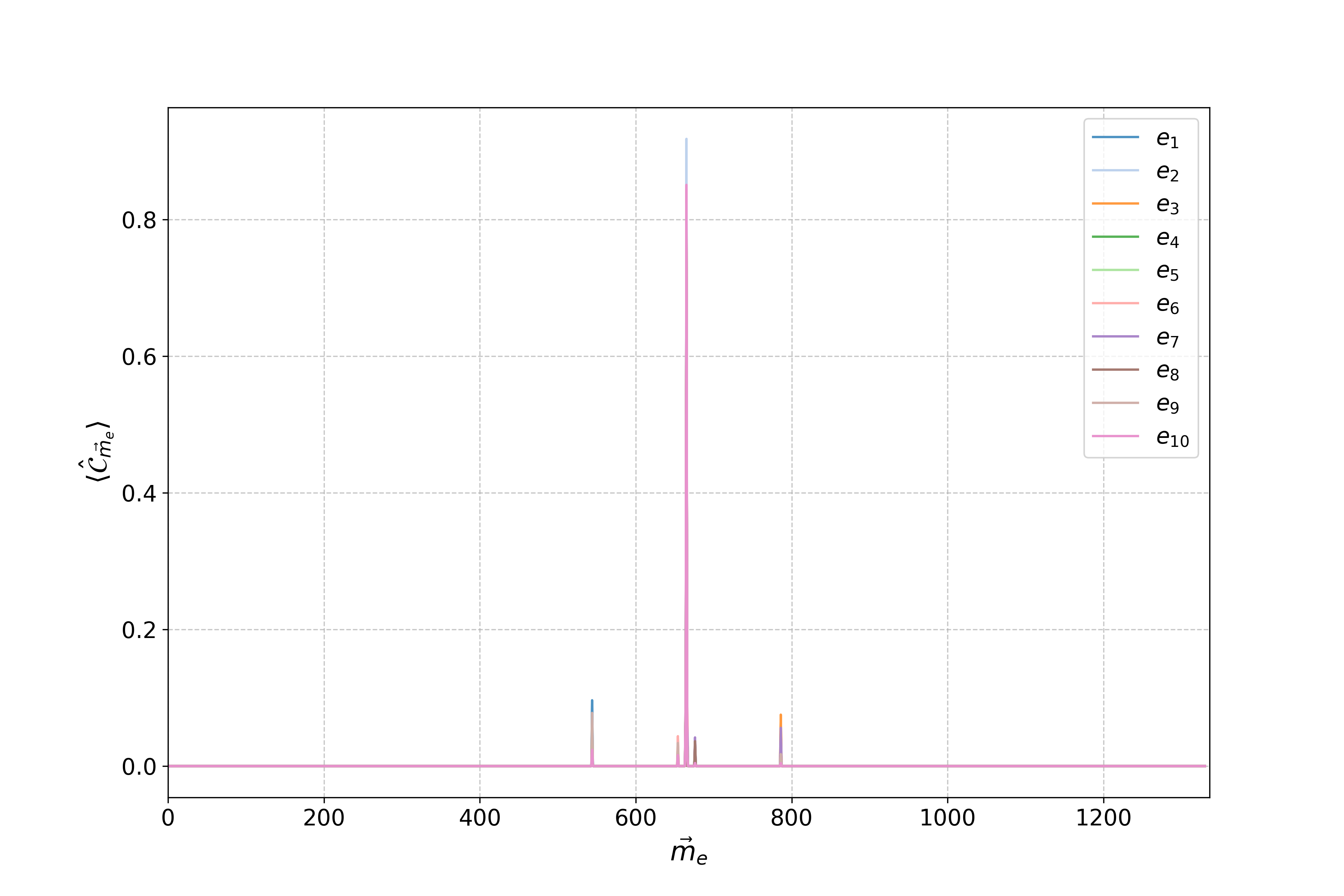}
        \caption{$m_\mathrm{max} = 5$}
        \label{fig:altchrom:c}
    \end{subfigure}
    \hfill
    \begin{subfigure}[t]{0.49\textwidth}
        \centering
        \includegraphics[width=\textwidth]{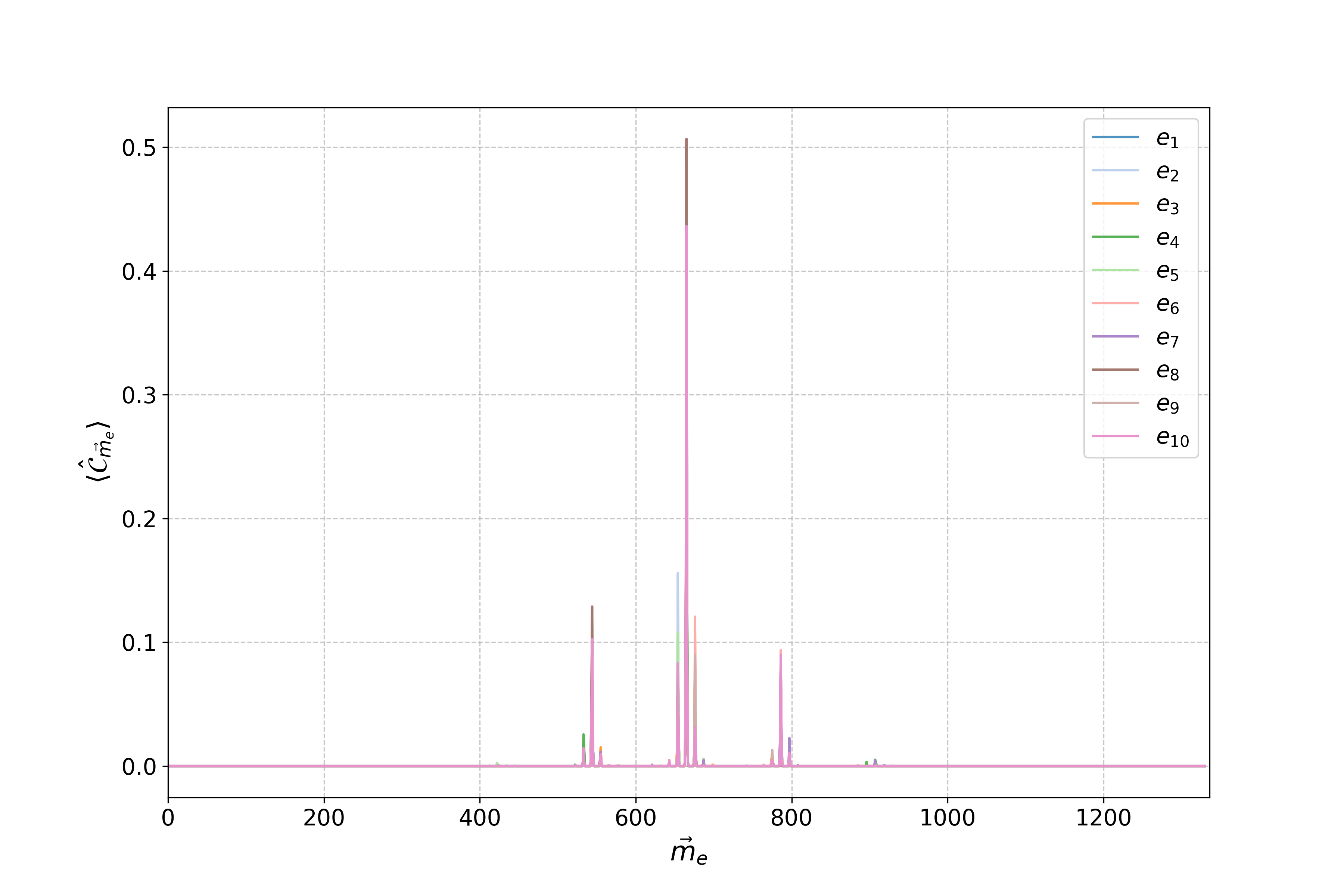}
        \caption{$m_\mathrm{max} = 5$}
        \label{fig:altchrom:d}
    \end{subfigure}

    \caption{
        Monte Carlo based chromaticity plots for representative type-B solutions for different cutoffs.
    }
    \label{fig:altchrom}
\end{figure}
\newparagraph
Figure \ref{fig:altchrom} shows Monte Carlo based chromaticity curves for representative type-B solutions for the cases of $m_\mathrm{max} = 1, 2$ in the top left and right respectively. Additionally, two $m_\mathrm{max} = 5$ chromaticity curves are shown at the bottom for two different solutions. As shown from all plots, the distinctive curve with the majority of the probability mass being concentrated on the $(0, 0, 0)$ charge vector is observed across all observed cutoffs. This is despite the fact that when the state-vectors are plotted using the stratified sampling, two solutions at the same cutoff although exhibiting different state-vector structure and thus indicating that their overlap is not substantial, they nevertheless have the same characteristic chromaticity curves as shown above. As also shown in the figure, contribution from near-cutoff charge vectors (concentrated on the left and right of each plot) are highly suppressed, if not zero, even as the cutoff is increased. This is a further evidence for the normalisability of the states obtained.

\subsubsection{Flatness and geometric observables}
\label{subsubsec:flatnessandobservables}
The variational smallness of the quadratic constraints establishes the proximity to the corresponding near-kernel subspaces but it does not, by itself, determine the geometric nature of the obtained states. This is particularly important in the present Thiemann regularised setting where the kernel of the Hamiltonian constraint is known to contain flat configurations as a distinguished subset of solutions. It is therefore natural to sharpen the characterisation by asking two additional questions. First, are the variational solutions supported on non-degenerate spatial geometries, as witnessed by a non-vanishing volume? Second, do they exhibit suppressed curvature excitations in the sense of (approximately) trivial holonomies around minimal loops? In what follows, we address these questions using two complementary sets of observables.
\newparagraph
To probe the curvature we use the minimal loop holonomy operators (as introduced in Section 3.1 of \cite{Sahlmann:2024kat}. Let $L(\gamma)$ denote the set of minimal loops $\alpha$ of the $K_5$ graph and let $\mathrm{tr} \hat{h}_\alpha$ be the corresponding loop quantised minimal holonomy operator. We monitor, for each $\alpha \in L(\gamma)$ both its expectation value $\langle \mathrm{tr}\hat{h}_\alpha \rangle_\mathrm{NQS}$ and its quantum fluctuation $(\Delta_\mathrm{NQS} \mathrm{tr}\hat{h}_\alpha)^2$.
\newparagraph
In the $\mathrm{U}(1)^3$ setting, flatness on a loop corresponds to the holonomy being component-wise trivial. Consequently, $\mathrm{tr} \hat{h}_\alpha$ attains its maximal value in the flat sector (equal to 3 in our normalisation) and suppressed fluctuations indicate that the state is sharply peaked on the flat value rather than merely averaging to it. As a second, global curvature diagnostic we additionally evaluate the BF-type flatness operator (cf. Section 3.1 of \cite{Sahlmann:2024kat})
\begin{equation}
    \hat{F} := 6|L(\gamma)|\mathbf{1} + \sum_{\alpha \in L(\gamma)} \mathrm{tr}[\hat{h}_\alpha + \hat{h}_\alpha^\dagger],
\end{equation}
and interpret $\langle \hat{f} \rangle_\mathrm{NQS} \approx 0$ as an aggregate indicator that the variational state satisfies, in expectation, the loop quantised analogue of a curvature (flatness) constraint. We emphasise that $\hat{F}$ is \emph{not} imposed during training but rather evaluated a posteriori and therefore provides and independent geometric probe that goes beyond the mere smallness of $\langle \hat{\mathcal{Q}}_i\rangle_\mathrm{NQS}$.
\newparagraph
Additionally, we compute for each vertex $v \in V(\gamma)$ the expectation value and fluctuation of the vertex volume operator $\hat{V}_v$. In the present truncation these quantities provide a direct proxy for degeneracy, as $\langle \hat{V}_v \rangle_\mathrm{NQS} \approx 0$ together with small fluctuations indicate that the probability mass of the variational state is concentrated on configurations that are predominantly volume-degenerate at $v$. 
\newparagraph
We begin with the type-A solutions whereby for this solution class, across all seeds and across all cutoffs, we find that for every $\alpha \in L(\gamma)$ the minimal loop holonomy expectations satisfy
\begin{equation}
    \langle \mathrm{tr} \hat{h}_\alpha \rangle_\mathrm{NQS} \approx 3 \quad,\quad (\Delta_\mathrm{NQS} \mathrm{tr}\hat{h}_\alpha)^2 \,\text{ strongly suppressed},
\end{equation}
up to Monte Carlo error bars. Consistently, the BF-type operator satisfies $\langle \hat{F} \rangle_\mathrm{NQS} \approx 0$ within the same statistical tolerance. Taken together, these observations provide strong evidence that the type-A solutions are \emph{strictly flat on all minimal loops} of $\gamma$, rather than merely exhibiting accidental cancellations in the quadratic constraint expectation value. We note that the same effect was observed in solutions obtained at $m_\mathrm{max} = 50$ and $m_\mathrm{max} = 100$ as well.
\newparagraph
The situation is distinctly different for the type-B solutions. In this case, for all cutoffs we observe minimal loop holonomy expectations being significantly displaced from the flat value, accompanied by fluctuations of comparable magnitude. Additionally, the curvature operator $\hat{F}$ is never driven close to zero. Thus, although type-B solution solve the alternative ordered quadratic constraint to (sometimes) near machine precision, they do \emph{not} correspond to flat configurations on the minimal cycles of the graph. 
\newparagraph
In combination with the chromaticity analysis of Section \ref{subsubsec:normalisability}, this establishes a sharp geometric distinction between the two ordering-induced solution classes. Type-A solutions are generically flat whereas type-B solutions represent non-flat curvature carrying states (and, as we will see below, typically volume-degenerate). This ordering dependence is important because it shows that solving the constraint in the variational sense is, at fixed graph and cutoff, not yet sufficient to isolate a unique or common geometric sector and different orderings can \emph{generically} admit qualitatively different near-kernel families one of which is driven into the flat sector, while the other supports curvature excitations. 
\newparagraph
Next, we examine the vertex volume operator as discussed prior. A priori, $\langle \hat{\mathcal{Q}}_i \rangle_\mathrm{NQS}$ being small does not fix these quantities as Thiemann-type operators are sensitive to volume-degenerate sectors and different orderings may (variationally) prefer states with very different volume support. Empirically, the two solution classes once again exhibit sharply distinct volume behaviour. For type-A solutions, we \emph{never} observe a collapse to the volume kernel. Instead, for all seeds and all cutoffs studied, the vertex volume expectations satisfy $\langle \hat{V}_v \rangle_\mathrm{NQS} > 0$ at \emph{every} vertex and in most runs, these values are approximately vertex-independent (i.e. $\langle \hat{V}_{v_1} \rangle_\mathrm{NQS} \approx \cdots \approx \langle \hat{V}_{v_5} \rangle_\mathrm{NQS}$ up to Monte Carlo uncertainty). Occasional mild asymmetry, typically in the form of one vertex exhibiting noticeably lower $\langle \hat{V}_v \rangle_\mathrm{NQS}$ is observed and is consistent with the fact that the volume depends explicitly on the embedding dependent orientation signs $\epsilon_v(e_i, e_j, e_k)$ at $v$. While the $K_5$ graph is combinatorially symmetric, the deterministic generic embeddings fixes a generic pattern of sign factors which need not be identical on all vertices and thus, $\hat{V}_v$ need not be related by an exact symmetry across $v$. In addition, we find that type-A volume fluctuations are large, often satisfying
\begin{equation}
    (\Delta_\mathrm{NQS} \hat{V}_v)^2 \gtrsim \langle \hat{V}_v \rangle_\mathrm{NQS},
\end{equation}
indicating that these states are \emph{not} sharply peaked on a semi-classical volume value but rather have broad support across volume eigen-sectors. Finally, it is also observed that volume expectation values systematically increase with the cutoff. This is expected qualitatively since enlarging $m_\mathrm{max}$ enlarges the accessible charge excitations and hence the typical scale of flux-like quantities entering $\hat{V}_v$ so neither the volume matrix elements nor the variationally realised volume support are expected to remain cutoff-independent without an additional renormalisation prescription. 
\newparagraph
By contrast, type-B solutions display an extreme and remarkably stable behaviour. Namely, for \emph{all} seeds and \emph{all} cutoffs, we find that
\begin{equation}
    \langle \hat{V}_v \rangle_\mathrm{NQS} \approx 0 \quad,\quad (\Delta_\mathrm{NQS} \hat{V}_v)^2 \approx 0,
\end{equation}
for all $v \in V(\gamma)$, within Monte Carlo tolerance. This indicates that the obtained type-B near-kernel states always lie in the kernel of $\hat{V}_v$. Moreover, $(\Delta_\mathrm{NQS} \hat{V}_v)^2 \approx 0$ indicates that the states are sharply supported on the \emph{zero} eigenvalue rather than merely having cancellations among different volume sectors. Importantly, this conclusion is independent of whether a particular Markov chain was initialised in an explicitly zero-volume configuration.
\newparagraph
Taken together, the minimal loop holonomy observables and the volume diagnostics reveal that the two ordering-induced near-kernel families occupy \emph{different} geometric sectors. Type-A solutions are consistently flat on all minimal loops while simultaneously exhibiting non-vanishing volume at every vertex, albeit with large relative fluctuations. Type-B solutions, on the other hand, are not flat yet they are exactly volume degenerate in the strong sense that they lie in $\bigcap_v \ker(\hat{V}_v)$. This ordering dependence is physically consequential, as it demonstrates that variational proximity to $\ker(\hat{\mathcal{Q}}_{\hat{H}})$ versus $\ker(\hat{\mathcal{Q}_2})$ does not merely change the numerical value of the quadratic constraint but can qualitatively select between flat, non-zero volume family and a curvature carrying zero volume family of solutions. In particular, it highlights that such additional geometric diagnostics are indispensable for interpreting the physical content of near-kernel NQS solutions and it motivates the correlation analysis of the next section as a further probe of such distinct geometric sectors.

\subsubsection{Long-range correlation}
\label{subsubsec:correlations}
On a finite graph, it is entirely possible to drive $\langle \hat{\mathcal{Q}}_i \rangle$ close to zero with states that are (i) essentially random in the charge-network basis (high entropy, almost maximally mixed on small subsystems), (ii) strongly constrained only locally (e.g. by Gauß invariance at vertices) or (iii) truly correlated across distant edges in a way that cannot be reduced to purely local constraint. Thereofre, the next characterisation test we conduct is constructed to probe whether the obtained near-kernel solutions exhibit \emph{non-local} structure on the graph. In the present setting, a natural probe of such non-locality is provided by edge-charge (chromaticity) correlations. The operators
\begin{equation}
    \hat{\mathcal{C}}_{\vec{m}}^{(e)} = \ket{\vec{m}_e} \bra{\vec{m}_e}
\end{equation}
project onto a fixed $\mathrm{U}(1)^3$ charge vector $\vec{m}$ on a given oriented edge $e \in E(\gamma)$. They therefore encode the most local information available in the charge-network basis, namely whether a particular edge carries a particular charge. If a variational solution class possesses non-trivial graph-scale organisation, one expects it to become apparent already at the level of joint statistics of these local projectors.
\newparagraph
Concretely, for two edges $e_1, e_2 \in E(\gamma)$ and two charge vectors $\vec{m}_1, \vec{m_2}$, we define the connected 2-point function
\begin{equation}
    G_{\vec{m}_1, \vec{m}_2}(e_1, e_2) := \langle \hat{\mathcal{C}}_{\vec{m}_1}^{(e_1)} \hat{\mathcal{C}}_{\vec{m}_2}^{(e_2)}\rangle - \langle \hat{\mathcal{C}}_{\vec{m}_1}^{(e_1)} \rangle \langle \hat{\mathcal{C}}_{\vec{m}_2}^{(e_2)} \rangle.
    \label{eq:chromaticity2pointfluctuations}
\end{equation}
This is precisely the covariance of the two indicator events ``edge $e_1$ carries $\vec{m}_1$'' and ``edge $e_2$ carries $\vec{m}_2$'', and thus isolates correlation content beyond what is already contained in the 1-point marginals.
\newparagraph
A key reason to focus on equation \eqref{eq:chromaticity2pointfluctuations} (rather than the raw $\langle \hat{\mathcal{C}}_{\vec{m}_1}^{(e_1)} \hat{\mathcal{C}}_{\vec{m}_2}^{(e_2)} \rangle$) is that the latter is dominated by the product of 1-point probabilities and can therefore seem structured even in completely uncorrelated states. In contrast, a genuine long-range effect must survive the subtraction in \eqref{eq:chromaticity2pointfluctuations}. Additionally, since $\hat{\mathcal{C}}_{\vec{m}}^{(e)}$ is a projector, one has the exact local fluctuation identity
\begin{equation}
    (\Delta \hat{\mathcal{C}}_{\vec{m}}^{(e)})^2 =  \langle (\hat{\mathcal{C}}_{\vec{m}}^{(e)})^2 \rangle - \langle \hat{\mathcal{C}}_{\vec{m}}^{(e)} \rangle^2 = \langle \hat{\mathcal{C}}_{\vec{m}}^{(e)} \rangle (1 - \langle \hat{\mathcal{C}}_{\vec{m}}^{(e)} \rangle),
    \label{eq:chromaticitydelta}
\end{equation}
which sets the intrinsic scale against which any connected correlations should be compared.
\newparagraph
We begin with the type-A family. Empirically (cf. the chromaticity discussion in Section \ref{subsubsec:normalisability}), these solutions exhibit a pronounced \emph{delocalisation} in charge space, whereby edge-wise marginals are close to the maximally mixed benchmark
\begin{equation}
    P_e(\vec{m}) = \langle \hat{\mathcal{C}}_{\vec{m}}^{(e)} \rangle \approx \frac{1}{(2m_\mathrm{max} + 1)^3},
    \label{eq:empiricaldelocalisation}
\end{equation}
with weak dependence on $e$ and on $\vec{m}$ (within Monte Carlo tolerance). For example, at $m_\mathrm{max} = 2$, one has $P_e(\vec{m}) \approx 8 \times 10^{-3}$, and consequently products of 1-point probabilities are naturally of order $P_{e_1}(\vec{m}_1) P_{e_2}(\vec{m}_2) \sim 6 \times 10^{-5}$. In what follows, we present a precise theoretical expectation for the connected correlator \eqref{eq:chromaticity2pointfluctuations} in the type-A class.
\begin{itemize}
    \item[(i)] \emph{Exact same-edge identities.} For $e_1 = e_2 = e$, the projector algebra fixes $G_{\vec{m}_1, \vec{m}_2}(e, e)$ essentially kinematically. If $\vec{m}_1 \neq \vec{m}_2$, then $\hat{\mathcal{C}}_{\vec{m}_1}^{(e)} \hat{\mathcal{C}}_{\vec{m}_2}^{(e)} = 0$ and hence,
    \begin{equation}
        G_{\vec{m}_1, \vec{m}_2}(e, e) = -P_e(\vec{m}_1) P_e(\vec{m}_2) \qquad (\vec{m}_1 \neq \vec{m}_2).
    \end{equation}
    If $\vec{m}_1 = \vec{m}_2 = \vec{m}$, then $(\hat{\mathcal{C}}_{\vec{m}}^{(e)})^2 = \hat{\mathcal{C}}_{\vec{m}}^{(e)}$ and therefore,
    \begin{equation}
        G_{\vec{m}_1, \vec{m}_2}(e, e) = P_e(\vec{m})(1 - P_e(\vec{m})),
    \end{equation}
    which is simply equation \eqref{eq:chromaticitydelta}. These relations must hold for \emph{any} state (type-A or type-B), and thus provide a stringent implementation-level check in Monte Carlo settings.

    \item[(ii)] \emph{Off-diagonal factorisation in the maximally mixed hypothesis.} The type-A empirical delocalisation \eqref{eq:empiricaldelocalisation} motivates the working hypothesis that the type-A solutions are, at least on small subsystems, close to a \emph{microcanonical} (high entropy) distribution over the gauge invariant charge configurations compatible with the truncation. In that regime, reduced density matrices on \emph{few} edges are expected to be close to maximally mixed (up to the local Gauß constraints at the vertices they touch), implying that for distinct edges $e_1 \neq e_2$, 
    \begin{equation}
        \langle \hat{\mathcal{C}}_{\vec{m}_1}^{(e_1)} \hat{\mathcal{C}}_{\vec{m}_2}^{(e_2)}\rangle \approx P_{e_1}(\vec{m}_1) P_{e_2}(\vec{m}_2) + \varepsilon_{\vec{m}_1, \vec{m}_2}(e_1, e_2),
        \label{eq:maxhypcorr}
    \end{equation}
    where $\varepsilon$ captures any constraint-induced correlations and any genuine dynamical correlations. For a purely factorised (product) state one would have $\varepsilon \equiv 0$ identically while for gauge invariant ensembles $\varepsilon$ is expected to be non-zero primarily when $e_1$ and $e_2$ share a vertex (since the local Gauß constraints couple incident edges) and to be strongly suppressed for disjoint edges. In either case, \eqref{eq:maxhypcorr} implies that the connected correlator should satisfy
    \begin{equation}
        G_{\vec{m}_1, \vec{m}_2}(e_1, e_2) \approx \varepsilon_{\vec{m}_1, \vec{m}_2}(e_1, e_2), \qquad e_1 \neq e_2.
        \label{eq:doestypeacorrelate}
    \end{equation}
    That is, it should be \emph{parametrically small} compared to the raw 2-point function and, crucially, it should \emph{not} display a stable large-scale pattern that persists under increasing Monte Carlo statistics.
\end{itemize}
We now confront the expectation \eqref{eq:doestypeacorrelate} with explicit Monte Carlo estimates of \eqref{eq:chromaticity2pointfluctuations} on the $K_5$ graph at $m_\mathrm{max} = 2$ as an example. We focus on some representative charge vectors
\begin{equation}
    \vec{m}_1 = (-2, -2, 0) \quad,\quad \vec{m}_2 = (1, -2, -1),
\end{equation}
and evaluate $G_{\vec{m}_1, \vec{m}_2}(e_1, e_2)$ for all ordered edge pairs $(e_1, e_2)$. Figure \ref{fig:typeA_connected_progressive} shows the resulting $10 \times 10$ correlation matrices for progressively increasing number of Monte Carlo samples, ranging from $N_S = 1350$ and up to $N_S = 225,000$.
\begin{figure}[ht]
  \centering
  \includegraphics[width=0.32\textwidth]{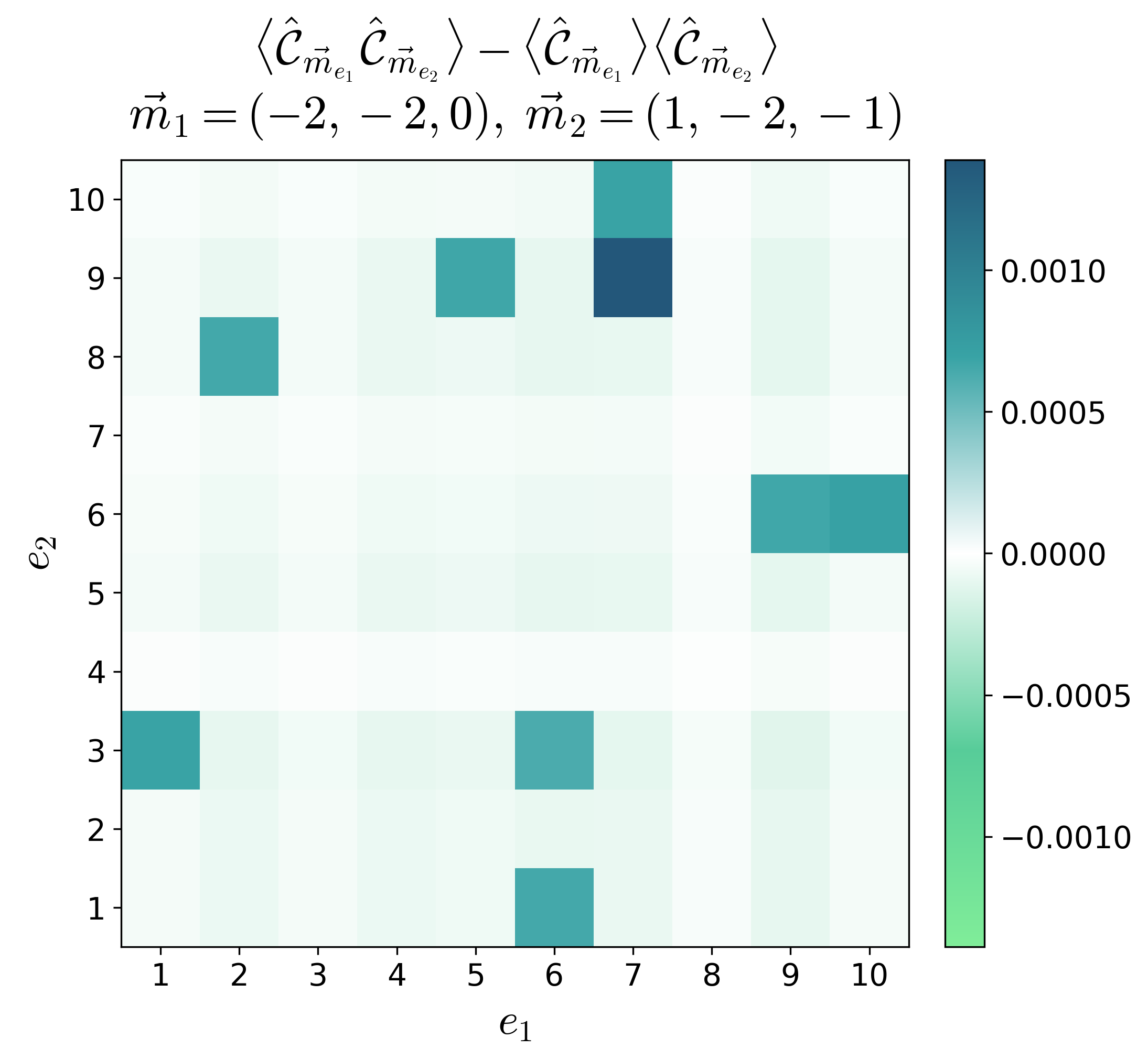}
  \hfill
  \includegraphics[width=0.32\textwidth]{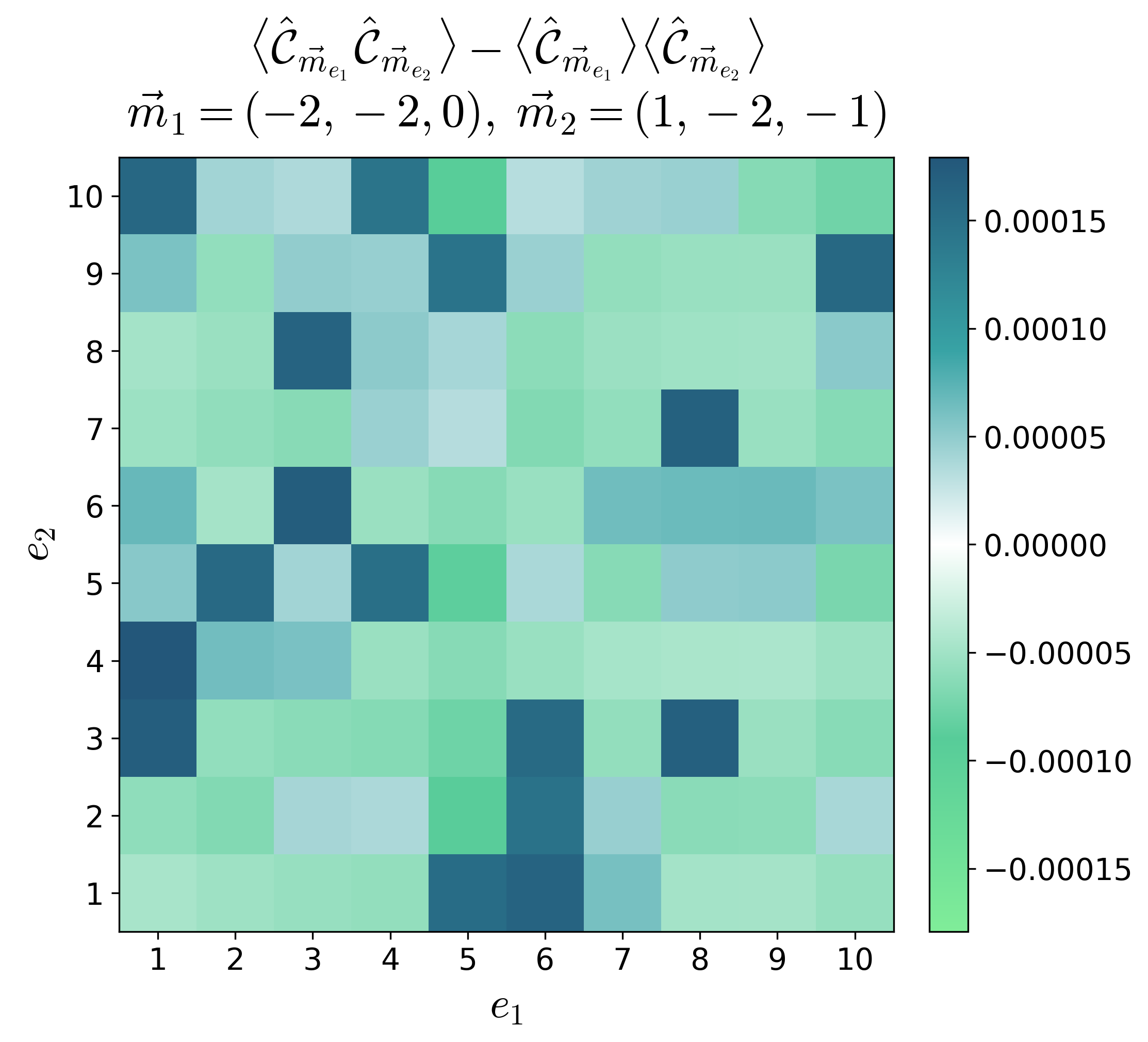}
  \hfill
  \includegraphics[width=0.32\textwidth]{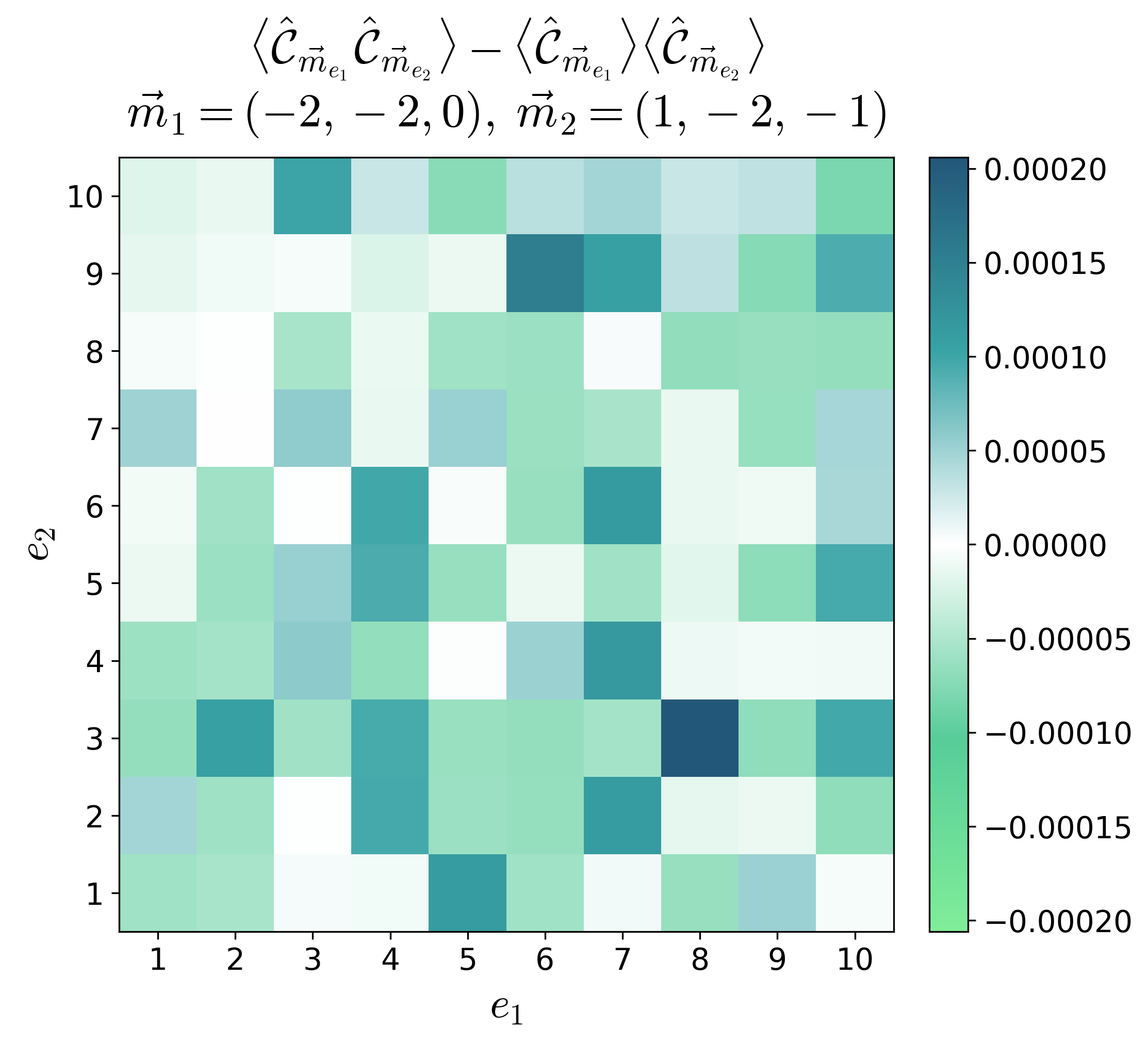}
  \\[2mm]
  \hspace{45pt}
  \includegraphics[width=0.36\textwidth]{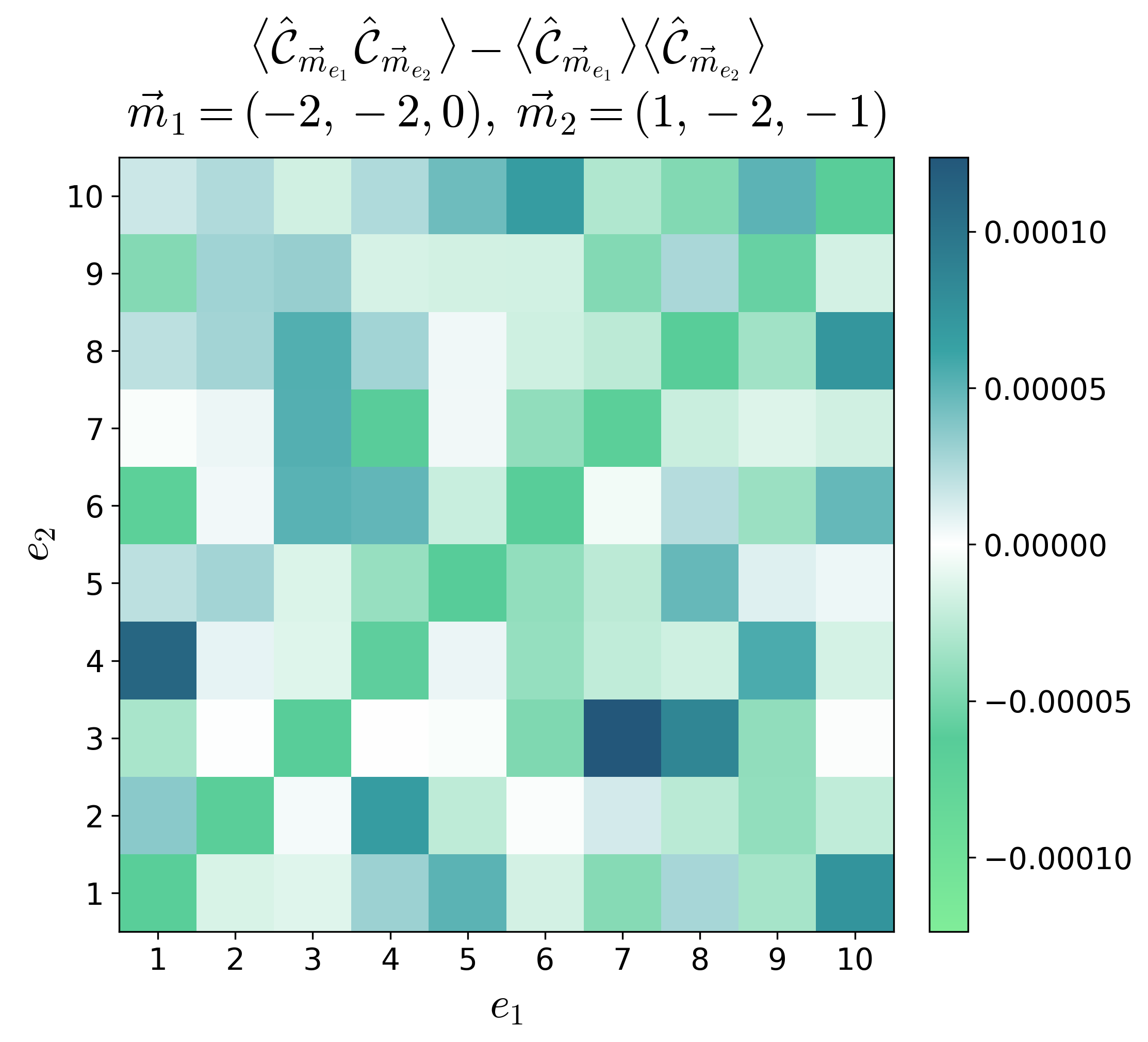}
  \hfill
  \includegraphics[width=0.32\textwidth]{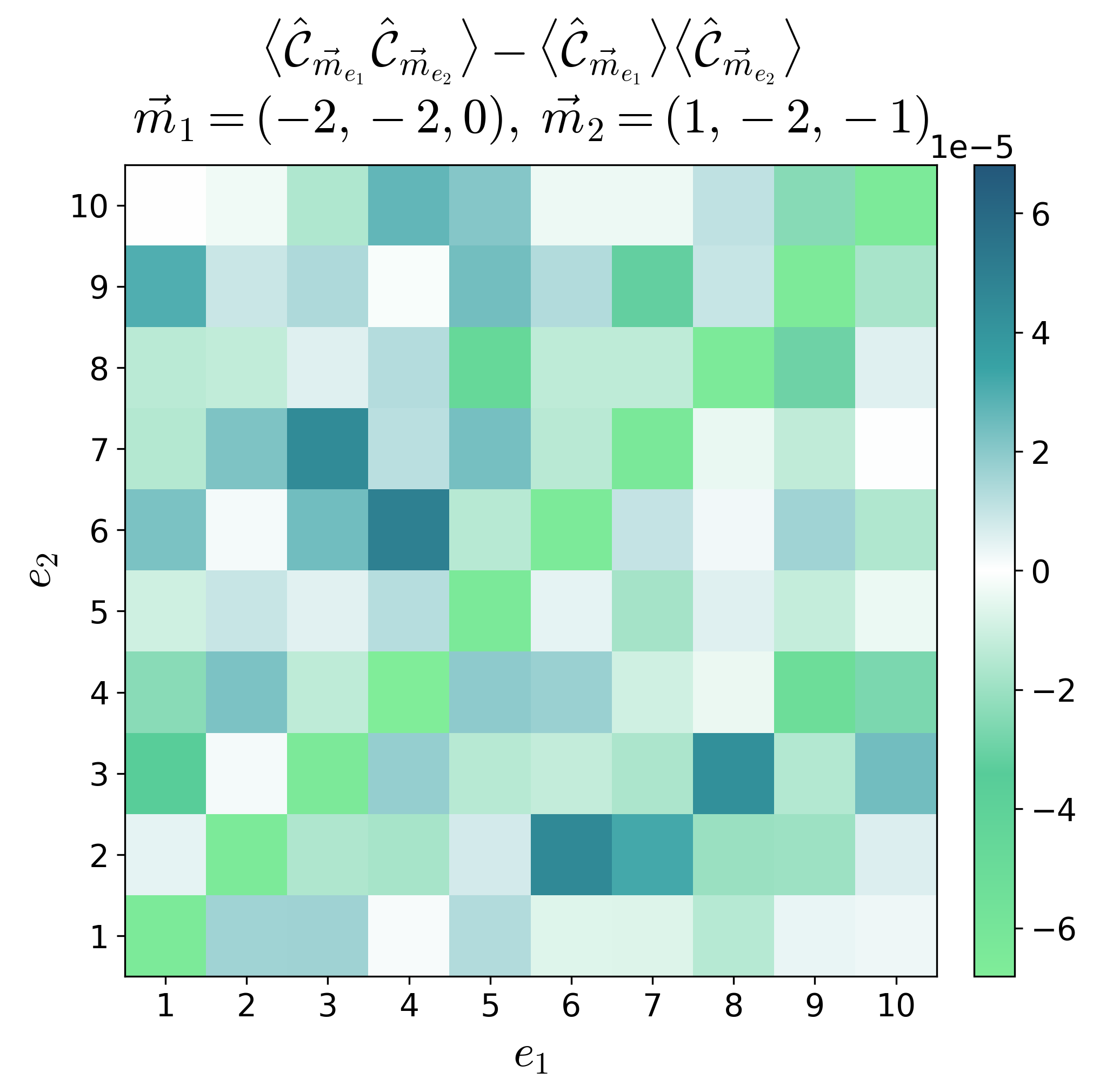}
  \hspace{55pt}
  \hfill
  \caption{
  Type-A (standard ordering) connected chromaticity 2-point function
  $G_{\vec m_1,\vec m_2}(e_1,e_2)$ at cutoff $m_{\max}=2$ for $\vec m_1=(-2,-2,0)$ and $\vec m_2=(1,-2,-1)$ is shown for increasing Monte Carlo sample sizes $N_S$.
  The colour scale is symmetric about zero in each panel. Reading from the top left, the $N_S$ is 1350, 9000 and 18,000 for the top row and 45,000 and 225,000 for the bottom left and right figures respectively.
  }
  \label{fig:typeA_connected_progressive}
\end{figure}
\newparagraph
Qualitatively from the Figure \ref{fig:typeA_connected_progressive}, two features are immediate. First, the matrices do not exhibit a robust spatial pattern that stabilises with increasing $N_S$. Rather, the dominant visible structure changes as statistics are increased. Second, the overall amplitude decreases substantially with $N_S$, with typical entries moving from the $10^{-4}$ - $10^{-5}$ range toward the few $10^{-5}$ range at $N_S = 225,000$. This is precisely the behaviour expected if the residual signal is dominated by Monte Carlo noise about a small underlying mean rather than by a persistent physical correlation.
\newparagraph
At the same time, purely visual inspection is not a reliable discriminator as even noise can produce apparently coherent patches at finite resolution in a small $10 \times 10$ matrix. To make the comparison with \eqref{eq:doestypeacorrelate} quantitative, we therefore perform the following two targeted diagnostics.
\begin{itemize}
    \item[(i)] \emph{Exact diagonal identity check.} For $\vec{m}_1 \neq \vec{m}_2$, the diagonal entries are fixed by 
    \begin{equation}
        G_{\vec{m}_1, \vec{m}_2} (e, e) \overset{!}{=} - P_e(\vec{m}_1) P_e(\vec{m}_2).
    \end{equation}
    Numerically, this identity is observed to be satisfied to machine precision (max, mean and root-mean squared (RMS) absolute deviations are all zero within floating-point representation), demonstrating that the prominent diagonal structure visible in the plots is \emph{not} a long-range effect but a kinematic projector-algebra consequence. In particular, at $m_\mathrm{max} = 2$, the diagonal magnitude $\sim 6.5 \times 10^{-5}$ matches the natural scale $P_e(\vec{m}_1) P_e(\vec{m}_2) \sim (1/125)^2$.

    \item[(ii)] \emph{Adjacency versus disjoint edge pairs.} On $K_5$, each edge is adjacent (shares a vertex) to exactly 6 edges and disjoint from exactly 3 edges. This yields a clean notion of graph distance for 2-edge observables. If type-A solutions had genuine long-range correlations in $\hat{\mathcal{C}}$, one would expect a clear difference between the typical magnitude of $G$ on adjacent pairs and on disjoint pairs (or, at least, a stable trend in that direction as $N_S$ increases). We therefore partition, for each fixed $e_1$, the set of $e_2 \neq e_1$ into adjacent and disjoint subsets and compare summary statistics of the corresponding connected correlator entries.  For the largest-statistics run with $N_S = 225,000$, we obtain
    \begin{align}
        \mathrm{meanabs}_\mathrm{adj} & = 1.650429 \times 10^{-5}, \quad &\mathrm{meanabs}_\mathrm{dis} & = 1.412663 \times 10^{-5}, \\
        \mathrm{rms}_\mathrm{adj} & = 2.005307 \times 10^{-5}, \quad &\mathrm{rms}_\mathrm{dis} & = 1.836863 \times 10^{-5},
    \end{align}
    corresponding to
    \begin{equation}
        \Delta_\mathrm{meanabs} := \mathrm{meanabs}_\mathrm{adj} - \mathrm{meanabs}_\mathrm{dis} = 2.377666 \times 10^{-6},
    \end{equation}
    and 
    \begin{equation}
        \frac{\mathrm{meanabs}_\mathrm{adj}}{\mathrm{meanabs}_\mathrm{dis}} = 1.168311.
    \end{equation}
    While the ratio is slightly above unity, its magnitude is small compared to the overall stochastic spread of entries and, by itself, does not establish a genuine effect. We note that here, we write $\mathrm{meanabs}_{\mathrm{dis}}$ (respectively $\mathrm{meanabs}_{\mathrm{adj}}$) for the empirical average of the absolute connected correlator entries over all off-diagonal disjoint (respectively adjacent) edge pairs. These sign-insensitive summaries quantify the typical fluctuation scale of the off-diagonal correlator and are therefore the natural diagnostics for distinguishing genuine adjacency-enhanced correlations from cancellations or noise.

    \item[(iii)] \emph{Scaling with Monte Carlo statistics.} Finally, we examine how the typical fluctuation scale decays with the sample size $N_S$. Table \ref{tab:typeA_scaling} summarises the RMS and adjacency/disjoint mean-absolute statistics across the progressively larger runs. The overall RMS decreases markedly with $N_S$ which is consistent with the expected $\sim N_S^{-1/2}$ Monte Carlo convergence of covariance estimators. In contrast, the adjacency/disjoint ratio does not stabilise to a clear non-trivial limit. Rather, it fluctuates around $O(1)$.
    \begin{table}[ht]
        \small
        \centering
        \caption{
        Scaling of adjacency vs.\ disjoint statistics for the type-A connected chromaticity correlator $G_{\vec m_1,\vec m_2}(e_1,e_2)$ at cutoff $m_{\max}=2$ (same $(\vec m_1,\vec m_2)$ as in Fig.~\ref{fig:typeA_connected_progressive}). Here ``adj'' means $e_1$ and $e_2$ share a vertex, ``dis'' means they are disjoint.
        }
        \begin{tabular}{c|cccc}
            \toprule[1.2pt]
            $N_S$ & $\mathrm{rms}_{\rm dis}$ & $\mathrm{rms}_{\rm adj}$ & $\Delta_{\rm meanabs}$ & $\mathrm{meanabs}_{\rm adj}/\mathrm{meanabs}_{\rm dis}$\\
            \midrule
            $22,500$  & $5.9574\times 10^{-5}$ & $5.3608\times 10^{-5}$ & $-8.0046\times 10^{-6}$ & $0.8416$\\
            $45,000$  & $4.2947\times 10^{-5}$ & $4.1169\times 10^{-5}$ & $\phantom{-}1.8823\times 10^{-6}$ & $1.0555$\\
            $90,000$  & $2.4933\times 10^{-5}$ & $3.1517\times 10^{-5}$ & $\phantom{-}3.7775\times 10^{-6}$ & $1.1749$\\
            $225,000$ & $1.3779\times 10^{-5}$ & $1.7241\times 10^{-5}$ & $\phantom{-}2.4600\times 10^{-6}$ & $1.2147$ \\
            \bottomrule[1.2pt]
        \end{tabular}
        \label{tab:typeA_scaling}
    \end{table}
\end{itemize}
Taken together, these diagnostics align closely with the theoretical expectation for the type-A family. The only robust non-zero structure in $G_{\vec{m}_1, \vec{m}_2}(e_1, e_2)$ is the diagonal contribution fixed by the projector identities, which is entirely kinematic and not indicative of inter-edge correlations. Off the diagonal, the connected correlator is small in magnitude and decreases with increasing Monte Carlo statistics, consistent with a near-factorised joint distribution and with residual fluctuations dominated by sampling noise. The type-A solutions are consistent with exhibiting \emph{no measurable long-range correlations} in chromaticity observables beyond those enforced locally by the projector algebra. 
\newparagraph
We now repeat the same correlation analysis for the type-B near-kernel family at the same cutoff $m_\mathrm{max} = 2$ where we consider the charge vectors
\begin{equation}
    \vec{m}_1 = (-1, 0, 0) \quad,\quad \vec{m}_2 = (1, 0, 0),
\end{equation}
i.e. opposite fundamental excitations in a single $\mathrm{U}(1)$ component. 
\begin{figure}[t]
\centering
\begin{subfigure}{0.32\textwidth}
    \centering
    \includegraphics[width=\textwidth]{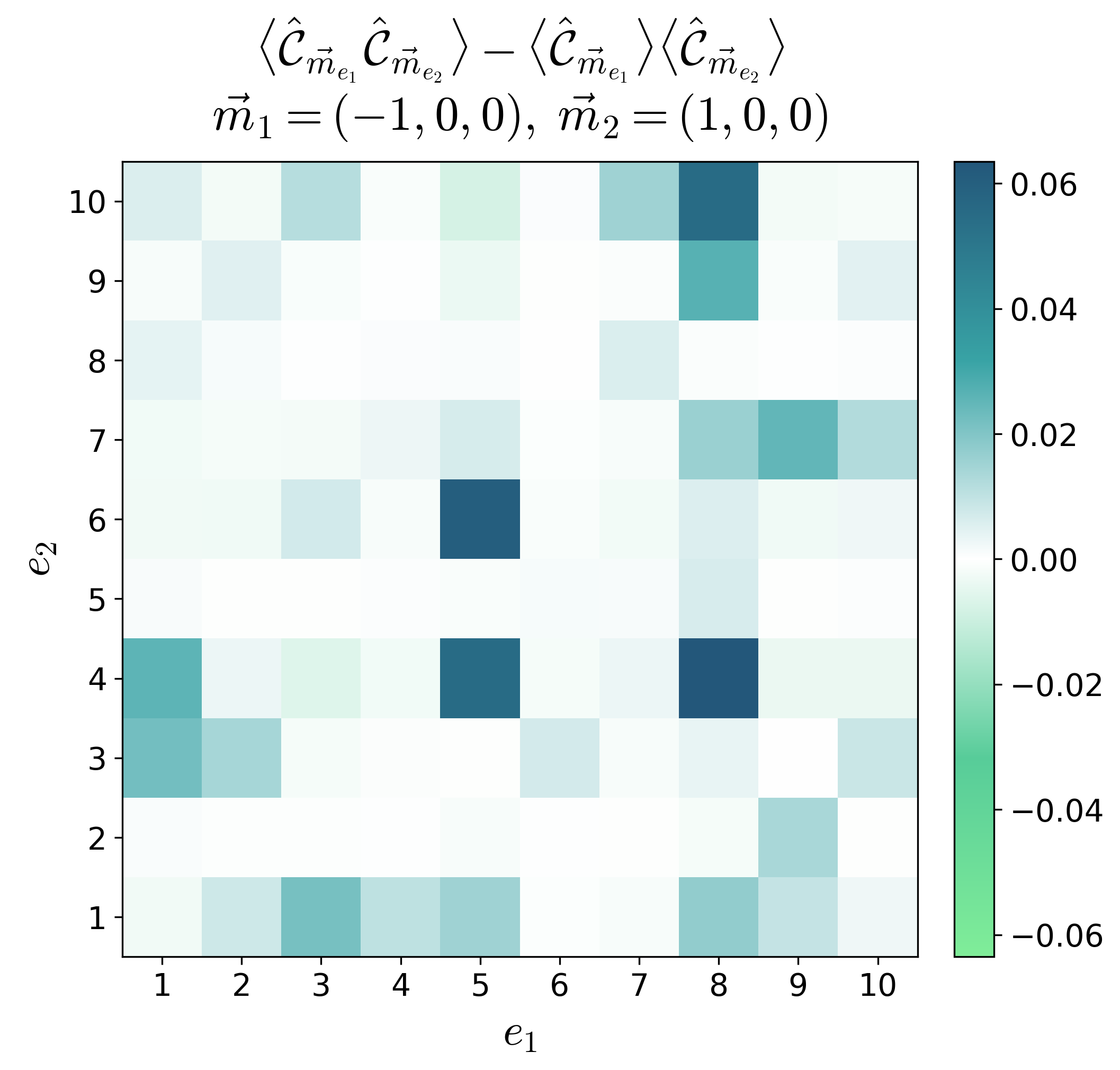}
\end{subfigure}\hfill
\begin{subfigure}{0.32\textwidth}
    \centering
    \includegraphics[width=\textwidth]{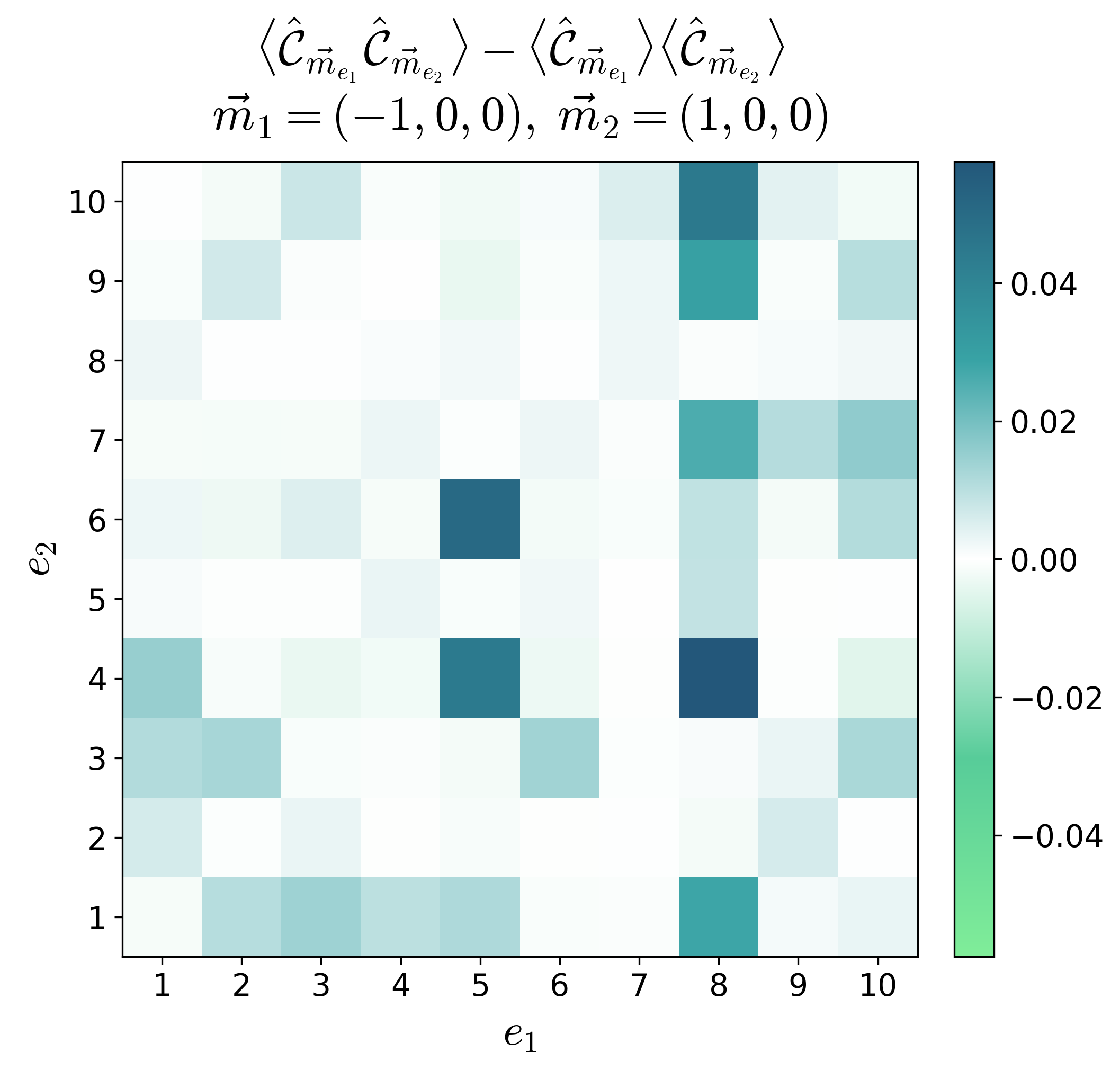}
\end{subfigure}\hfill
\begin{subfigure}{0.32\textwidth}
    \centering
    \includegraphics[width=\textwidth]{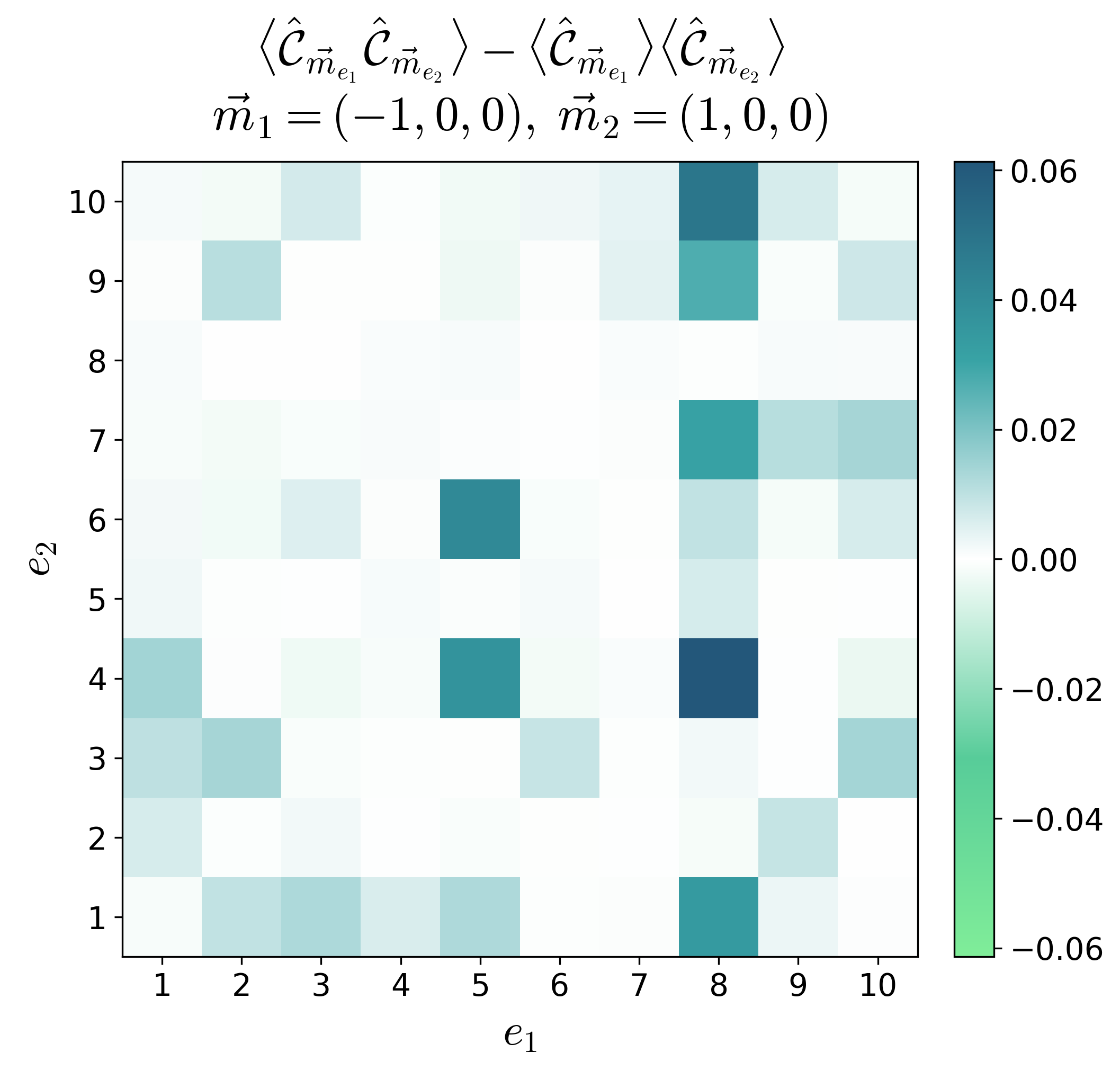}
\end{subfigure}

\medskip

\hspace{50pt}
\begin{subfigure}{0.32\textwidth}
    \centering
    \includegraphics[width=\textwidth]{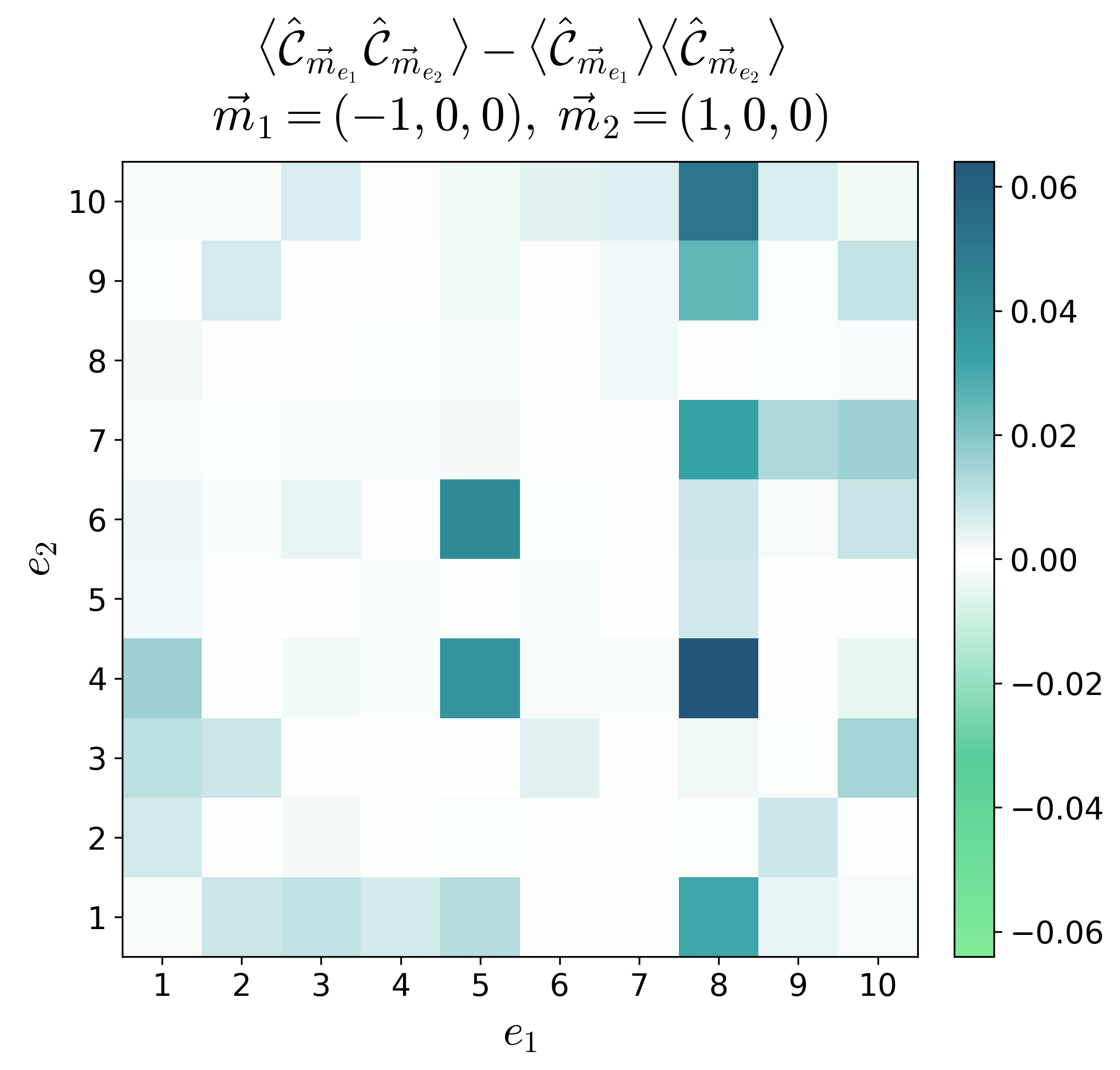}
\end{subfigure}\hfill
\begin{subfigure}{0.32\textwidth}
    \centering
    \includegraphics[width=\textwidth]{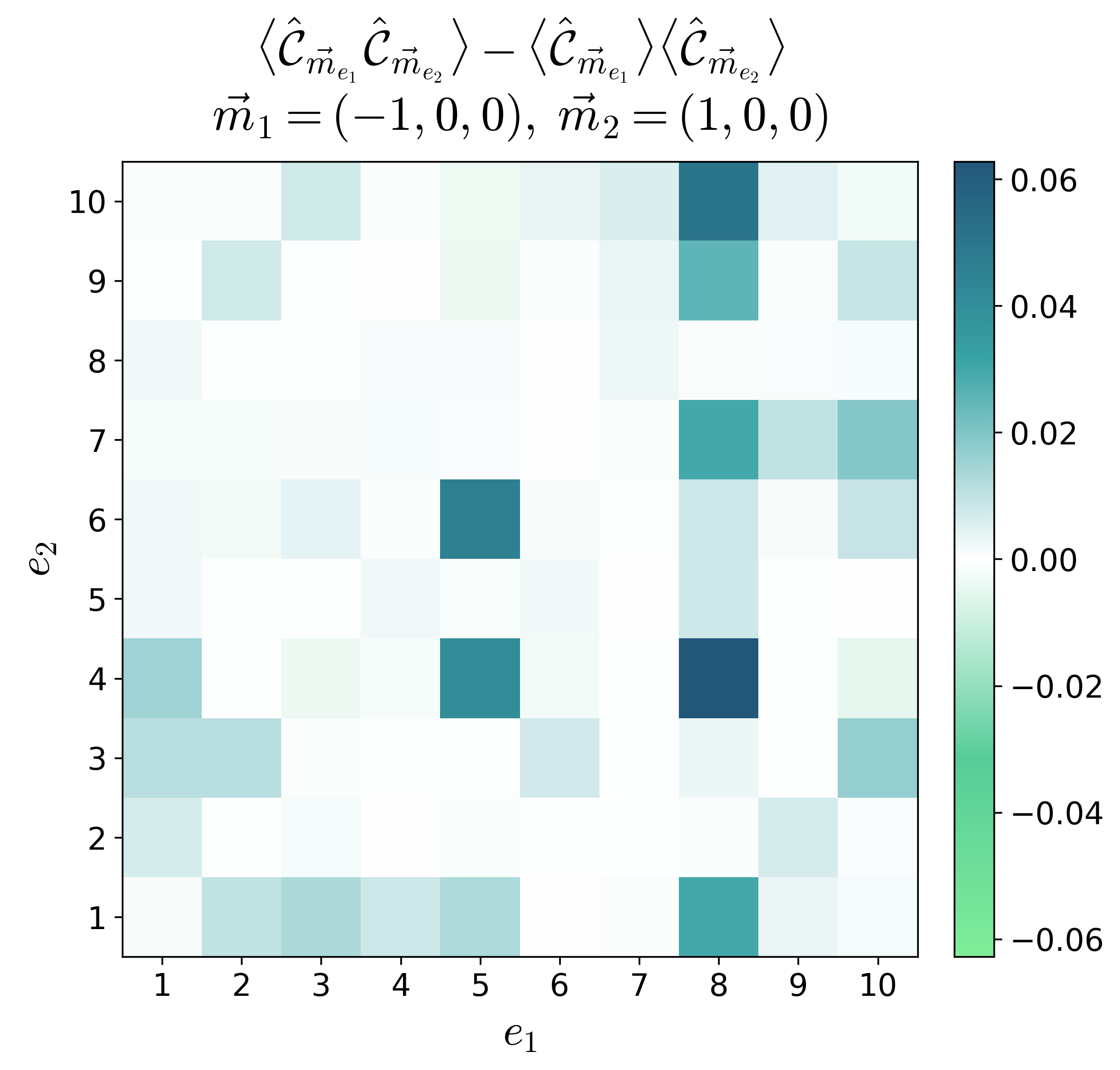}
\end{subfigure}
\hspace{50pt}
\caption{Type-B connected chromaticity 2-point function $G_{\vec m_1,\vec m_2}(e_1,e_2)$ at cutoff $m_{\max}=2$ for $\vec m_1=(-1,0,0)$ and $\vec m_2=(1,0,0)$, shown for increasing Monte Carlo statistics $N_S$ (top row left-to-right: $1350$, $9000$ and $18,000$, bottom row: $45,000$ and $225,000$). The colour scale is symmetric about zero in each panel. Unlike the type-A case, the off-diagonal pattern persists and stabilises rapidly with increasing $N_S$, indicating a robust correlation structure.}
\label{fig:typeB_corr}
\end{figure}
\newparagraph
Figure \ref{fig:typeB_corr} shows $G_{\vec{m}_1, \vec{m}_2}(e_1, e_2)$ for a representative type-B solution at $m_\mathrm{max} = 2$ as a $10 \times 10$ matrix over ordered edge pairs $(e_1, e_2)$ for increasing Monte Carlo sample sizes. In stark contrast to the type-A case (Figure \ref{fig:typeA_connected_progressive}), two qualitative features are immediate. First, the off-diagonal structure is already pronounced at low statistics and remains visually stable as $N_S$ increases. Second the overall amplitude scale does \emph{not} decay toward zero with increasing $N_S$, which already suggests that the observed pattern is not a sampling artefact but reflects a genuine covariance structure of the learned states.
\newparagraph
To once again convert the above qualitative statements into a precise statement about adjacency-driven versus longer-range structure, we apply the same two diagnostics as for type-A.
\begin{itemize}
    \item[(i)] \emph{Exact diagonal identity check.} As emphasised above, the diagonal entries for $\vec{m}_1 \neq \vec{m}_2$ are fixed kinematically by the projector algebra. Numerically, this identity is again satisfied to machine precision as in the case for type-A solutions. Thus, as in type-A, any long-range information must reside in the off-diagonal sector.

    \item[(ii)] \emph{Adjacency versus disjoint edge pairs.} We again partition the off-diagonal entries row-wise into adjacent and disjoint sets and compare summary statistics of their magnitudes. For the largest statistics run ($N_S = 225,000$), we obtain
    \begin{align}
        \mathrm{meanabs}_\mathrm{adj} & = 8.906847 \times 10^{-3}, \quad &\mathrm{meanabs}_\mathrm{dis} & = 1.022913 \times 10^{-3}, \\
        \mathrm{rms}_\mathrm{adj} & = 1.571974 \times 10^{-2}, \quad &\mathrm{rms}_\mathrm{dis} & = 1.875296 \times 10^{-3},
    \end{align}
    corresponding to
    \begin{equation}
        \Delta_\mathrm{meanabs} := \mathrm{meanabs}_\mathrm{adj} - \mathrm{meanabs}_\mathrm{dis} = 7.883934 \times 10^{-3},
    \end{equation}
    and 
    \begin{equation}
        \frac{\mathrm{meanabs}_\mathrm{adj}}{\mathrm{meanabs}_\mathrm{dis}} = 8.707339.
    \end{equation}
    This is qualitatively different from type-A solutions. Namely, adjacency-pair connected correlations are nearly an order of magnitude larger in typical absolute magnitude than disjoint-pair correlations.

    \item[(iii)] \emph{Scaling with Monte Carlo statistics.} Finally, Table \ref{tab:typeB_scaling} reports the same summary statistics across progressively increasing $N_S$. The key observation is that, unlike type-A, neither $\mathrm{rms}_\mathrm{adj}$ nor the adjacency/disjoint ratio collapses toward unity with increasing statistics. Rather, $\mathrm{rms}_\mathrm{adj}$ remains stable at $\approx 1.6 \times 10^{-2}$ across all runs and the ratio stabilises at a value around $7 \sim 8$. This stability is precisely what one expects for a \emph{physical} correlation signal that has converged in Monte Carlo error, rather than for noise-dominated covariance estimates).
        \begin{table}[ht]
        \small
        \centering
        \caption{
        Scaling of adjacency vs.\ disjoint statistics for the type-B connected chromaticity correlator $G_{\vec m_1,\vec m_2}(e_1,e_2)$ at cutoff $m_{\max}=2$ for $\vec{m}_1 = (-1, 0, 0)$ and $\vec{m}_2 = (1, 0, 0)$. Here ``adj'' means $e_1$ and $e_2$ share a vertex while ``dis'' means they are disjoint. In contrast to type-A (Table \ref{tab:typeA_scaling}), the adjacency enhancement is large and stable in $N_S$ indicating a robust correlation structure rather than noise.
        }
        \begin{tabular}{c|cccc}
            \toprule[1.2pt]
            $N_S$ & $\mathrm{rms}_{\rm dis}$ & $\mathrm{rms}_{\rm adj}$ & $\Delta_{\rm meanabs}$ & $\mathrm{meanabs}_{\rm adj}/\mathrm{meanabs}_{\rm dis}$\\
            \midrule
            $22,500$  & $2.3598\times 10^{-3}$ & $1.6420\times 10^{-2}$ & $7.5167\times 10^{-3}$ & $6.7575$\\
            $45,000$  & $1.8444\times 10^{-3}$ & $1.5578\times 10^{-2}$ & $7.8853\times 10^{-3}$ & $8.4462$\\
            $90,000$  & $2.1096\times 10^{-3}$ & $1.5962\times 10^{-2}$ & $8.1019\times 10^{-3}$ & $8.0343$\\
            $225,000$ & $2.1467\times 10^{-3}$ & $1.5788\times 10^{-2}$ & $7.8982\times 10^{-3}$ & $7.9406$ \\
            \bottomrule[1.2pt]
        \end{tabular}
        \label{tab:typeB_scaling}
    \end{table}
\end{itemize}
The type-B results are not only different from type-A, they are, in fact, structurally consistent with what one should expect from a sharply concentrated gauge invariant state whose dominant non-trivial support consists of sparse flux excitations over a mostly-$\vec{0}$ background. 
\newparagraph
First, the strong adjacency enhancement is precisely the signature of local Gauß driven organisation. If a $(-1, 0, 0)$ event occurs on an edge, then (in a regime where most other incident edges carry 0 charge) the constraint must be satisfied by producing compensating $(+1, 0, 0)$ charge on one of the neighbouring edges at one of the endpoints, yielding large positive covariances on adjacent pairs. The observed near order of magnitude separation in $\mathrm{meanabs}_\mathrm{adj} / \mathrm{meanabs}_\mathrm{dis}$ therefore aligns with the picture that type-B near-kernel solutions are \emph{not} high-entropy, near factorised states but rather they exhibit strong constraint saturated local structure.
\newparagraph
Second, while the disjoint sector is substantially suppressed relative to adjacency, \emph{it remains non-negligible} in absolute terms and, crucially, does not decay away with increasing $N_S$ as would be expected for pure Monte Carlo noise. On $K_5$, the notion of ``disjoint'' already corresponds to the maximal available separation between two edges (graph distance 2). Thus, a non-vanishing disjoint component is consistent with correlations propagating through intermediate vertices and/or with extended divergence-free flux patterns (closed-loop excitations) that necessarily couple multiple edges simultaneously. Given the small size of the graph, it would be premature to over-interpret the disjoint signal as a clean ``correlation length'' phenomenon. 

\subsubsection{Relation to kinematical vacua}
\label{subsubsec:relationtokinvacua}
The correlation diagnostics above can be naturally interpreted in the previous section can be naturally interpreted in terms of the two standard vacuum notions that coexists already at the kinematical level of LQG, namely the Ashtekar-Lewandowski (AL) vacuum \cite{Ashtekar:1994wa} and the Dittrich-Geiller (DG) vacuum \cite{Dittrich:2014wpa}. In our current representation, the AL vacuum is the cylindrically consistent constant function in the holonomy picture, which in the  $\mathrm{U}(1)^3$ charge basis is sharply supported on the trivial label
\begin{equation}
    \ket{\Omega_\mathrm{AL}} = \ket{\vec{0} \cdots \vec{0}} \implies P_e^\mathrm{AL} := \langle \Omega_\mathrm{AL}, \hat{\mathcal{C}}_{\vec{m}}^{(e)} \Omega_\mathrm{AL} \rangle = \delta_{\vec{m}, \vec{0}},
\end{equation}
for every edge $e$. Consequently, for any pair of edges $e_1, e_2$ and any $\vec{m}_1, \vec{m}_2$, the connected chromaticity correlator vanishes identically in the strict AL vacuum (i.e. the AL vacuum is ultra-local in these projector observables).
\newparagraph
By contrast, the DG vacuum is defined (heuristically on $\gamma$) as a flatness projector state. Namely, in a holonomy representation it is supported on configurations with trivial curvature around a generating set of loops. This object is distirbutional with respect to the AL measure, as in it is not expected to define a normalisable vector in the kinematical Hilbert space and in the charge basis, its Fourier coefficients are correspondingly delocalised (``magnetic'' dual to the AL vacuum). Operationally, a DG solution is therefore expected to (i) strongly suppress curvature diagnostics, (ii) fail to concentrate in charge space under cutoff refinement, and (iii) not display robust long-range organisation when probed by very local charge projectors unless additional nonlocal constraints are present.
\newparagraph
With all the characterisation tests done so far in mind, the ordering-induced near-kernel families can be read as follows. The type-A family exhibits flatness on all minimal loops while simultaneously having non-vanishing volume. Together with the observation that the connected chromaticity correlator for type-A shows no robust off-diagonal signal beyond kinematic projector identities and Monte Carlo noise, this is precisely the pattern one expects when probing a DG sector using charge-projector observables. Specifically, the physically distinguishing information is encoded in flatness and holonomy constraints, while the local charge events themselves appear close to factorised. 
\newparagraph
In contrast, the type-B family is exactly volume-degenerate in the strong sense that it lies sharply in $\bigcap_v \ker(\hat{V}_v)$ while being not flat (i.e. it belongs to a curvature-carrying family of solutions with degenerate spatial geometry). At the same time, its single edge chromaticity is dominated by mostly-$\vec{0}$ background and its connected correlator shows a large, stable enhancement on adjacent edge pairs together with a suppressed but persistent disjoint component. Lastly, as also shown in the stratified state-vector plots in Section \ref{subsubsec:normalisability}, the all zero basis state appears to hold significantly more weight compared to all other basis states. This combination is naturally interpreted as an AL-like vacuum sector dressed by sparse divergence-free flux excitations. The $\vec{0}$ dominance is the direct AL signature in the charge basis, while the strong adjacency enhancement is the expected imprint of Gauß saturation when rare non-trivial charges occur in an otherwise trivial background. The residual disjoint signal (at graph distance 2 on $K_5$) is then consistent with the fact that even local Gauß organisation can propagate through intermediate vertices and/or arise from closed-loop flux patterns that necessarily couple several edges at once, without yet implying a clean correlation-length interpretation on such a small graph.
\newparagraph
To close, the long-range correlation analysis therefore closes the conceptual loop with the geometric diagnostics. The type-A solutions are most naturally associated with a DG vacuum sector, whereas the type-B solutions are most naturally associated with an AL-like vacuum sector. In this sense, the ordering choice in the quadratic constraint does not merely alter convergence properties, it selects between two qualitatively distinct vacuum sectors that are already familiar from the kinematical structure of LQG.

\subsubsection{Geometric structure}
\label{subsubsec:geometricstructure}
Having established that the two ordering-induced near-kernel families occupy sharply distinct geometric sectors which can be related to both an AL-like and DG vacua, we now turn to a more direct characterisation of the \emph{spatial geometric structure} encoded by these states. In what follows, we isolate two complementary questions. First, whether the induced discrete geometry is approximately isotropic or instead anisotropic and second, whether the geometry admits a consistent \emph{effective radius} when inferred independently from area and from volume.
\newparagraph
The conceptual link between the $K_5$ graph and a spherical spatial geometry is not ad hoc. Namely, the boundary of a 4-simplex is a closed 4-dimensional simplicial complex homeomorphic to $S^3$, consisting of five tetrahedra. The dual 1-skeleton of this boundary triangulation has one node per boundary tetrahedron and one link per shared triangular face. Since any two distinct boundary tetrahedra in the 4-simplex share a triangle, the dual 1-skeleton is precisely the complete graph on five vertices, i.e. the $K_5$ graph.
\newparagraph
In this sense, a charge-network state on $K_5$ can be read as a discrete quantum geometry living on a minimal triangulation of a closed spatial slice (topologically $S^3$) in the same way that spin-network data on dual graphs encode discrete geometries in simplicial or twisted-geometry interpretations. On such a graph, the edge labels provide the most direct discrete flux data, while gauge invariance at vertices enforces the corresponding discrete closure conditions. Geometric operators such as area and volume therefore probe \emph{metric} content of the learned state beyond the characterisation tests already discussed. In what follows, we outline the two geometric characterisation diagnostics we conduct.
\begin{itemize}
    \item[(i)] \emph{Anisotropy analysis via area cuts.} The first diagnostic is designed to detect anisotropy on this minimal discretisation of a closed geometry. We fix two embedded transverse 2-surfaces $S_\mathrm{hor}$ and $S_\mathrm{ver}$ through the graph (in practice, one horizontal and one vertical cut in the chosen embedding used throughout the work), illustrated in Figure \ref{fig:cutk5graph}.
    \begin{figure}[ht]
        \centering
        \includegraphics[width=0.5\linewidth]{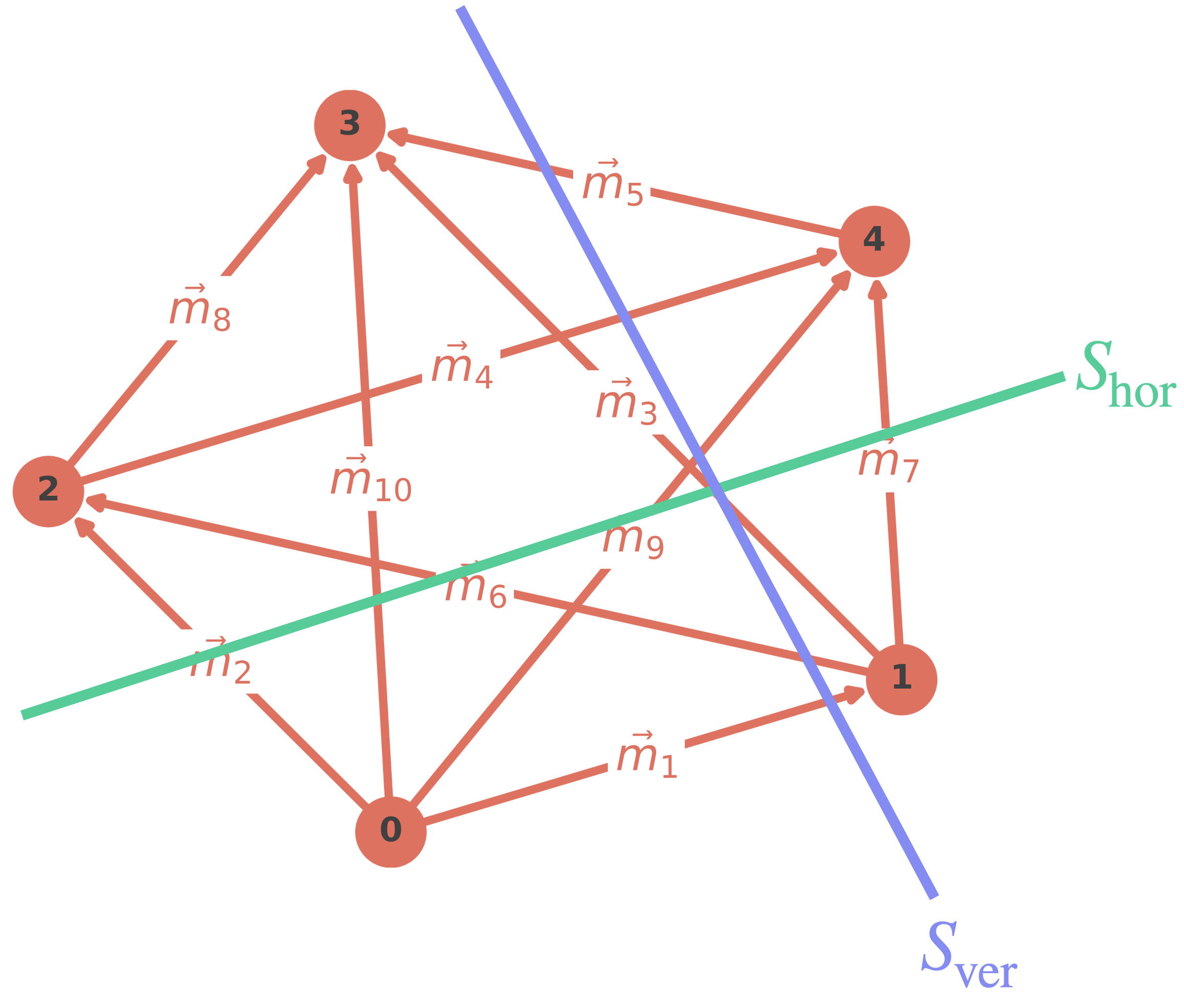}
        \caption{The two area surfaces $S_\mathrm{hor}$ and $S_\mathrm{ver}$ cutting through the $K_5$ graph in a horizontal and vertical manner respectively are shown.}
        \label{fig:cutk5graph}
    \end{figure}
    \\\noindent
    Each surface intersects a subset of edges and we denote the corresponding puncture set by
    \begin{equation}
        E(S) := \{e \in E(\gamma)\,|\, e \cap S \neq \emptyset\}.
    \end{equation}
    The associated area operator takes the standard sum-over-punctures form. In the present $\mathrm{U}(1)^3$ setting, its action is diagonal in the charge-network basis and reduces to
    \begin{equation}
        \hat{A}(S)\ket{\{\vec{m}_e\}} = \Big(\sum_{e \in E(S)} a(\vec{m}_e)\Big)\ket{\{\vec{m}_e\}}, 
        \qquad 
        a(\vec{m}_e) \propto \sqrt{\sum_{i=1}^3 (m_e^{(i)})^2},
    \end{equation}
    up to the usual overall LQG prefactor (set to 1 here since it is irrelevant for the relative measure below). We evaluate the expectation values $\langle \hat{A}(S_\mathrm{hor}) \rangle$ and $\langle \hat{A}(S_\mathrm{ver}) \rangle$ on every accepted near-kernel run and introduce a dimensionless \emph{anisotropy magnitude} comparable across cutoffs,
    \begin{equation}
        \mathcal{A} := \frac{\big|\langle \hat{A}(S_\mathrm{hor}) \rangle - \langle \hat{A}(S_\mathrm{ver}) \rangle\big|}{\langle \hat{A}(S_\mathrm{hor}) \rangle}.
        \label{eq:anisotropy_magnitude}
    \end{equation}
    This observable measures the \emph{relative size} of a horizontal-vertical area asymmetry in units of the horizontal area scale. The absolute value ensures that $\mathcal{A}$ is insensitive to the arbitrary orientation choice of the embedding (or to discrete graph automorphisms that exchange the two cuts) as in an ensemble with no preferred axis, individual runs may yield either sign of $\langle \hat{A}(S_\mathrm{hor}) \rangle - \langle \hat{A}(S_\mathrm{ver}) \rangle$, but the physically relevant question is whether the \emph{magnitude} of this asymmetry is consistently small (approximately isotropic family) or systematically large (strongly anisotropic family). Importantly, $\mathcal{A}$ is not a curvature diagnostic. It isolates directional anisotropy in area data irrespective of whether the state lies in a flat or curved sector.

    \item[(ii)] \emph{Radius consistency.} The second diagnostic asks whether the discrete geometry admits a consistent effective radius when inferred from (i) the total 3-volume and (ii) an appropriate 2-area cross section. The expectation value of the total volume operator,
    \begin{equation}
        \hat{V}_\mathrm{tot} = \sum_{v \in V(\gamma)} \hat{V}_v,
    \end{equation}
    yields the volume-inferred radius via the continuum $S^3$ relation $\mathrm{Vol}(S^3_R)=2\pi^2R^3$,
    \begin{equation}
        R_V = \left(\frac{\langle \hat{V}_\mathrm{tot} \rangle}{2 \pi^2}\right)^{1/3}.
    \end{equation}
    Independently, we infer an area-based radius from a discrete equatorial $S^2$ cross section where within the fixed embedding we select a set of six edges that form a maximally representative ``equatorial belt'' on the $K_5$ graph (i.e. edges lying on the equatorial cut used throughout), compute their individual area contributions, and use their mean area $\langle \hat{A}(S_\mathrm{eq})\rangle$ as a coarse estimator of the equatorial cross-sectional area scale. For a round 3-sphere, $\mathrm{Area}(S^2_R)=4\pi R^2$, motivating the area-inferred radius
    \begin{equation}
        R_A = \left(\frac{\langle \hat{A}(S_\mathrm{eq})\rangle}{4\pi}\right)^{1/2}.
    \end{equation}
    If the learned state encodes a geometry well-approximated (at this coarse level) by a round $S^3$-like spatial slice, one expects $R_A$ and $R_V$ to agree up to fluctuations and discretisation artefacts. Conversely, a systematic mismatch indicates either strong anisotropy, strong inhomogeneity, or (in extreme cases) a degenerate family of solutions in which volume support collapses while non-trivial area excitations persist.
\end{itemize}
For each cutoff $m_\mathrm{max}\in\{1,\ldots,5\}$ we collect all converged training simulations whose constraint expectation satisfies the admissibility criterion
\begin{equation}
    \langle \hat{\mathcal{Q}}_i \rangle_\mathrm{NQS} \leq \varepsilon,
    \qquad
    \varepsilon = 1.0,
\end{equation}
and compute the anisotropy magnitude $\mathcal{A}$ from \eqref{eq:anisotropy_magnitude} as well as the radii diagnostics $R_A$, $R_V$, and the symmetric relative mismatch
\begin{equation}
    \Delta R[\%] := 100\times\frac{|R_A - R_V|}{(R_A + R_V)/2}.
    \label{eq:deltaR_def}
\end{equation}
Unless stated otherwise, ensemble summaries are reported across the accepted runs at fixed cutoff.
\newparagraph
To visualise the full distribution of anisotropy magnitudes across the accepted ensemble at fixed cutoff, Figure \ref{fig:anisotropy_magnitude_boxplots} shows box-and-whisker plots of $\{\mathcal{A}_{m_\mathrm{max},i}\}_{i=1}^{N_{m_\mathrm{max}}}$ for each solution family. For each cutoff, the box spans the interquartile range (25th - 75th percentile), the horizontal line inside the box marks the median, and the whiskers extend to the standard non-outlier range (points beyond are shown explicitly as outliers). The filled marker indicates the sample mean $\bar{\mathcal{A}}_{m_\mathrm{max}}$. The vertical error bar attached to this marker shows the \emph{propagated Monte Carlo standard error} on the mean arising from finite sampling within each accepted run. Namely, if run $i$ provides an estimate $\mathcal{A}_{m_\mathrm{max},i}$ with Monte Carlo standard error $\sigma_{\mathcal{A},i}$, then the displayed uncertainty is $\sigma_{\bar{\mathcal{A}}} = \frac{1}{N}\sqrt{\sum_i \sigma_{\mathcal{A},i}^2}$. The annotation $N$ records the number of accepted runs at each cutoff.
\begin{figure}[ht]
    \centering
    \includegraphics[width=0.49\linewidth]{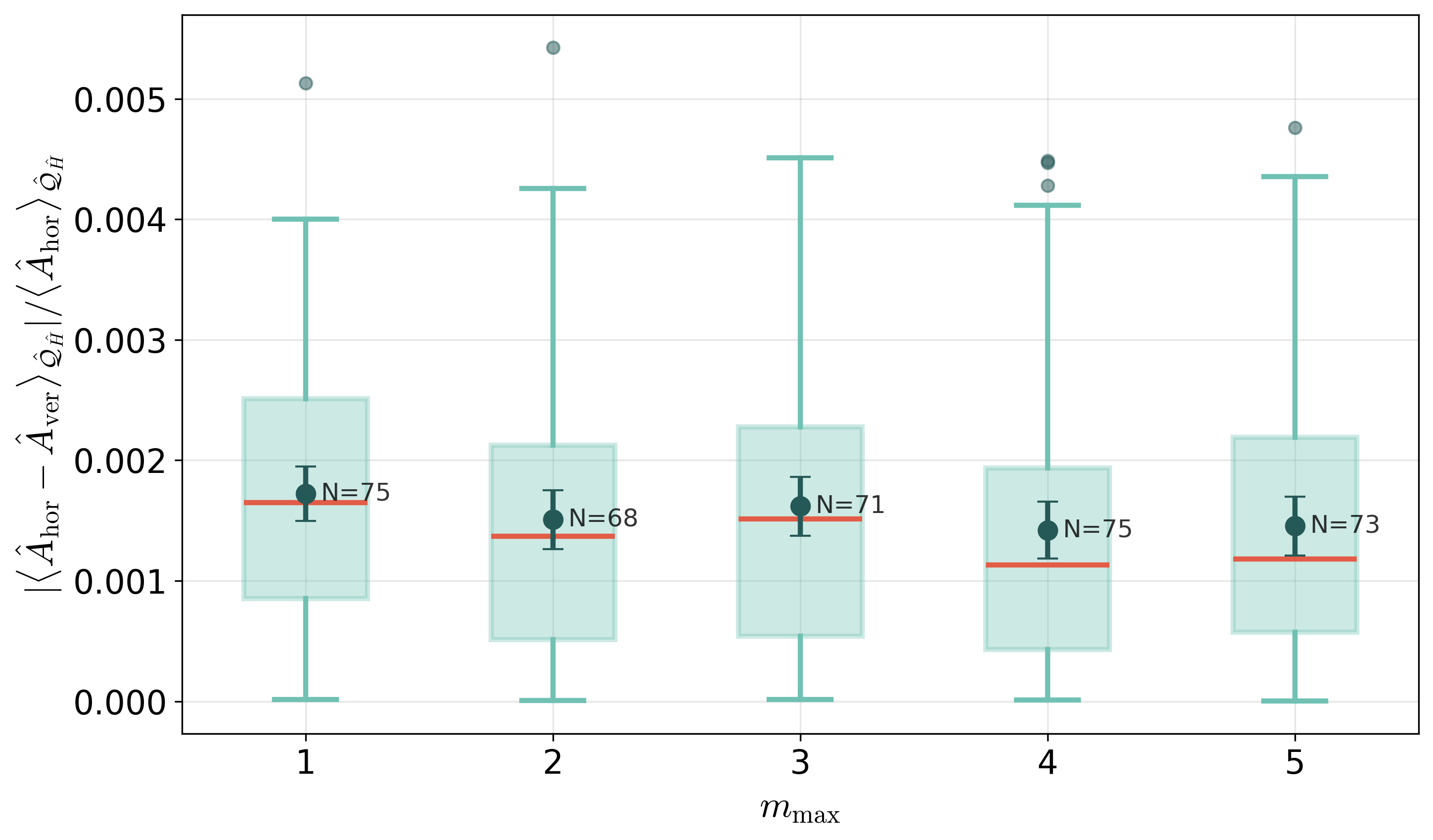}\hfill
    \includegraphics[width=0.49\linewidth]{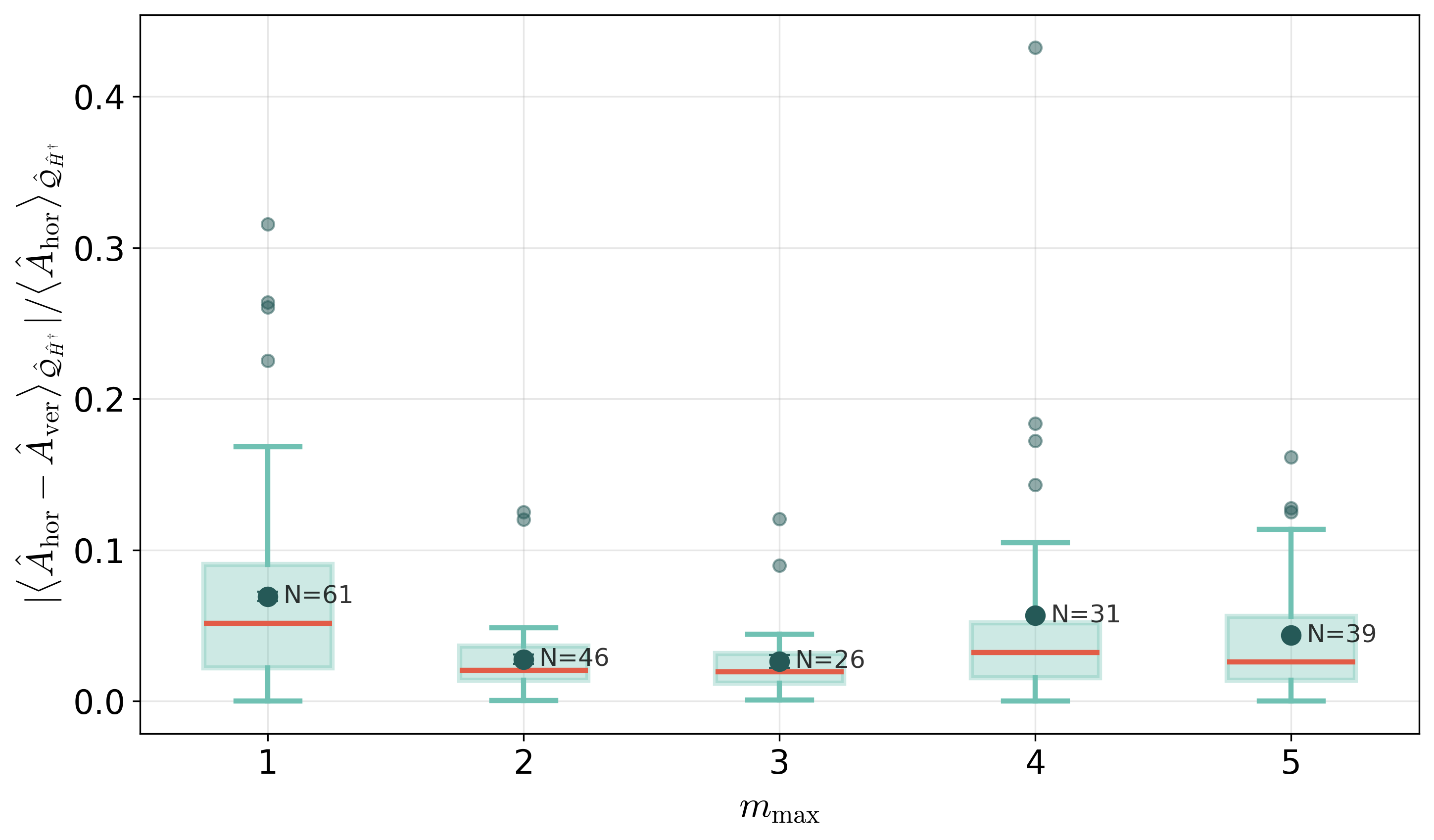}
    \caption{
    Distribution of the anisotropy magnitude $\mathcal{A}$ across the accepted near-kernel ensemble as a function of cutoff $m_\mathrm{max}$, shown separately for the two solution families (type-A on the left, type-B on the right). Boxes indicate the interquartile range with median (horizontal line), whiskers indicate the non-outlier range, and points denote outliers. The filled marker gives the sample mean $\bar{\mathcal{A}}_{m_\mathrm{max}}$ and the vertical error bar shows the propagated Monte Carlo standard error on the mean. The number of accepted runs $N$ at each cutoff is annotated next to the mean marker.
    }
    \label{fig:anisotropy_magnitude_boxplots}
\end{figure}
\newparagraph
Figure \ref{fig:anisotropy_magnitude_boxplots} shows that the standard (type-A) family yields anisotropy magnitudes $\mathcal{A}$ that remain tightly concentrated at the $\mathcal{O}(10^{-3})$ level across cutoffs, with narrow interquartile ranges and only occasional outliers. This indicates that, for most accepted runs, the two transverse area probes $S_\mathrm{hor}$ and $S_\mathrm{ver}$ return very similar expectation values, i.e. the learned discrete metric data do not single out a preferred axis at the level of this coarse probe. By contrast, the alternative (type-B) family exhibits substantially larger typical $\mathcal{A}$ and markedly broader run-to-run variability, with a heavier tail of outliers. Operationally, this means that a significant fraction of accepted type-B runs display a non-negligible hierarchy between the two transverse area cuts, consistent with a more heterogeneous and often strongly anisotropic geometry on the minimal $K_5$ discretisation.
\newparagraph
Turning to radius consistency, Figure \ref{fig:deltaR_standard_boxplot} summarises the distribution of the radius mismatch $\Delta R$ (defined in \eqref{eq:deltaR_def}) for the \emph{standard} (type-A) solution family as a function of cutoff. We again use box-and-whisker plots to display the empirical distribution across accepted runs at fixed $m_\mathrm{max}$, with the filled marker indicating the sample mean and the attached vertical error bar giving the propagated Monte Carlo standard error on the mean (computed from within-run sampling uncertainties and combined across runs). Across cutoffs, the $\Delta R$ distributions are narrow with only mild outliers, and the central tendency increases from a few percent at $m_\mathrm{max}=1$ to a stable $\mathcal{O}(10\%)$ mismatch for $m_\mathrm{max}\geq 2$.
\newparagraph
For the \emph{alternative} (type-B) family, the learned near-kernel states are volume-degenerate in the sense that $\langle \hat{V}_\mathrm{tot}\rangle$ is driven to (numerically) vanishing values across the accepted ensemble. In this regime $R_V$ becomes ill-defined as an effective bulk-radius proxy and the comparison encoded by $\Delta R$ ceases to be meaningful. We therefore do not report a radius-consistency plot for the alternative family.
\begin{figure}[ht]
    \centering
    \includegraphics[width=0.65\linewidth]{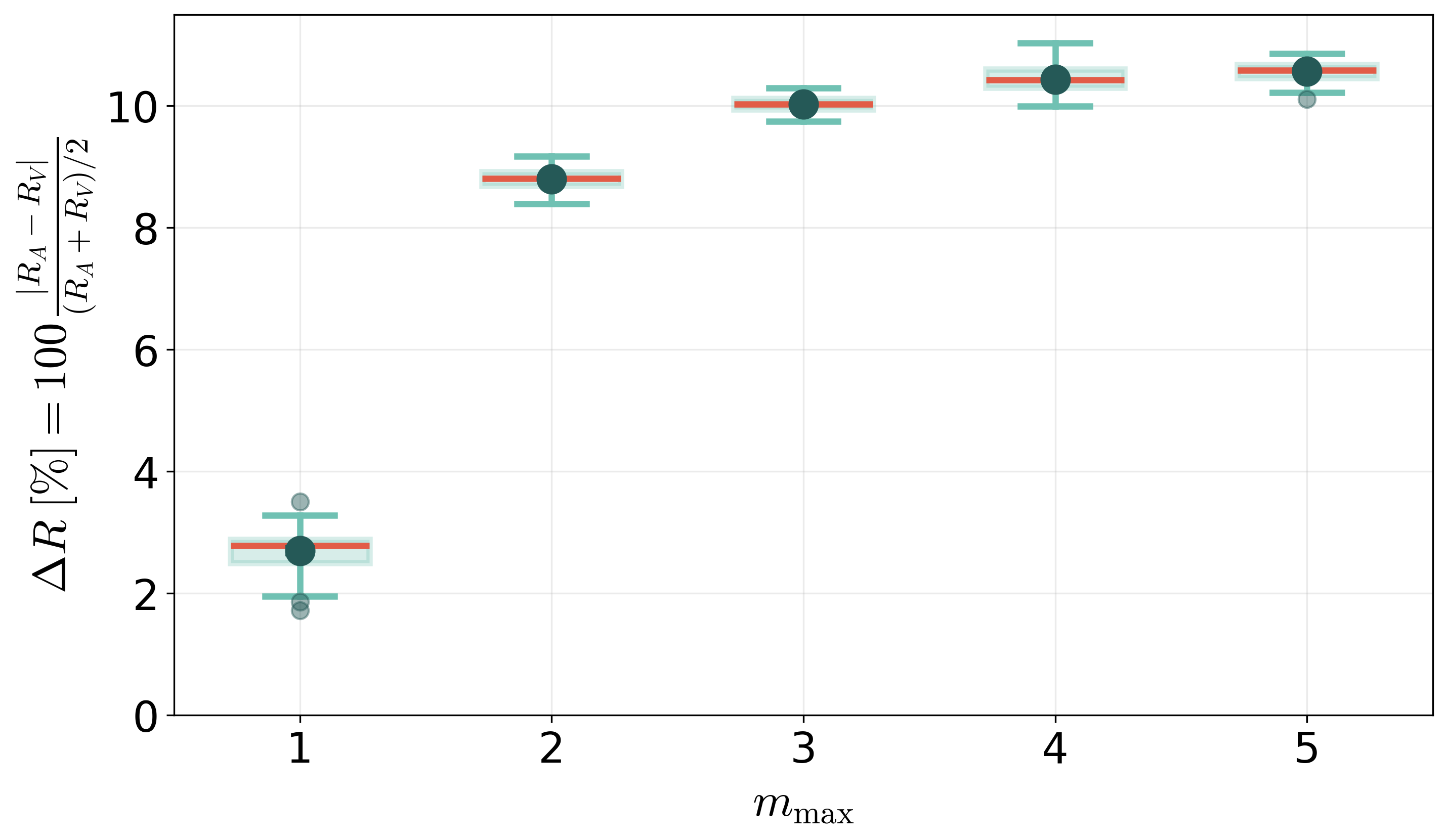}
    \caption{
    Distribution of the radius mismatch $\Delta R$ for the standard (type-A) solution family as a function of cutoff $m_\mathrm{max}$. Boxes indicate the interquartile range with median (horizontal line), whiskers indicate the non-outlier range, and points denote outliers. The filled marker gives the sample mean and the vertical error bar shows the propagated Monte Carlo standard error on the mean.}
    \label{fig:deltaR_standard_boxplot}
\end{figure}
\newparagraph
Taken together, these diagnostics support a clear geometric interpretation of the two near-kernel families on the coarse $K_5$ discretisation. For the standard (type-A) solutions, the anisotropy magnitudes in Figure \ref{fig:anisotropy_magnitude_boxplots} are uniformly small across cutoffs, indicating that the two transverse area probes $S_\mathrm{hor}$ and $S_\mathrm{ver}$ agree to high relative accuracy on most accepted runs. At the same time, the radius-consistency test in Figure \ref{fig:deltaR_standard_boxplot} shows that the area-inferred radius $R_A$ (from the discrete equatorial cross section) and the volume-inferred radius $R_V$ (from $\langle \hat{V}_\mathrm{tot}\rangle$) remain of the same order with a stable $\mathcal{O}(10\%)$ mismatch for $m_\mathrm{max}\geq 2$, as expected for an extended but not perfectly round geometry reconstructed from a minimal triangulation and approximate near-kernel states. In this operational sense, the type-A family behaves as a set of ``geometric sphere'' solutions, they encode an approximately isotropic, three-dimensionally extended spatial geometry that is consistent with a coarse $S^3$-like interpretation. By contrast, the alternative (type-B) family fails the radius-consistency criterion already at the level of definitions. The learned states are volume-degenerate, driving $\langle \hat{V}_\mathrm{tot}\rangle$ (and hence $R_V$) to vanishing values, so that an $S^3$-like effective radius cannot be meaningfully assigned even when transverse area excitations persist. This is precisely the sense in which the results distinguish an extended, approximately spherical family (type-A) from a non-spherical, effectively volume-collapsed family (type-B) on the $K_5$ discretisation.

\subsection{Quasi-solutions for a symmetric ordered quadratic constraint}
\label{subsec:symmordsols}

The results of the previous sections establish two robust families of variational near-kernel states associated with the two quadratic constraints
\begin{equation}
    \hat{\mathcal{Q}}_{\hat{H}} = \sum_{v \in V(\gamma)} \hat{H}_v \hat{H}_v^\dagger
    \quad,\quad
    \hat{\mathcal{Q}}_{\hat{H}^\dagger} = \sum_{v \in V(\gamma)} \hat{H}_v^\dagger \hat{H}_v,
\end{equation}
which we have referred to as type-A (for $\hat{\mathcal{Q}}_{\hat{H}}$) and type-B (for $\hat{\mathcal{Q}}_{\hat{H}^\dagger}$). Each family contains states that solve its \emph{respective} quadratic constraint to near-kernel accuracy. A notable asymmetry is observed when the \emph{other} ordering is evaluated a posteriori. Namely, type-B solutions (obtained from $\hat{\mathcal{Q}}_{\hat{H}^\dagger}$) are often simultaneously solutions under the standard ordering $\hat{\mathcal{Q}}_{\hat{H}}$,
whereas type-A solutions (obtained from $\hat{\mathcal{Q}}_{\hat{H}}$) typically do \emph{not} solve the alternative ordering $\hat{\mathcal{Q}}_{\hat{H}^\dagger}$ to the same degree. This motivates the question whether one can construct interpolating \emph{quasi-solutions} that retain desirable features of both classes, without being variationally driven into either extreme.
\newparagraph
This problem can be approached through several facets. First, we attempted to obtain such interpolating states by augmenting the primary optimisation objective (i.e. either one of the quadratic constraints) with explicit penalty terms that disfavour the known pathologies of the two families. For example, to push type-B solutions away from the volume-degenerate sector, one can add an inverse-volume penalty on a chosen vertex $v_\star$,
\begin{equation}
    \hat{\mathcal{Q}}_{\hat{H}^\dagger} \longrightarrow \hat{\mathcal{Q}}_{\hat{H}^\dagger} + \lambda \frac{1}{(\langle \hat{V}_{v_\star}\rangle+ a)^2},
\end{equation}
with fixed $a>0$ and tunable $\lambda$. Conversely, to push type-A solutions away from strict flatness, one can add a minimal loop holonomy penalty for a chosen loop $\alpha_\star$,
\begin{equation}
    \hat{\mathcal{Q}}_{\hat{H}} \longrightarrow \hat{\mathcal{Q}}_{\hat{H}} + \lambda \left(1 - \dfrac{1}{2} \Re \langle \mathrm{tr}\hat{h}_{\alpha_\star} \rangle \right).
\end{equation}
In practice, both strategies exhibit pronounced sensitivity. For small $\lambda$ the optimiser remains in the original type-A/type-B basin, while for slightly larger $\lambda$ the penalty dominates and the quadratic constraint is no longer reduced satisfactorily. Already at $m_\mathrm{max}=1$ this produces unstable trade-offs rather than controlled interpolation, and we therefore do not pursue purely penalty-based approaches for the bare quadratic constraints further.
\newparagraph
Instead, we construct quasi-solutions by modifying the quadratic constraint itself, and then stabilising the optimisation against the two dominant collapse mechanisms. The starting point is a symmetrised linear combination of the vertex constraints, which we denote as
\begin{equation}
\label{eq:symm_quad_const}
    \hat{\mathcal{Q}}_{\hat{H}_\mathrm{S}} = \sum_{v\in V(\gamma)} (\hat{H}_{\mathrm{S}, v})^2 \quad,\quad \hat{H}_{\mathrm{S}, v} := \frac{1}{2}\bigl(\hat{H}_v + \hat{H}^\dagger_v\bigr),
\end{equation}
and we minimise this corresponding symmetric quadratic objective. This choice forces the variational state to suppress the constraint in a way that couples the two orderings directly, rather than privileging either $\hat{\mathcal{Q}}_{\hat{H}}$ or $\hat{\mathcal{Q}}_{\hat{H}^\dagger}$.
\newparagraph
We note that in practice, since the type-B solutions are simultaneously solutions for the standard ordered quadratic constraint, we implement several mechanisms which further push the optimiser from the solution manifold associated to strictly type-B solutions. First, to prevent the optimisation from collapsing into the AL-like vacuum sector (a behaviour associated with type-B-like solutions), we act with a so called \emph{no-vacuum projector} that suppresses the all-zero basis configuration $\ket{\{\vec{m}\}_{\mathrm{vac}}}$. Specifically, one can write this as
\begin{equation}
    \hat{\Pi} = \mathbb{I} - \ket{\{\vec{m}\}_\mathrm{vac}} \bra{\{\vec{m}\}_\mathrm{vac}}.
\end{equation}
In practice, we suppress the output of the neural network for the given vacuum configuration such that, for a given configuration $\sigma$, then
\begin{equation}
    \Psi^{(\neq 0)}_\mathrm{NQS}(\ket{\{\vec{m}\}}) = \begin{cases}
        0, \quad & \ket{\{\vec{m}\}} = \ket{\{\vec{m}\}_\mathrm{vac}}, \\
        \Psi_\mathrm{NQS}(\ket{\{\vec{m}\}}), \quad & \text{otherwise,}
    \end{cases}
\end{equation}
thereby removing the variational incentive to place essentially all amplitude weight on that single configuration\footnote{An equivalent effect can also be obtained by explicitly penalising mode collapse at the level of the network output, cf. Section \ref{sec:computationalframework}.}. Second, to further prevent collapse into the volume-degenerate sector (a behaviour associated with type-B-like solutions), we include the inverse-volume regulariser on a chosen vertex $v_\star$ as above. In practice, the symmetric constraint, the no-vacuum projector, and the inverse-volume term are used jointly.
\newparagraph
In this combined setup, the symmetric quadratic constraint can be solved variationally, although not to the level of either $\hat{\mathcal{Q}}_{\hat{H}}$ or $\hat{\mathcal{Q}}_{\hat{H}^\dagger}$ independently in terms of final absolute expectation value. Concretely, at $m_\mathrm{max}=1$ a typically obtained variational state is one with $\langle \hat{\mathcal{Q}}_{\hat{H}_{\mathrm{S}}}\rangle \sim 5$, indicating a non-trivial suppression of the constraint but not an almost vanishing expectation value. To place the achieved expectation value on an intrinsic scale, as discussed in Section \ref{subsubsec:genericvmagconst12}, one can define a proxy for the generic maximum variational magnitude for a constraint by solving an auxiliary optimisation problem that \emph{maximises} its expectation value within the same variational family, implemented via the regularised inverse-square objective.
\newparagraph
Figure \ref{fig:Qsym_scale} shows representative values for such a regularised inverse-square objective for both the standard and alternative ordering constraints as well as the symmetric constraint shown in \eqref{eq:symm_quad_const}. As shown in the figure, for $m_\mathrm{max} = 1$, this proxy has a value of 20 - 30 for the considered constraints. By comparison, the typical expectation value of the symmetric constraint $\hat{\mathcal{Q}}_{\hat{H}_\mathrm{S}}$ is of the order $\sim 5$ at the same cutoff. This corresponds to a reduction by roughly a factor of 4 to 6 relative to the scale expected for a generic state, corresponding to about 15 - 25 \% of that generic expectation value scale. This is a clear and non-trivial suppression, although it is not yet close to zero when measured against this intrinsic reference scale. Reaching substantially smaller expectation values for $\hat{\mathcal{Q}}_{\hat{H}_{\mathrm{S}}}$ appears to require a separate and rather sensitive optimisation problem, which we do not attempt to solve here. For the present purposes, the $m_\mathrm{max}=1$ results already provide a proof of principle that structurally stable solutions to the symmetric quadratic constraint can be constructed within the combined objective.
\begin{figure}[ht]
    \centering
    \includegraphics[width=0.65\textwidth]{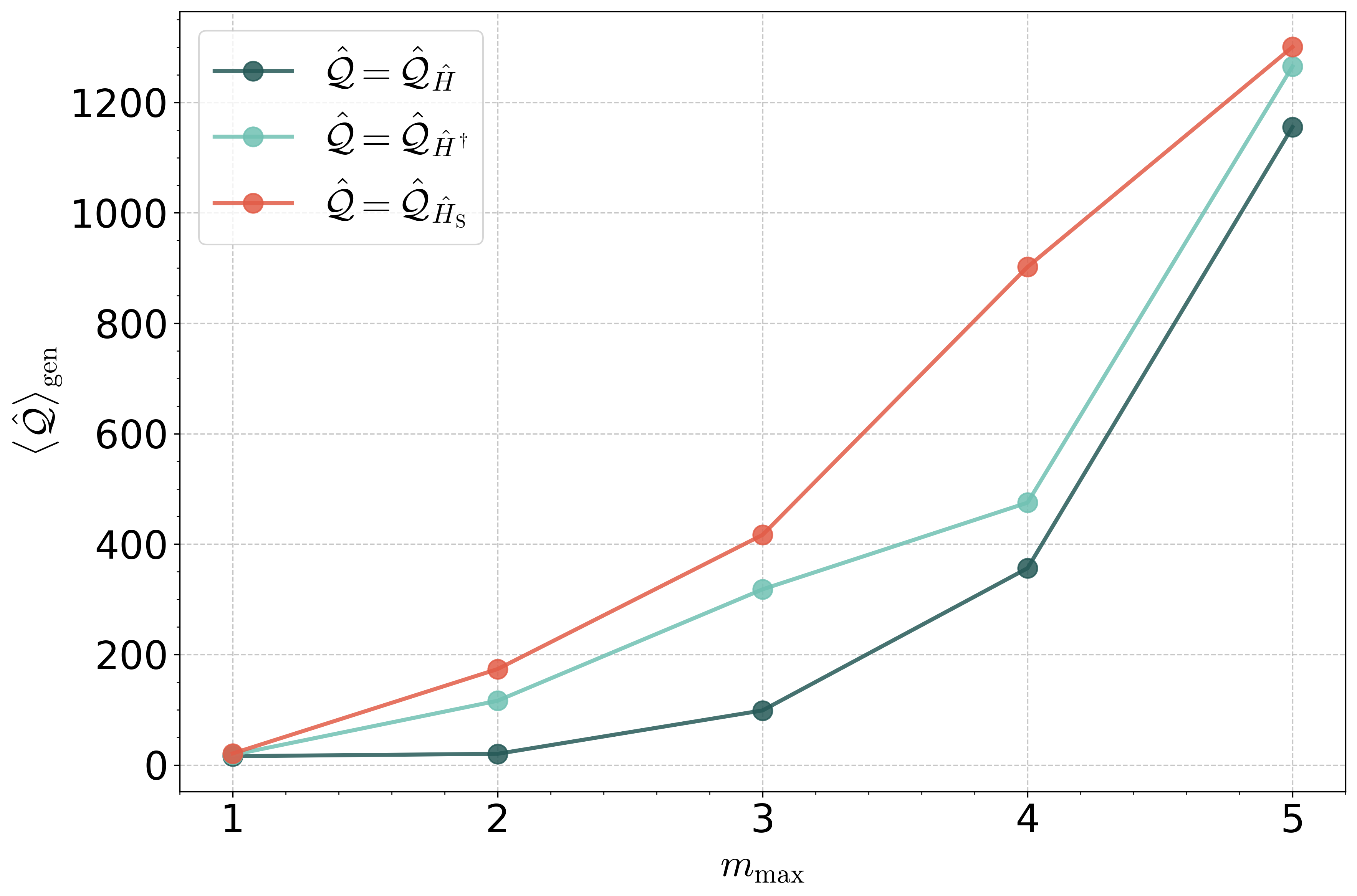}
    \caption{Typical variational magnitude of the constraints across cutoffs, obtained via the inverse-square proxy of Section \ref{subsubsec:genericvmagconst12}. At $m_\mathrm{max}=1$ this scale is of order $20$ to $30$ for the constraints shown, providing a reference against which $\expect{\hat{\mathcal{Q}}_{\hat{H}_\mathrm{S}}}$ of order $\sim 5$ can be interpreted as a non-trivial but not yet near-kernel suppression.}
    \label{fig:Qsym_scale}
\end{figure}
\newparagraph
A key question is then \emph{how close} these obtained variational states come to solving the underlying vertex constraints $\hat{H}_v$ and $\hat{H}_v^\dagger$ themselves, and whether they do so in a balanced manner. Figure \ref{fig:HvHdagger_box} summarises this directly by plotting, for each vertex, the ensemble distribution (across 75 different runs) of the Monte Carlo estimates of $\langle \hat{H}_v\rangle$ and $\langle \hat{H}^\dagger_v\rangle$ evaluated on the solution ensemble. In each panel, the box-and-whisker elements display run-to-run variability, while the filled marker indicates the ensemble mean at that vertex.
\newparagraph
Note that the filled-marker error bars in Figure \ref{fig:HvHdagger_box} are 95\% confidence intervals for the ensemble mean across runs computed from the empirical spread across runs, (not to be confused with the earlier area box-plots (e.g.\ the anisotropy and radius-consistency summaries) where the mean-marker error bars represented a propagated Monte Carlo standard error from within-run sampling uncertainty). The lower sub-panels in Figure \ref{fig:HvHdagger_box} show the distribution of MC errors reported by each run, which quantify sampling noise and should not be confused with physical spread.
\newparagraph
As seen in Figure \ref{fig:HvHdagger_box}, both $\langle \hat{H}_v\rangle$ and $\langle \hat{H}^\dagger_v\rangle$ are centred close to zero on all vertices in the solution ensemble of the symmetric constraint (small medians and small ensemble means). However, the run-to-run spread remains visibly broader than would be expected for an exact kernel state. This is further corroborated by the corresponding quantum variances, where even when the means are close to zero, $\mathrm{Var}(\hat{H}_v)$ and $\mathrm{Var}(\hat{H}_v^\dagger)$ in such solutions are frequently non-negligible (often larger than $\langle \hat{H}_v\rangle^2$ and $\langle \hat{H}^\dagger_v\rangle^2$ themselves), indicating that the state is not sharply concentrated in the strict kernel of either operator.
\begin{figure}[ht]
    \centering
    \includegraphics[width=0.48\textwidth]{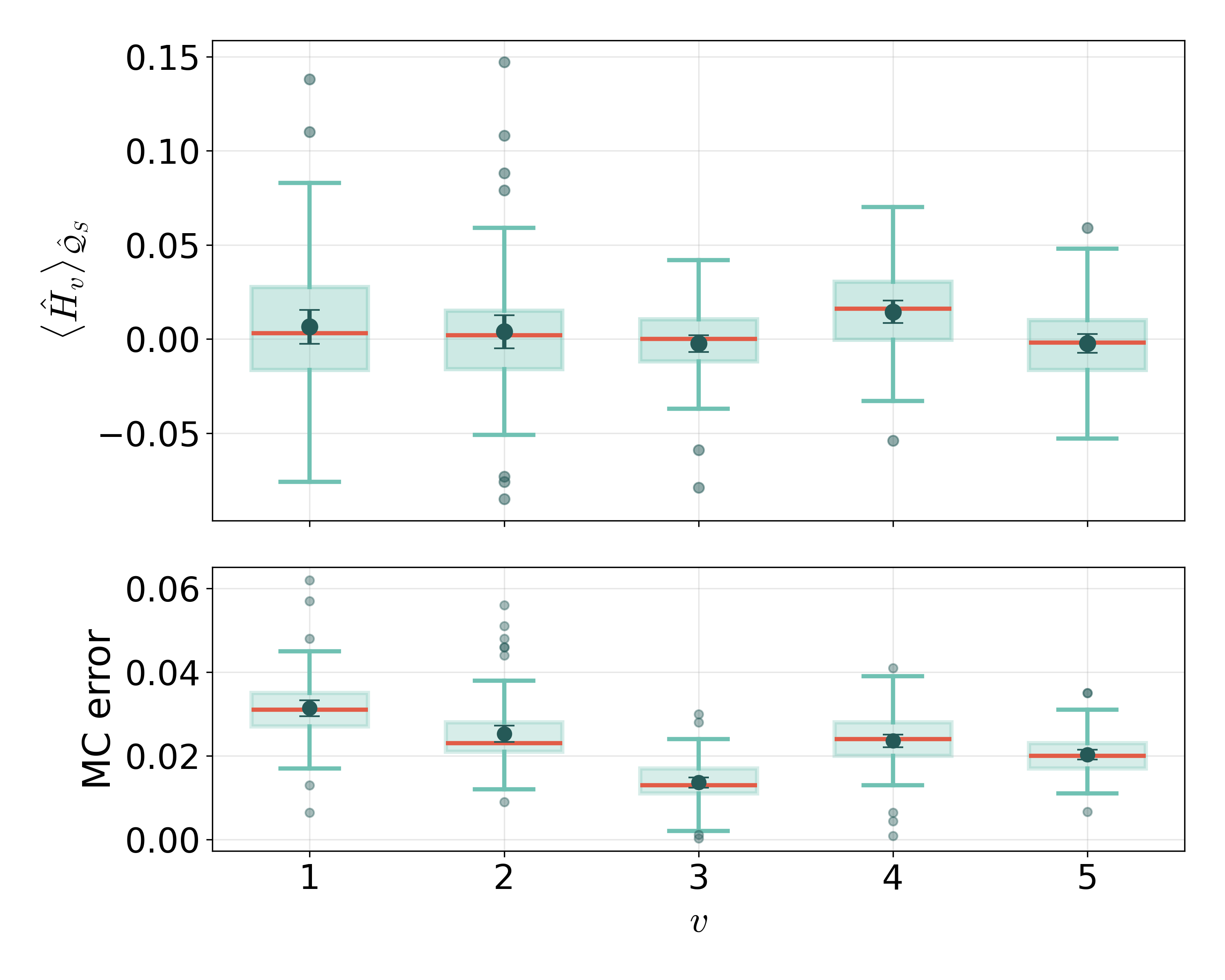}
    \includegraphics[width=0.48\textwidth]{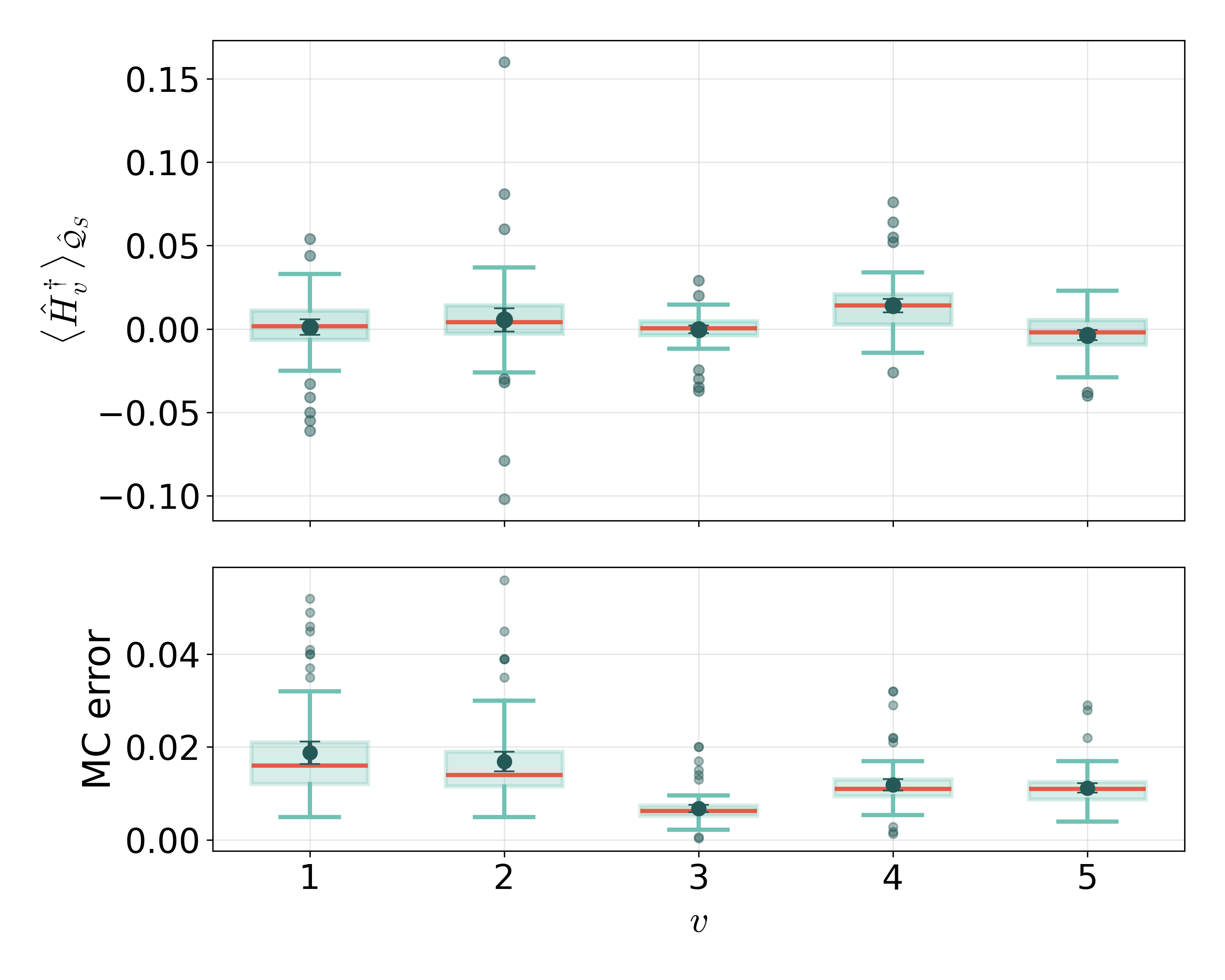}
    \caption{Constraint diagnostics on the $m_\mathrm{max}=1$ quasi-solution ensemble. On the left, the distribution across accepted runs of the per-vertex expectation values $\langle \hat{H}_v\rangle$. On the right, the analogous distribution for $\langle \hat{H}^\dagger_v\rangle$. In the upper panels, boxes show the interquartile range across runs with the median marked by the horizontal line while whiskers indicate the standard non-outlier range and points denote outliers. The filled marker is the sample mean across runs and its vertical error bar is a 95\% confidence interval for that mean (across-run uncertainty). The lower panels show the corresponding distributions of the reported Monte Carlo errors (sampling uncertainties) for each observable.}
    \label{fig:HvHdagger_box}
\end{figure}
\\\noindent
In this sense, these solutions are \emph{quasi}-solutions, as they balance suppression of both orderings in expectation, without being exact annihilation states. Consequently, they are also \emph{interpolant} solutions between the type-A and type-B families rather than as a third distinct family. Concretely, we observe the following mixed signatures:
\begin{enumerate}
    \item In the chromaticity diagnostics, the connected 2-point functions retain the characteristic peaked structure of the type-B charge sector, indicating non-trivial correlation patterns (Figure~\ref{fig:chromsymm}, left).
    
    \item In contrast to genuine type-B solutions, the resulting states are \emph{not} volume-degenerate, as vertex volume expectations are consistently non-zero, although typically smaller than in type-A solutions.
    
    \item The geometric reconstruction, however, remains type-A-like. The anisotropy probes show no systematic directional hierarchy, consistent with an approximately isotropic (``almost spherical'') interpretation at this coarse level.
    
    \item Additionally, minimal loop holonomies remain close to the flat sector, and correspondingly $\langle\hat{F}\rangle$ does not indicate dominant curvature excitations.
\end{enumerate}
Taken together, the modified objective yields quasi-solutions that combine type-B-like charge-sector correlations with type-A-like geometric flatness while avoiding strict volume degeneracy.
\newparagraph
Naturally, one might ask whether these interpolant quasi-solutions can be driven to either solution family (i.e. type-A or type-B). This can indeed be done by continuously biasing the symmetric constraint by introducing a two-parameter family
\begin{equation}
    \hat{\mathcal{Q}}_{\hat{H}_\mathrm{S}}(\alpha,\beta) := \sum_{v\in V(\gamma)}\bigl(\alpha\,\hat{H}_v + \beta\,\hat{H}^\dagger_v\bigr)^2,
\end{equation}
with the same projector and volume regulariser applied in the optimisation objective. Varying $(\alpha,\beta)$ allows one to steer the quasi-solution along a controlled path between more type-A-like and more type-B-like regimes while retaining the key desiderata (non-degenerate volume and near-flat holonomy behaviour).
\begin{figure}[ht]
    \centering
    \includegraphics[width=0.48\textwidth]{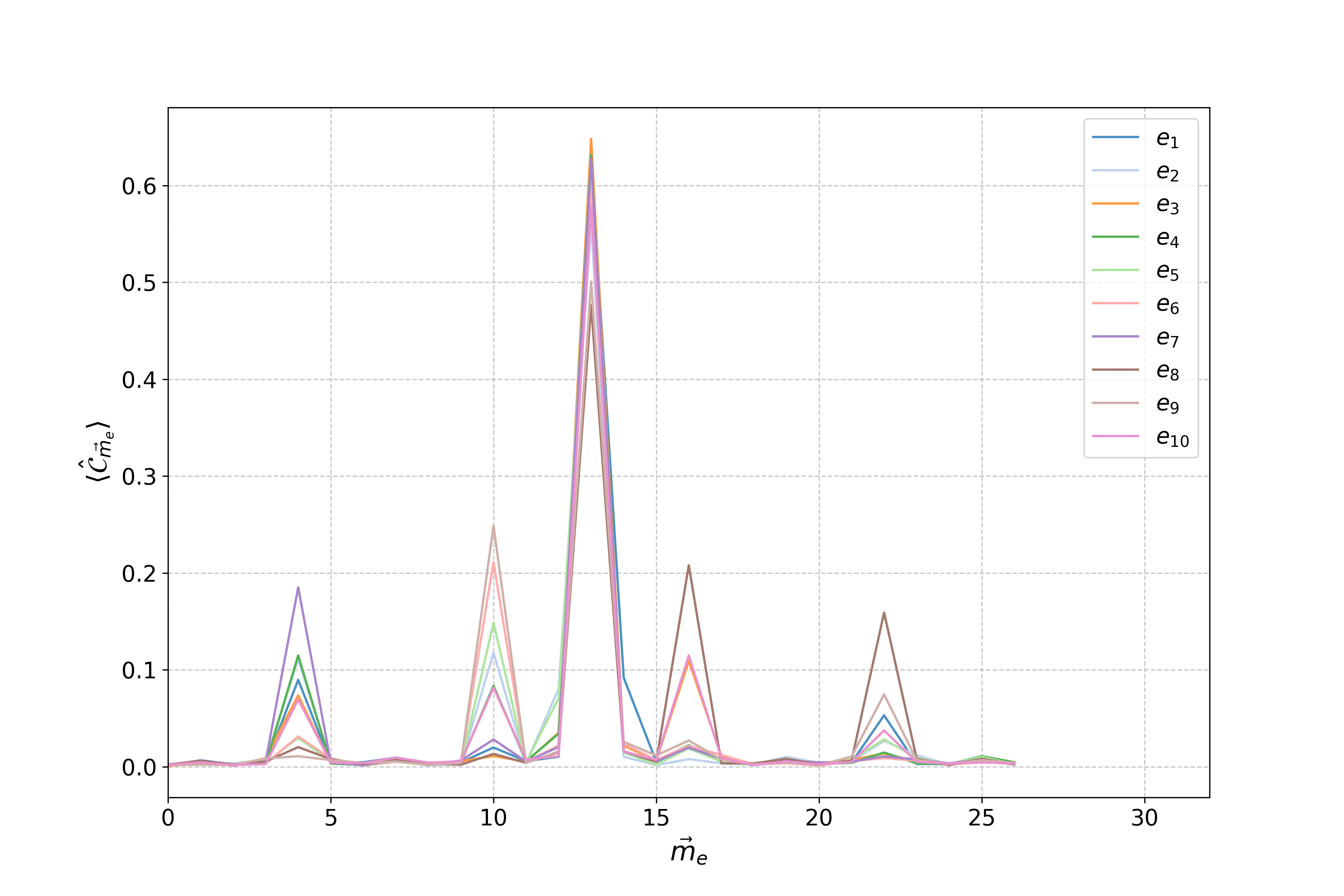}
    \includegraphics[width=0.48\textwidth]{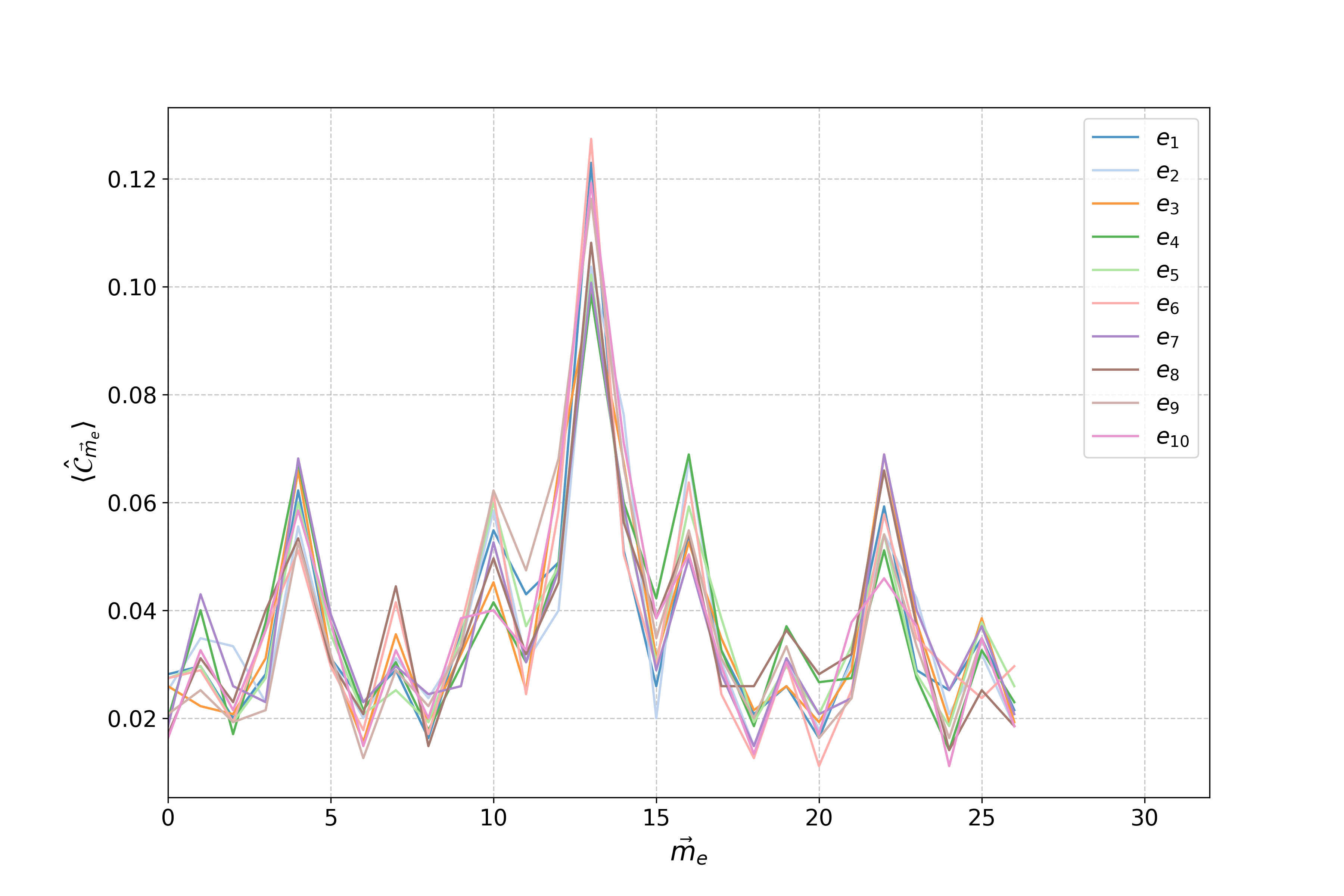}
    \caption{Chromaticity curves for the modified symmetric constraint objective at $m_\mathrm{max}=1$. \emph{Left:} quasi-solution obtained by minimising the symmetric constraint with a no-vacuum projector and an inverse-volume regulariser. \emph{Right:} corresponding result after $(\alpha,\beta)$-tuning using the biased family $\hat{\mathcal{Q}}_{\hat{H}_\mathrm{S}}(\alpha,\beta)$.}
    \label{fig:chromsymm}
\end{figure}
\\\noindent
An example of such tuning is shown in Figure \ref{fig:chromsymm} (right), where the chromaticity scale is pushed closer to the type-A characteristic magnitude while preserving remnants of the type-B peak structure. A systematic exploration of this $(\alpha,\beta)$ family across cutoffs and seeds is left for future work. Here we restrict to the $m_\mathrm{max}=1$ proof-of-principle demonstration.

\section{Discussion and outlook}
\label{sec:discussionoutlook}
This work advances a \emph{variational} approach to LQG. Namely, instead of restricting attention to symmetry-reduced sectors, ansätze, analytically tractable sectors or numerical methods restricted inherently by design to small systems, we employ NQS together with Monte Carlo methods. In doing so, we therefore not merely empirically obtain solution but also are able to probe their structure, at scale. This aspect is rather important since many physical questions one would ask in LQG are intrinsically questions about families of states rather than isolated examples.
\newparagraph
On the physics side, the dominant qualitative finding is that the quadratic, Master-constraint-like implementation of the vertex constraints on the $K_5$ graph supports distinct, robust solution families whose qualitative properties are dictated merely by operator ordering. In other words, at fixed graph and without symmetry reduction, ordering acts as a genuine sector selection mechanism whereby it does not simply affect optimisation speed or variance but can move the variational search into different parts of the physical state space.
\newparagraph
The sector assignment considered in this work is not inferred from a single observable, rather it is supported by a coherent pattern across holonomy and flatness diagnostics, volume (non-)degeneracy and chromaticity 1- and 2-point functions, together with geometric probes discussed throughout Section \ref{sec:results}. At the same time, despite the $K_5$ graph being significantly larger than graphs considered previously in similar settings \cite{Sahlmann:2024pba,Sahlmann:2024kat}, its nevertheless small size imposes a clear cautionary remark. Specifically, on this graph, long-range should be read operationally and not as a statement about an emergent correlation length in a large system. A second cautionary remark concerns what it means to be a solution in a quadratic constraint setting. Small values of a quadratic objective can hide imbalances between the constituent constraints and their adjoints. One may, for example, observe small expectation values while variances remain comparatively large for some of the constituent operators which would in turn indicate that the state is not sharply concentrated near the strict kernel of each constraint. This is precisely where a variational framework is valuable, as it allows one to diagnose such issues systematically and to design controlled modification (i.e. the symmetric constraint or multi-objective losses discussed in Section \ref{subsec:symmordsols}) to steer the variational search. 
\newparagraph
Looking forward, two directions appear both natural and technically well-motivated, the first of which variational methodology and scalability. The main difficulty is increasingly not necessarily the implementation of estimators but rather the design of ansätze and training protocols that reliably (a) locate the relevant solution manifolds, (b) avoid branch bias and mode collapse, and (c) provide principled uncertainty quantification for correlation and geometry diagnostics. This motivates systematic architecture, sampler and estimator studies and explicitly constraint-resolved objectives. As large-scale simulations are now, for the first time, generally accessible, such technical issues take more of a centre stage. When addressed, however, as done in this work one sees that variational LQG is not only a computational tool but a framework for model analysis which allows us to not only obtain but also physically characterise the solutions of constraints.
\newparagraph
Second, the present results strongly suggests that ordering issues will remain first-order in the non-Abelian theory, but the technical and conceptual stakes will increase. On the technical side, the local Hilbert spaces become more complicated due to the now intertwiners on vertices of the graph introducing a non-trivial recoupling structures. Evaluating constraint actions is expected to therefore require substantially more expensive algebra than the Abelian model. This will stress both estimator variance and optimisation stability and will likely require improved operator kernels, variance-reduction techniques and more expressive inductive biases in the chosen ansatz (e.g. explicit gauge-equivariance or invariance, graph symmetries and attention to intertwiner degrees of freedom). 
\newparagraph
On the conceptual side, the interplay between ordering, representation and normalisability becomes even sharper. Sector-selection by ordering may, as in the case of this work, map onto physically different notions vacuum and semiclassicality. Deciding which sector is appropriate must ultimately be tied to the intended physical inner product and continuum interpretation. A practical lesson from the present work is therefore that ordering should not be treated as an afterthought when moving to the SU(2) case, but rather as a knob that can control which sector the variational search sees and hence, which physics is extracted.
\newparagraph
Further, it would be interesting to connect the variational numerical picture developed here with analytical work on  the $\mathrm{U}(1)^3$-limit of LQG. In particular, \cite{Long:2022wcLQGArea} provides a concrete route by which $\mathrm{U}(1)^3$ effective dynamics can be related to the $\mathrm{SU}(2)$ theory through a parametrisation of $\mathrm{SU}(2)$ holonomy-flux variables by $\mathrm{U}(1)^3$ data and a coherent-state path integral description. Could one extract this effective dynamics also from the numerical solutions considered in the present work? Likewise, recent Dirac quantised treatment of the $\mathrm{U}(1)^3$ model in \cite{Bakhoda:2025U1cubedVolume} shows how one may extract a Schrödinger-like evolution, a physical Hilbert space as well as a notion of geometric quantum time directly from the Hamiltonian constraint. It may be possible to implement this notion of time also numerically and thus obtain the corresponding evolution. These developments suggest that the variational programme pursued here could also serve as a concrete framework in which ideas such as effective dynamics and quantum time may be implemented and explored.
\newparagraph
With this work, the pure gravity $\mathrm{U}(1)^3$ model of Euclidean LQG is shown to be amenable to variational methods, which they themselves offer several approaches to not only obtain solutions but also characterise them. While one can now consider different $\mathrm{U}(1)^N$ based models (e.g. spherically symmetric models which are compatible with such degrees of freedom or even quantum reduced loop gravity models), the next natural step in this programme, aside from the mentioned relation to analytical results in this model, is perhaps to consider the full SU(2) theory, where both such technical and conceptual aspects discussed above will be considered in future work.

\section{Conclusion}
\label{sec:conclusion}
This work set out to make the space of physical (near-kernel) states of 4-d Euclidean loop quantum gravity in Smolin's weak-coupling limit empirically accessible on a non-trivial graph, without symmetry assumptions. Neural network quantum states with variational Monte Carlo methods were used to minimise quadratic, Master-like constraints on the $\mathrm{U}(1)^3$ gauge invariant subspace of the kinematical Hilbert space of the oriented $K_5$ graph. The kinematical degrees of freedom of the representation labels associated to the edges were truncated by imposing a cutoff $m_\mathrm{max}$ such that the $\mathrm{U}(1)^3$ group is replaced by the quantum $\mathcal{U}_q(1)^3$. This work has mainly focused on cutoffs of $m_\mathrm{max} = 1, \ldots, 5$ for the main analysis and for scaling purposes demonstrated convergence for cutoffs of $m_\mathrm{max} = 50$ and $100$. We then analysed the resulting variational states using correlation probes, geometric operators and normalisability diagnostics. 
\newparagraph
The central result is that the factor ordering in the quadratic constraint is not a benign implementation choice. The two orderings $\hat{\mathcal{Q}}_{\hat{H}} = \sum_{v \in V(\gamma)} \hat H_v \hat H_v^\dagger$ and $\hat{\mathcal{Q}}_{\hat{H}^\dagger} = \sum_{v \in V(\gamma)} \hat H_v^\dagger \hat H_v$ generically drive the variational optimisation process to select two distinct and robust solution families with sharply different physical signatures. The type-A family (from $\hat{\mathcal{Q}}_{\hat{H}}$ is characterised by non-degenerate volume, near flatness on minimal loops and weak connected chromaticity correlators consistent with an Dittrich-Geiller-like vacuum sector on the fixed graph. The type-B family (from $\hat{\mathcal{Q}}_{\hat{H}^\dagger}$) exhibits strong chromaticity signatures together with volume degeneracy and altered normalisability behaviour, consistent with a Ashtekar-Lewandowski-like vacuum sector. A symmetrised modification of the quadratic constraint, together with a projected ansatz and an inverse-volume penalty term intended to maximise volume on a selected vertex, produces quasi-solutions interpolating between these behaviours, thus making the ordering dependence itself a controllable handle on the explored sector.
\newparagraph
In all, the results show that variational LQG is not only a viable route to \emph{obtaining} candidate physical states on non-trivial graphs, but also as a practical route to \emph{characterising} them at scale via correlators and geometric observables. The observed ordering-induced sector selection provides a concrete constraint on how one should interpret and compare numerical solutions of Master-like dynamics and it motivates carrying the same variational and diagnostic programme to the $\mathrm{SU}(2)$ theory.

\ack
The authors are grateful for comments and discussions with Thomas Thiemann, especially regarding the connection between solutions of the Hamilton constraint and minimum energy solutions, and the master constraint program. H.S. thanks B. Dittrich for hospitality at PI, and for interesting comments and discussions, and he acknowledges the contribution of the COST Action CA23130. This research was supported in part by Perimeter Institute for Theoretical Physics. Research at Perimeter Institute is supported by the Government of Canada through the Department of Innovation, Science and Economic Development and by the Province of Ontario through the Ministry of Research, Innovation and Science. The authors gratefully acknowledge the scientific support and HPC resources provided by the Erlangen National High Performance Computing Center (NHR@FAU) of the Friedrich-Alexander-Universität Erlangen-Nürnberg (FAU). The hardware is funded by the German Research Foundation (DFG). Simulations for this work were performed using \neuralqx \cite{Sherif:2026neuralqx}. \neuralqx is a high-performance simulations toolkit for loop quantum gravity built on top of \netket \cite{Carleo:2019netket,Vicentini:2022netket3}. The implementation relies on \jax \cite{Bradbury:2018jax} and \flax \cite{Heek:2024flax}. MPI-distributed execution used \mpijax \cite{Haefner:2021mpi4jax}.

\clearpage

\section*{Appendices}
\renewcommand{\thesubsection}{\Alph{subsection}} 
\setcounter{footnote}{0}

\subsection{State visualisation in large Hilbert spaces}
\label{app:sequential_indexing_and_plotting}
Typical Hilbert space dimensions, even in the gauge invariant subspace, considered in this work make it prohibitive to conduct any visualisation of state vectors (i.e. a plot of all complex amplitudes of the obtained variational near-kernel states). In what follows, we detail an efficient approximate algorithm which enables one to structurally visually inspect all amplitudes in a controlled but efficient manner which requires a fractional amount of computing resources compared to standard methods.

\subsubsection{Sequential labelling of charge-network basis states}
\label{app:sequential_labelling}
Throughout this work, we consider the truncated kinematical charge-network basis
\begin{equation}
    \ket{\vec m_1 \,\vec m_2 \cdots \vec m_{n_e}} =\ket{\vec m_1}\otimes\ket{\vec m_2}\otimes\cdots\otimes\ket{\vec m_{n_e}},
\end{equation}
where $\vec m_e=\bigl(m^{(1)}_e,m^{(2)}_e,m^{(3)}_e\bigr)\in M^3$, $M=\{-m_{\max},\ldots,m_{\max}\}\subset\mathbb{Z}$ and $n_e:=|E(\gamma)|$ (in particular, $n_e = 10$ for the present $K_5$ graph considered in this work). The corresponding truncated kinematical Hilbert space has dimension
\begin{equation}
    D := \dim\mathscr{H}_\mathrm{kin}^\mathrm{GI} = (2m_\mathrm{max} + 1)^{3n_e}.
\end{equation}
To visualise solutions over such spaces, we require a deterministic identification of computational basis states with integers (i.e. a bijective sequential labelling). Concretely, we introduce a bijection
\begin{equation}
    \iota : \mathscr{B}^{(m_\mathrm{max})} \rightarrow \{0, 1, \ldots, D-1\},
\end{equation}
where $\mathscr{B}^{(m_\mathrm{max})}$ denotes the ordered basis set at cutoff $m_\mathrm{max}$. In the present implementation of neuraLQX \cite{Sherif:2026neuralqx}, basis elements are stored in a gauge-strided layout done by concatenating the three $\mathrm{U}(1)$ components across all edges
\begin{equation}
    \sigma \equiv \sigma(\{\vec m_e\}_{e\in E(\gamma)}) =\Bigl(m^{(1)}_1,\ldots,m^{(1)}_{n_e}\ \Big|\ m^{(2)}_1,\ldots,m^{(2)}_{n_e}\ \Big|\ m^{(3)}_1,\ldots,m^{(3)}_{n_e}\Bigr)\in M^{3n_e}.
\label{eq:gauge_strided_sigma}
\end{equation}
Thus, one may equivalently regard states in $\mathscr{H}_\mathrm{kin}^{m_\mathrm{max}}$  as the tensor product of $N := 3n_e$ discrete local degrees of freedom each taking values in $M$. Let $b := |M| = 2m_\mathrm{max} + 1$ and define the local index map
\begin{equation}
    \ell : M \rightarrow \{0, 1, \ldots, b-1\} \quad,\quad \ell(m) := m + m_\mathrm{max}.
\end{equation}
Writing $\sigma := (\sigma_1, \ldots, \sigma_M)$ with $\sigma_j \in M$, the sequential label of $\sigma$ is then given by the mixed-radix expansion
\begin{equation}
    \iota(\sigma) = \sum_{j = 1}^N \ell(\sigma_j)b^{N-j},
    \label{eq:big_endian_indexing}
\end{equation}
Equivalently, $\sigma_1$ is the most significant digit (it changes slowest as $\iota$ increases) while $\sigma_N$ is the least significant digit. The inverse map $\iota^{-1}$ is obtained by repeated Euclidean division. Given $i \in \{0, 1, \ldots, D-1\}$, define integers $q_0 := i$ and $q_{j-1} = bq_j + r_j$ for $j = 1, \ldots, N$, with $r_j \in \{0, \ldots, b-1\}$. Then, $\ell(\sigma_j) = r_j$ and therefore $\sigma_j = r_j - m_\mathrm{max}$. Finally the reconstructed configuration $\sigma$ may be regrouped into edge charge vectors according to equation \eqref{eq:gauge_strided_sigma}. 
\newparagraph
When restricting to a gauge-invariant subspace, the sequential labelling remains a bijection between $\{0, \ldots, \dim\mathscr{H}_\mathrm{kin}^{\mathrm{GI}, (m_\mathrm{max})} - 1\}$ and the \emph{allowed} configurations, but the map is no longer given by the closed form of equation \ref{eq:big_endian_indexing}. Instead, $\iota^{-1}(i)$ is given as the $i$-th admissible configuration in the inherited lexicographic ordering from the parent space (i.e. the ordering induced by equation \eqref{eq:big_endian_indexing} prior to imposing the Gauß constraint).

\subsubsection{Stratified-envelope representation of large wavefunctions}
\label{subsubsec:fast_envelope_sampling}
Let $\Psi \in \mathscr{H}_\mathrm{kin}^{\mathrm{GI}, m_\mathrm{max}}$ be a complex variational state with expansion coefficients $\psi_i$ with the index label $i$ is ordered according to the lexicographic ordering discussed in the prior section. In our setting, $D$ grows beyond any feasible explicit enumeration already at modest $m_\mathrm{max}$ and hence it is impossible to directly plot $\Re(\psi_i)$ or $\Im(\psi_i)$ for all $i$. Instead, we represent the wavefunction by an \emph{envelope} over a coarse partition of index space, constructed by stratified sampling with unbiased global normalisation.
\newparagraph
To begin, fix an integer $B \ll D$ and partition $\{0, 1, \ldots, D-1\}$ into $B$ contiguous bins
\begin{equation}
    \mathcal{I}_b := \{i \in \mathbb{Z} \,:\, K_b \leq i < U_b\}\quad,\quad b = 0, \ldots, B-1,
\end{equation}
where $L_b$ and $U_b$ satisfy 
\begin{equation}
    0 = L_0 < U_0 = L_1 < \cdots < U_{B-1} = D
\end{equation}
and $W_B := |\mathcal{I}_b| = U_b - L_b$. For each bin, we define the bin-centre index
\begin{equation}
    x_b := \left\lfloor \frac{L_b + U_b}{2} \right\rfloor.
\end{equation}
Now, choose an integer $K$ representing the number of samples per bin and draw i.i.d samples within each stratum
\begin{equation}
    i_{b,1}, \ldots, i_{b, K} \sim \text{Uniform}(\mathcal{I}_b).
\end{equation}
Using the sequential decoding map $\iota^{-1}$, each sampled index $i_{b, t}$ defines a basis configuration $\sigma_{b, t} := \iota^{-1}(i_{b,t})$ and hence a model amplitude. Since the NQS is naturally evaluated in terms of log-amplitudes, we write the generally unnormalised complex coefficient as
\begin{equation}
    \tilde{\psi}_{b,t} := \exp[\log \psi_\theta(\sigma_{b,t})] \in \mathbb{C},
\end{equation}
where here $\psi_\theta( \cdot)$ denotes the network output for the given configuration. The exact norm of the unnormalised amplitudes is
\begin{equation}
    \tilde{\Psi}^2 = \sum_{i = 0}^D |\tilde{\psi}_i|^2 = \sum_{b = 0}^{B-1} \sum_{i \in \mathcal{I}_b} |\tilde{\psi}_i|^2.
\end{equation}
Define the per-bin weight $w_b := W_B / K$. Then, the stratified estimator
\begin{equation}
    \widehat{\|\tilde{\Psi}\|^2} := \sum_{b = 0}^{B-1} w_b \sum_{t=1}^K |\tilde{\psi}_{b,t}|^2,
    \label{eq:stratified_norm_estimator}
\end{equation}
satisfies $\mathbb{E}[\widehat{\|\tilde{\Psi}\|^2}] = \|\tilde{\Psi}\|^2$ (i.e. it is unbiased). We then set $\widehat{\|\tilde{\Psi}\|} = \sqrt{\widehat{\|\tilde{\Psi}\|^2}}$, and define the approximately normalised sampled amplitudes
\begin{equation}
    \hat{\psi}_{b, t} := \frac{\tilde{\psi}_{b,t}}{\widehat{\|\tilde{\Psi}\|}}.
\end{equation}
For numerical stability (given the exponentially enormous dynamic range of $|\tilde{\psi}_{b,t}|$), the estimator \eqref{eq:stratified_norm_estimator} is evaluated in log-space via the identity
\begin{equation}
    w_b |\tilde{\psi}_{b,t}|^2 = \exp[2\Re(\log \psi_\theta(\sigma_{b,t})) + \log w_b],
\end{equation}
followed by a stable $\log\sum\exp$ reduction.
\newparagraph
Now, for each bin $b$ we estimate the extrema of the real and imaginary parts by the corresponding sample extrema
\begin{align}
    \widehat{m}_b^\Re := \min_{t = 1, \ldots, K} \Re(\hat{\psi}_{b,t}) \quad &,\quad \widehat{M}_b^\Re := \max_{t = 1, \ldots, K} \Re(\hat{\psi}_{b,t}), \\
    \widehat{m}_b^\Im := \min_{t = 1, \ldots, K} \Im(\hat{\psi}_{b,t}) \quad &,\quad \widehat{M}_b^\Im := \max_{t = 1, \ldots, K} \Im(\hat{\psi}_{b,t}).
\end{align}
The envelope plots shown in the main text are then the filled bands between $(x_b, \widehat{m}_b^\Re)$ and $(x_b, \widehat{M}_b^\Re)$ (and similarly for the imaginary part), with linear interpolation between neighbouring bin-centres (See Section \ref{subsubsec:streaming_envelopes} below for more details). For orientation, we also display the midrange curve
\begin{equation}
    \widehat{\mu}_b^\Re := \frac{1}{2}(\widehat{m}_b^\Re + \widehat{M}_b^\Re) \quad, \quad \widehat{\mu}_b^\Im := \frac{1}{2}(\widehat{m}_b^\Im + \widehat{M}_b^\Im),
\end{equation}
which provides a compact visual proxy for the bin-local spread.
\newparagraph
While the estimator \eqref{eq:stratified_norm_estimator} is unbiased, at finite $K$ it is a random quantity whose fluctuations controls the uncertainty of the global vertical scale in the state-vector plots. One can write indeed compute variance
\begin{equation}
    \mathrm{Var}(\widehat{\|\tilde{\Psi}\|^2}) = \sum_{b = 0}^{B-1} \frac{W_b^2}{K} \mathrm{Var}(|\tilde{\psi}_{b,t}|^2),
\end{equation}
as well as the standard error
\begin{equation}
    \mathrm{SE}(\widehat{\|\tilde{\Psi}\|^2}) = \sqrt{\mathrm{Var}(\widehat{\|\tilde{\Psi}\|^2})},
\end{equation}
and as such the relative error decreases as $K^{-1/2}$ (for fixed $B$) provided the $\mathrm{Var}(|\tilde{\psi}_{b,t}|^2)$ are finite. In practice, the latter can be estimated from the same stratified samples via plug-in sample variances, yielding an internal estimate of the standard error and thus of relative uncertainty $\mathrm{SE}(\widehat{\|\tilde{\Psi}\|^2}) / \widehat{\|\tilde{\Psi}\|^2}$. Since the plotted amplitudes are normalised by $\widehat{\|\tilde{\Psi}\|}$, one can show that moderate uncertainty i $\widehat{\|\tilde{\Psi}\|^2}$ only induces a \emph{global} multiplicative rescaling of the entire envelope by a comparable relative factor. Thus, unless one is using the envelope to make claims about absolute amplitude scales across different states, the dominant qualitative features of interest are largely insensitive to this normalisation error. 
\newparagraph
The global vertical scale thus converges (in probability) to the true one as $BK \rightarrow \infty$. By contrast, the min/max envelope is \emph{not} an unbiased estimator of the true per-bin extrema. Namely, if rare spikes occupy a tiny fraction of $\mathcal{I}_b$, a finite $K$ may miss them. Hence, the envelope should be interpreted as a controlled diagnostic representation whose fidelity increases monotonically with $K$ and whose purpose is to detect \emph{large-scale structure} (i.e. decay, concentration, etc.) across index space without requiring any enumeration of the full space.

\subsubsection{Streaming segment envelopes for 1-dimensional signals}
\label{subsubsec:streaming_envelopes}
In addition to the stratified regime above, it is often useful to compress an arbitrary discrete signal
\begin{equation}
    y : \{0, 1, \ldots, N - 1\} \rightarrow \mathbb{R},
\end{equation}
into an envelope representation with at most $K \ll N$ points \emph{without approximation of the segment statistics}. This is achieved by partitioning the index domain into $K$ contiguous segments
\begin{equation}
    S_k := \{i \in \mathbb{Z} \,:\, s_k \leq i < e_k\}
\end{equation}
with $k = 0, \ldots, K-1$, $0 = s_0 < e_0 = s_1 < \cdots e_{K-1} = N$ and $|S_k| \in \{\lfloor N/K\rfloor, \lceil N/K \rceil\}$. On each segment, define the exact per-segment statistics
\begin{equation}
    m_k := \min_{i \in S_k} y_i \quad,\quad M_K := \max_{i \in S_k} y_i \quad,\quad \mu_k := \frac{1}{|S_k|} \sum_{i \in S_k} y_i,
\end{equation}
and the representative abscissa (segment midpoints)
\begin{equation}
    x_k := \frac{s_k + e_k}{2}.
\end{equation}
The corresponding envelope plot consists of the band between the polylines through $(x_k, m_k)$ and $(x_k, M_k)$, optionally accompanied by the mean polyline through $(x_k, \mu_k)$. Crucially, the above statistics can be computed \emph{streamingly} (i.e. by processing contiguous blocks of the data without ever storing the full length-$N$ array). To see this, suppose the data arrive as chunks supported on disjoint index ranges $\{\tau, \tau + 1, \ldots, \tau + L - 1\}$. Each sample index $i$ is deterministically assigned to a unique segment label $k(i)$ defined by  $k(i) := \min \{k \,:\, e_k > i\}$. One then maintains per-segment accumulators $(m_k, M_k, S_k, C_k)$ initialised by
\begin{equation}
    m_k = + \infty, \quad M_k = -\infty, \quad S_k = 0, \quad C_k = 0,
\end{equation}
and updates them for each processed chunk by the associative reductions
\begin{equation}
    m_k \leftarrow \min (m_k, \min_{i \in \text{chunk}, \, k(i) = k}y_i), \quad M_k \leftarrow \max (M_k, \max_{i \in \text{chunk}, \, k(i) = k}y_i),
\end{equation}
as well as
\begin{equation}
    S_k \leftarrow S_k + \sum_{i \in \text{chunk}, \, k(i) = k} y_i, \quad C_k \leftarrow C_k + \#\{i \in \text{chunk} \,:\, k(i) = k\}.
\end{equation}
After all chunks are processed, one sets $\mu_k := S_k/C_k$. Because $\min$, $\max$ and addition are associative, the final statistics coincide with the monolithic (non-streaming) computation up to floating point round-off. Thus, this streaming approach is a memory-management strategy and does not introduce an algorithmic approximation.

\subsection{Proof of Propositions \ref{pr:main_proposition1}, \ref{pr:main_proposition2}}
\label{app:proof_minimization_problem}
\begin{proof}[Proof of proposition \ref{pr:main_proposition1}]
We assume that there is $|\Psi) \in \mathcal{H}_\text{diff}$ such that for a certain $\gamma_0$
\begin{equation}
\label{eq:necessary_cond}
    \left (\Psi \right |\left .  \hat{H}(N) f_{\gamma_0} \right \rangle = 0 \text{ for all } N: \Sigma \mapsto \mathbb{R}, f_{\gamma_0}\in \mathcal{H}_{\gamma_0}. 
\end{equation}
Note that any diffeomorphism invariant state, and hence in particular $|\Psi)$ can be written as 
\begin{equation}
    |\Psi) = \sum_{\gamma' \in S} \left \lvert \eta(\Psi_{\gamma'}) \right ), \qquad \Psi_{\gamma'} \in  \mathcal{H}'_{\gamma'} 
\end{equation}
where the sum is over some $\gamma'$ collected in some set $S$ and $\Psi_{\gamma'}\in \mathcal{H}'_\gamma$ and we can assume without loss of generality that $\Psi_{\gamma'}$ is invariant under the symmetries of  $\gamma'$. 
\newparagraph
Some of the $\gamma \in S$ may be diffeomorphic to $\gamma_0$ or a proper subgraph of $\gamma_0$. Let $S_{\gamma_0} \subset S$ be the set of such $\gamma$. For a $\gamma \in S_{\gamma_0}$, there may even be diffeomorphism classes $\varphi_{\gamma}, \varphi'_\gamma, \ldots \in \text{Diff}/\text{Diff}_\gamma$ that map $\gamma$ to \emph{different} subgraphs of $\gamma_0$, or to the same subgraph in different ways. For a given subgraph $\alpha$ of $\gamma_0$, let 
\begin{equation}
    \widetilde\Psi_{\alpha} := \sum_{\gamma \in S_{\gamma_0}, \varphi_\gamma:  \varphi_\gamma(\gamma) =\alpha } U_{\varphi_\gamma} \Psi_\gamma.  
\end{equation}
If $\alpha$ is not diffeomorphic to any graph in $S_{\gamma_0}$, we take $\widetilde\Psi_{\alpha}=0$ under this definition.  

By definition of the diffeomorphism invariant inner product \cite{Ashtekar:1995diffeo,Ashtekar:2004eh}, we have
\begin{align}
    (\Psi  |  \hat{H}(N) f_{\gamma_0} \rangle 
    &=\sum_{\gamma\in S} \sum_{\varphi\in \sfrac{\text{Diff}}{\text{Diff}_\gamma}} \scpr{U_\varphi\Psi_\gamma}{\hat{H}(N) f_{\gamma_0}}\\
    &=\sum_{\gamma\in S_{\gamma_0}} \sum_{\varphi=\varphi_{\gamma},\varphi'_{\gamma}, \ldots} \scpr{U_\varphi\Psi_\gamma}{\hat{H}(N) f_{\gamma_0}}\\
    &= \sum_{\alpha \subseteq \gamma_0} \scpr{\widetilde{\Psi}_\alpha}{\hat{H}(N) f_{\gamma_0}}\\
    &= \scpr{\widetilde{\Psi}}{\hat{H}(N) f_{\gamma_0}}
\end{align}
where for the second equality we have made use of the properties of the kinematical scalar product: all other terms vanish. In the second step we used the definitions from above. And in the third step, we introduced 
\begin{equation}
    \widetilde{\Psi}= \sum_{\alpha \subseteq \gamma_0} \widetilde{\Psi}_\alpha. 
\end{equation}
Hence 
\begin{align}
    \text{\eqref{eq:necessary_cond}} \quad &\Leftrightarrow \quad \scpr{\widetilde{\Psi}}{\hat{H}(N) f_{\gamma_0}}
    \text{ for all } N: \Sigma \mapsto \mathbb{R}, f_{\gamma_0}\in \mathcal{H}_{\gamma_0}\\
    &\Leftrightarrow \quad \scpr{\hat{H}(N)^\dagger \widetilde{\Psi}}{ f_{\gamma_0}} \text{ for all } N, f_{\gamma_0}\in \mathcal{H}_{\gamma_0}\\
    &\Leftrightarrow \quad P_{\mathcal{H}_{\gamma_0}} \hat{H}(N)^\dagger\widetilde{\Psi} \text{ for all } N,
\end{align}
where in the last step $P_{\mathcal{H}_{\gamma_0}}$ is the projector onto $\mathcal{H}_{\gamma_0}$ and we have used the fact that the cylindrical functions on $\gamma_0$ are dense in $\mathcal{H}_{\gamma_0}$. 
We decompose $H^\dagger(N)$ according to $\mathcal{H}= \mathcal{H}_{\gamma_0} \oplus \mathcal{H}_{\gamma_0}^\perp$: 
\begin{equation}
   \hat{H}(N)=\begin{pmatrix}\hat{H}(N)\rvert_{\mathcal{H}_{\gamma_0}}& \hat{X}(N)\\ 0 & \hat{Y}(N)\end{pmatrix}, 
\end{equation}
hence
\begin{equation}
     P_{\mathcal{H}_{\gamma_0}}  \hat{H}(N)^\dagger \widetilde{\Psi} = \begin{pmatrix}
         \one & 0\\ 0&0
     \end{pmatrix}
     \begin{pmatrix}
         \hat{H}(N)\rvert^\dagger_{\mathcal{H}_{\gamma_0}}& 0\\
         \hat{X}(N)^\dagger& \hat{Y}(N)^\dagger
     \end{pmatrix}
     \begin{pmatrix}
         \widetilde{\Psi}\\0
     \end{pmatrix} 
     = \begin{pmatrix}
     \hat{H}(N)\rvert^\dagger_{\mathcal{H}_{\gamma_0}}\widetilde{\Psi}\\0
     \end{pmatrix}.  
\end{equation}
Hence we can continue the string of equivalences. 
\begin{align}
    \text{\eqref{eq:necessary_cond}} \quad &\Leftrightarrow \quad \hat{H}(N)\rvert^\dagger_{\mathcal{H}_{\gamma_0}}\widetilde{\Psi} =0  \text{ for all } N \\
    &\Leftrightarrow \quad \hat{H}_v\rvert^\dagger_{\mathcal{H}_{\gamma_0}}\widetilde{\Psi} =0  \text{ for all } v\in \gamma_0 \label{eq:kernel}\\
    &\Rightarrow \quad \hat{H}_v\rvert_{\mathcal{H}_{\gamma_0}} \hat{H}_v\rvert^\dagger_{\mathcal{H}_{\gamma_0}}\widetilde{\Psi} =0  \text{ for all } v\in \gamma_0. \label{eq:quadratic_constraint}
\end{align}
On the other hand
\begin{align}
    \text{\eqref{eq:quadratic_constraint}} \quad &\Rightarrow \quad \scpr{\hat{H}_v\rvert_{\mathcal{H}_{\gamma_0}} \hat{H}_v\rvert^\dagger_{\mathcal{H}_{\gamma_0}}\widetilde{\Psi}}{\widetilde{\Psi}} = 0 \text{ for all } v\in \gamma_0  \label{eq:expect}\\
    &\Leftrightarrow \quad \norm{\hat{H}_v\rvert^\dagger_{\mathcal{H}_{\gamma_0}}\widetilde{\Psi}}^2=0 \text{ for all } v\in \gamma_0 \\
    &\Leftrightarrow \quad \text{\eqref{eq:kernel}}.  
\end{align}
All in all, we have a chain of equivalences between \eqref{eq:necessary_cond} and \eqref{eq:expect}. Finally
\begin{equation}
    \scpr{\hat{H}_v\rvert_{\mathcal{H}_{\gamma_0}} \hat{H}_v\rvert^\dagger_{\mathcal{H}_{\gamma_0}}\widetilde{\Psi}}{\widetilde{\Psi}} = 0 \text{ for all } v\in \gamma_0 
    \quad\Leftrightarrow \quad \langle \hat{\mathcal{Q}}_{\hat{H}\rvert_{\mathcal{H}_{\gamma_0}}} \rangle_{\widetilde{\Psi}} = 0
\end{equation}
for 
\begin{equation}
    \hat{\mathcal{Q}}_{\hat{H}\rvert_{\mathcal{H}_{\gamma_0}}} = \sum_{v \in \gamma_0} \hat{H}_v\rvert_{\mathcal{H}_{\gamma_0}} \hat{H}_v\rvert_{\mathcal{H}_{\gamma_0}}^\dagger 
\end{equation}
since each term in the sum for the expectation value is non-negative. \end{proof}

\begin{proof}[Proof of proposition \ref{pr:main_proposition2}]
In preparation for the main part of the proof, we note the following simple fact: For the projector $P'_\gamma$ on $\mathcal{H}'_\gamma$ and a diffeomorphism $\varphi$, one has 
\begin{equation}
\label{eq:projector_covariance}
    U_\varphi\, P'_\gamma \, U_\varphi^{-1} = P'_{\varphi(\gamma)}.
\end{equation}
Now, given $\widetilde{\Psi}\in \mathcal{H}_{\gamma_0}$ with 
\begin{equation}
\label{eq:start}
   \expect{\hat{\mathcal{Q}}_{\hat{H}}}_{\widetilde{\Psi}}=0, 
\end{equation}
for some $\hat{H}$ that is covariant wrt. $\text{GS}_{\gamma_0}$. Of course \eqref{eq:start} is equivalent to 
\begin{equation}
\label{eq:startt}
    \hat{H}^\dagger_v\widetilde{\Psi}=0 \quad \text{ for all }v \in V(\gamma).  
\end{equation}
First, note that because $\hat{H}$ is covariant with respect to GS$_{\gamma_0}$, we can group average and still have \eqref{eq:start}:
\begin{equation}
    \hat{H}^\dagger_v\   \sum_{g \in \text{GS}_{\gamma_0}}  U_g \widetilde{\Psi}
    = \sum_{g \in \text{GS}_{\gamma_0}}  (U_{g^{-1}} \hat{H}_v U_{g^{-1}}^\dagger)^\dagger \widetilde{\Psi} 
    =\sum_{g \in \text{GS}_{\gamma_0}}   \hat{H}_{g^{-1}v}^\dagger \widetilde{\Psi} = 0. 
\end{equation}
From now on, $\widetilde{\Psi}$ will be assumed invariant under GS$_{\gamma_0}$. 
\newparagraph
We decompose $\widetilde{\Psi}$ according to
\begin{equation}
    \widetilde{\Psi} =\sum_\alpha \widetilde{\Psi}_\alpha, \qquad  \widetilde{\Psi}_\alpha = P'_\alpha \widetilde{\Psi}  \in \mathcal{H}'_\alpha, \qquad \alpha \subseteq \gamma_0. 
\end{equation}
Because of the assumed invariance of $\widetilde{\Psi}$ and \eqref{eq:projector_covariance} we have for any $g \in$ GS$_{\gamma_0}$ 
\begin{equation}
    U_g \widetilde{\Psi}_\alpha = \widetilde{\Psi}_{g (\alpha)}.
\end{equation}
Moreover, because we assume \eqref{eq:gs_map} is onto, for any g' in GS$_\alpha$ there is $g \in$ GS$_{\gamma_0}$ such that   
\begin{equation}
\label{eq:inv}
    U_{g'} \widetilde{\Psi}_\alpha = U_{g'} \widetilde{\Psi}_\alpha = \widetilde{\Psi}_{g'(\alpha)} = \widetilde{\Psi}_{\alpha}.
\end{equation}
Hence under the assumptions we made, $\widetilde{\Psi}_\alpha$ is invariant under  GS$_\alpha$. 
\newparagraph
Now we can define the diffeomorphism invariant state 
\begin{equation}
    \lvert\Psi) = \sum_{\alpha\subseteq\gamma_0}\, \frac{1}{n(\alpha, \gamma_0)} \, \vert\eta (\widetilde{\Psi}_\alpha))\qquad \in \mathcal{H}_{\text{diff}}. 
\end{equation}
Then, for a cylindrical function $f_{\gamma_0}$, consider
\begin{align}
    (\Psi| \hat{H}_v f_{\gamma_0}\rangle 
    &= \sum_{\alpha\subseteq\gamma_0} \sum_{[\varphi]\in \sfrac{\text{Diff}}{\text{Diff}_\alpha}} \frac{1}{n(\alpha, \gamma_0)} \frac{1}{|\text{GS}_{\alpha}|}\sum_{g\in \text{GS}_{\alpha}} \scpr{U_\varphi U_g \widetilde{\Psi}_\alpha}{\hat{H}_v f_{\gamma_0}}\\
    &= \sum_{\alpha\subseteq\gamma_0} \sum_{[\varphi]\in \sfrac{\text{Diff}}{\text{Diff}_\alpha}} \frac{1}{n(\alpha, \gamma_0)} \scpr{U_\varphi \widetilde{\Psi}_\alpha}{\hat{H}_v f_{\gamma_0}},
\end{align}
where in the second line we have used \eqref{eq:inv}. 
Note that the sum over $\varphi\in \sfrac{\text{Diff}}{\text{Diff}_\alpha}$ will only have non-zero contributions when $\varphi$ maps $\alpha$ onto a different subgraph of $\gamma_0$. We have assumed that if this is the case, $\varphi$ can be identified with a graph symmetry $g_{[\varphi]} \in \text{GS}_{\gamma_0}$, see \eqref{eq:gs_inclusion}, so 
\begin{align}
    (\Psi| \hat{H}_v f_{\gamma_0}\rangle 
    &= \sum_{\alpha\subseteq\gamma_0} \sum_{[\varphi]\in \sfrac{\text{Diff}}{\text{Diff}_\alpha}} \frac{1}{n(\alpha, \gamma_0)} \scpr{U_{g_{[\varphi]}} \widetilde{\Psi}_\alpha}{\hat{H}_v f_{\gamma_0}}\\
    &= \sum_{\alpha\subseteq\gamma_0} \sum_{[\varphi]\in \sfrac{\text{Diff}}{\text{Diff}_\alpha}} \frac{1}{n(\alpha, \gamma_0)} \scpr{ \widetilde{\Psi}_{\varphi(\alpha)}}{\hat{H}_v f_{\gamma_0}}\\ 
    &= \sum_{[\alpha]:\alpha \subseteq\gamma_0 }
    \sum_{\alpha'\in [\alpha]}
    \sum_{\alpha''\in [\alpha]}
    \frac{1}{n(\alpha, \gamma_0)} \scpr{ \widetilde{\Psi}_{\alpha''}}{\hat{H}_v f_{\gamma_0}}\\ 
    &= \sum_{[\alpha]:\alpha \subseteq\gamma_0 }
    \sum_{\alpha''\in [\alpha]}
    \scpr{ \widetilde{\Psi}_{\alpha''}}{\hat{H}_v f_{\gamma_0}}
    \label{eq:sumsplit}\\ 
    &= \sum_{\alpha \subseteq\gamma_0 }
    \scpr{ \widetilde{\Psi}_{\alpha}}{\hat{H}_v f_{\gamma_0}}\\ 
    &= \scpr{\hat{H}_v^\dagger \widetilde{\Psi}}{f_{\gamma_0}}\\
    &=0. 
\end{align}
In \eqref{eq:sumsplit} we have introduced the notation $[\alpha]$ for the set of graphs $\subseteq\gamma_0$ that are diffeomorphic to $\alpha$, and split the sum over $\alpha$ into the sum over diffeomorphism equivalence classes $[\alpha]$ and their elements $\alpha'$. Since nothing depends on $\alpha'$, there are $|[\alpha]|=n(\alpha, \gamma_0)$ equal terms, hence the $1/n$ factor cancels in the next step. Finally, \eqref{eq:startt} is used. 
\end{proof}


\section*{References}

\begin{thebibliography}{}


\bibitem{Rovelli:1997yv}
C.~Rovelli,
``Loop quantum gravity,''
Living Rev. Rel. \textbf{1}, 1 (1998)
doi:10.12942/lrr-1998-1
[arXiv:gr-qc/9710008 [gr-qc]].


\bibitem{Thiemann:2001gmi}
T.~Thiemann,
``Introduction to Modern Canonical Quantum General Relativity,''
[arXiv:gr-qc/0110034 [gr-qc]].


\bibitem{Thiemann:2007pyv}
T.~Thiemann,
``Modern Canonical Quantum General Relativity,''
Cambridge University Press, 2007,
ISBN 978-0-511-75568-2, 978-0-521-84263-1
doi:10.1017/CBO9780511755682


\bibitem{Ashtekar:2004eh}
A.~Ashtekar and J.~Lewandowski,
``Background independent quantum gravity: A Status report,''
Class. Quant. Grav. \textbf{21} (2004), R53
doi:10.1088/0264-9381/21/15/R01
[arXiv:gr-qc/0404018 [gr-qc]].


\bibitem{Ashtekar:1986yd}
A.~Ashtekar,
``New Variables for Classical and Quantum Gravity,''
Phys. Rev. Lett. \textbf{57} (1986), 2244-2247
doi:10.1103/PhysRevLett.57.2244


\bibitem{BarberoG:1994eia}
J.~F.~Barbero G.,
``Real Ashtekar variables for Lorentzian signature space times,''
Phys. Rev. D \textbf{51} (1995), 5507-5510
doi:10.1103/PhysRevD.51.5507
[arXiv:gr-qc/9410014 [gr-qc]].


\bibitem{Baez:1996aima}
J.~C.~Baez,
``Spin network states in gauge theory,''
Adv. Math. \textbf{117} (1996), 253-272
doi:10.1006/aima.1996.0012
[arXiv:gr-qc/9411007 [gr-qc]].


\bibitem{Ashtekar:1996eg}
A.~Ashtekar and J.~Lewandowski,
``Quantum theory of geometry. 1: Area operators,''
Class. Quant. Grav. \textbf{14} (1997), A55-A82
doi:10.1088/0264-9381/14/1A/006
[arXiv:gr-qc/9602046 [gr-qc]].


\bibitem{Rovelli:1995discreteness}
C.~Rovelli and L.~Smolin,
``Discreteness of area and volume in quantum gravity,''
Nucl. Phys. B \textbf{442} (1995), 593-619
[erratum: Nucl. Phys. B \textbf{456} (1995), 753-754]
doi:10.1016/0550-3213(95)00150-Q
[arXiv:gr-qc/9411005 [gr-qc]].


\bibitem{Ashtekar:1997volume}
A.~Ashtekar and J.~Lewandowski,
``Quantum theory of geometry. II: Volume operators,''
Adv. Theor. Math. Phys. \textbf{1} (1997), 388-429
doi:10.4310/ATMP.1997.v1.n2.a8
[arXiv:gr-qc/9711031 [gr-qc]].


\bibitem{Thiemann:1996aw}
T.~Thiemann,
``Quantum spin dynamics (QSD),''
Class. Quant. Grav. \textbf{15} (1998), 839-873
doi:10.1088/0264-9381/15/4/011
[arXiv:gr-qc/9606089 [gr-qc]].


\bibitem{Thiemann:1996av}
T.~Thiemann,
``Quantum spin dynamics (qsd). 2.,''
Class. Quant. Grav. \textbf{15}, 875-905 (1998)
doi:10.1088/0264-9381/15/4/012
[arXiv:gr-qc/9606090 [gr-qc]].

\bibitem{Thiemann:1997rv}
T.~Thiemann,
``QSD 3: Quantum constraint algebra and physical scalar product in quantum general relativity,''
Class. Quant. Grav. \textbf{15} (1998), 1207-1247
doi:10.1088/0264-9381/15/5/010
[arXiv:gr-qc/9705017 [gr-qc]].

\bibitem{Thiemann:1997ru}
T.~Thiemann,
``QSD 4: (2+1) Euclidean quantum gravity as a model to test (3+1) Lorentzian quantum gravity,''
Class. Quant. Grav. \textbf{15} (1998), 1249-1280
doi:10.1088/0264-9381/15/5/011
[arXiv:gr-qc/9705018 [gr-qc]].

\bibitem{Reisenberger:1996pu}
M.~P.~Reisenberger and C.~Rovelli,
``'Sum over surfaces' form of loop quantum gravity,''
Phys. Rev. D \textbf{56}, 3490-3508 (1997)
doi:10.1103/PhysRevD.56.3490
[arXiv:gr-qc/9612035 [gr-qc]].

\bibitem{Lewandowski:2014hza}
J.~Lewandowski and H.~Sahlmann,
``Symmetric scalar constraint for loop quantum gravity,''
Phys. Rev. D \textbf{91} (2015) no.4, 044022
doi:10.1103/PhysRevD.91.044022
[arXiv:1410.5276 [gr-qc]].

\bibitem{Varadarajan:2022dgg}
M.~Varadarajan,
``Anomaly free quantum dynamics for Euclidean LQG,''
doi:10.48550/arXiv.2205.10779
[arXiv:2205.10779 [gr-qc]].

\bibitem{Varadarajan:2021zrk}
M.~Varadarajan,
``Euclidean LQG Dynamics: An Electric Shift in Perspective,''
Class. Quant. Grav. \textbf{38} (2021) no.13, 135020
doi:10.1088/1361-6382/abfc2d
[arXiv:2101.03115 [gr-qc]].

\bibitem{Ashtekar:2021shortreview}
A.~Ashtekar and E.~Bianchi,
``A short review of loop quantum gravity,''
Rep. Prog. Phys. \textbf{84} (2021) no.4, 042001
doi:10.1088/1361-6633/abed91
[arXiv:2104.04394 [gr-qc]].


\bibitem{Guedes:2024zbu}
T.~L.~M.~Guedes, G.~A.~Mena Marug{\'a}n, M.~M{\"u}ller and F.~Vidotto,
``Taming Thiemann{\textquoteright}s Hamiltonian constraint in canonical loop quantum gravity: Reversibility, eigenstates, and graph-change analysis,''
Phys. Rev. D \textbf{112} (2025) no.2, 026024
doi:10.1103/hjdk-kdhk
[arXiv:2412.20272 [gr-qc]].


\bibitem{Guedes:2024duc}
T.~L.~M.~Guedes, G.~A.~Mena Marug{\'a}n, F.~Vidotto and M.~M{\"u}ller,
``Computing the Graph-Changing Dynamics of Loop Quantum Gravity,''
Universe \textbf{11} (2025) no.12, 387
doi:10.3390/universe11120387
[arXiv:2412.20257 [gr-qc]].


\bibitem{Assanioussi:2017tql}
M.~Assanioussi, J.~Lewandowski and I.~M{\"a}kinen,
``Time evolution in deparametrized models of loop quantum gravity,''
Phys. Rev. D \textbf{96} (2017) no.2, 024043
doi:10.1103/PhysRevD.96.024043
[arXiv:1702.01688 [gr-qc]].

\bibitem{Kisielowski:2022wvk}
M.~Kisielowski,
``Bouncing Universe in loop quantum gravity: full theory calculation,''
Class. Quant. Grav. \textbf{40} (2023) no.19, 195025
doi:10.1088/1361-6382/acf271
[arXiv:2211.04440 [gr-qc]].

\bibitem{Makinen:2026rof}
I.~M{\"a}kinen,
``Time evolution of semiclassical states in the one-vertex model of quantum-reduced loop gravity,''
[arXiv:2604.00999 [gr-qc]].


\bibitem{Sahlmann:2024pba}
H.~Sahlmann and W.~Sherif,
``Towards quantum gravity with neural networks: solving the quantum Hamilton constraint of U(1) BF theory,''
Class. Quant. Grav. \textbf{41} (2024) no.22, 225014
doi:10.1088/1361-6382/ad84af
[arXiv:2402.10622 [gr-qc]].


\bibitem{Carleo:2016svm}
G.~Carleo and M.~Troyer,
``Solving the quantum many-body problem with artificial neural networks,''
Science \textbf{355} (2017) no.6325, 602-606
doi:10.1126/science.aag2302
[arXiv:1606.02318 [cond-mat.dis-nn]].


\bibitem{Sahlmann:2024kat}
H.~Sahlmann and W.~Sherif,
``Towards quantum gravity with neural networks: solving quantum Hamilton constraints of 3d Euclidean gravity in the weak coupling limit,''
Class. Quant. Grav. \textbf{41} (2024) no.21, 215006
doi:10.1088/1361-6382/ad7c14
[arXiv:2405.00661 [gr-qc]].


\bibitem{Smolin:1992wj}
L.~Smolin,
``The G(Newton) ---{\ensuremath{>}} 0 limit of Euclidean quantum gravity,''
Class. Quant. Grav. \textbf{9} (1992), 883-894
doi:10.1088/0264-9381/9/4/007
[arXiv:hep-th/9202076 [hep-th]].


\bibitem{Ashtekar:1994mh}
A.~Ashtekar and J.~Lewandowski,
``Projective techniques and functional integration for gauge theories,''
J. Math. Phys. \textbf{36} (1995), 2170-2191
doi:10.1063/1.531037
[arXiv:gr-qc/9411046 [gr-qc]].


\bibitem{Ashtekar:1994wa}
A.~Ashtekar and J.~Lewandowski,
``Differential geometry on the space of connections via graphs and projective limits,''
J. Geom. Phys. \textbf{17} (1995), 191-230
doi:10.1016/0393-0440(95)00028-G
[arXiv:hep-th/9412073 [hep-th]].


\bibitem{Thiemann:2006phoenix}
T.~Thiemann,
``The Phoenix Project: Master Constraint Programme for Loop Quantum Gravity,''
Class. Quant. Grav. \textbf{23} (2006), 2211-2247
doi:10.1088/0264-9381/23/7/002
[arXiv:gr-qc/0305080 [gr-qc]].


\bibitem{Dittrich:2006master}
B.~Dittrich and T.~Thiemann,
``Testing the Master Constraint Programme for Loop Quantum Gravity: I. General Framework,''
Class. Quant. Grav. \textbf{23} (2006), 1025-1065
doi:10.1088/0264-9381/23/4/001
[arXiv:gr-qc/0411138 [gr-qc]].


\bibitem{Ashtekar:1995diffeo}
A.~Ashtekar, J.~Lewandowski, D.~Marolf, J.~M.~Mour\~ao and T.~Thiemann,
``Quantization of diffeomorphism invariant theories of connections with local degrees of freedom,''
J. Math. Phys. \textbf{36} (1995), 6456-6493
doi:10.1063/1.531252
[arXiv:gr-qc/9504018 [gr-qc]].



\bibitem{Bakhoda:2020fiy}
S.~Bakhoda and T.~Thiemann,
``Covariant origin of the U(1)$^{3}$ model for Euclidean quantum gravity,''
Class. Quant. Grav. \textbf{39} (2022) no.2, 025006
doi:10.1088/1361-6382/ac37a4
[arXiv:2011.00031 [gr-qc]].

\bibitem{Bakhoda:2022rut}
S.~Bakhoda,
``The $U(1)^3$ Model of Euclidean Quantum Gravity'',
Doctoral thesis,
Friedrich-Alexander-Universit\"at Erlangen--N\"urnberg (FAU),
2022, https://open.fau.de/handle/openfau/19853.


\bibitem{Long:2022wcLQGArea}
G.~Long and Y.~Ma,
``Effective dynamics of weak coupling loop quantum gravity,''
Phys. Rev. D \textbf{105} (2022), 044043
doi:10.1103/PhysRevD.105.044043
[arXiv:2111.11844 [gr-qc]].


\bibitem{Bakhoda:2025U1cubedVolume}
S.~Bakhoda and Y.~Ma,
``Geometrical Quantum Time in the U(1)$^3$ Model of Euclidean Quantum Gravity,''
Commun. Theor. Phys. \textbf{77} (2025), 055401
doi:10.1088/1572-9494/ad972b
[arXiv:2411.19435 [gr-qc]].


\bibitem{Sherif:2025hfl}
W.~Sherif,
``Simultaneous approximation of multiple degenerate states using a single neural network quantum state,''
Mach. Learn. Sci. Tech. \textbf{7} (2026) no.1, 015029
doi:10.1088/2632-2153/ae3e39
[arXiv:2509.02658 [quant-ph]].


\bibitem{Dittrich:2014wpa}
B.~Dittrich and M.~Geiller,
``A new vacuum for Loop Quantum Gravity,''
Class. Quant. Grav. \textbf{32} (2015) no.11, 112001
doi:10.1088/0264-9381/32/11/112001
[arXiv:1401.6441 [gr-qc]].


\bibitem{Sherif:2026neuralqx}
W.~Sherif,
``neuraLQX: a high-performance simulations toolkit for loop quantum gravity,''
version 1.1.1 (2026)
[software, available at: http://github.com/waleed-sh/neuraLQX].


\bibitem{Carleo:2019netket}
G.~Carleo, K.~Choo, D.~Hofmann, J.~E.~T.~Smith, T.~Westerhout, F.~Alet, E.~J.~Davis, S.~Efthymiou, I.~Glasser, S.-H.~Lin, M.~Mauri, G.~Mazzola, C.~B.~Pereira and F.~Vicentini,
``NetKet: A machine learning toolkit for many-body quantum systems,''
SoftwareX \textbf{10} (2019), 100311
doi:10.1016/j.softx.2019.100311.


\bibitem{Vicentini:2022netket3}
F.~Vicentini, D.~Hofmann, A.~Szabó, D.~Wu, C.~Roth, C.~Giuliani, G.~Pescia, J.~Nys, V.~Vargas-Calderón, N.~Astrakhantsev and G.~Carleo,
``NetKet 3: Machine Learning Toolbox for Many-Body Quantum Systems,''
SciPost Phys. Codebases (2022), 7
doi:10.21468/SciPostPhysCodeb.7.


\bibitem{Bradbury:2018jax}
J.~Bradbury, R.~Frostig, P.~Hawkins, M.~J.~Johnson, C.~Leary, D.~Maclaurin, G.~Necula, A.~Paszke, J.~VanderPlas, S.~Wanderman-Milne and Q.~Zhang,
``JAX: Composable Transformations of Python + NumPy Programs,''
(2018).


\bibitem{Heek:2024flax}
J.~Heek, A.~Levskaya, A.~Oliver, M.~Ritter, B.~Rondepierre, A.~Steiner and M.~van~Zee,
``Flax: A Neural Network Library and Ecosystem for JAX,''
(2024).


\bibitem{Haefner:2021mpi4jax}
D.~Häfner and F.~Vicentini,
``mpi4jax: Zero-copy MPI communication of JAX arrays,''
J. Open Source Softw. \textbf{6} (2021) no.65, 3419
doi:10.21105/joss.03419.

\end{thebibliography}
\end{document}